\documentclass{article}
\usepackage[utf8]{inputenc}
\usepackage[dvips,letterpaper,margin=1.0in]{geometry}
\usepackage{sty}
\usepackage{hyperref}
\usepackage{framed}

\def\Cx{\mathbb{C}}

\def\Rl{\mathbb{R}}

\def\K{\mathbb{K}}

\def\bb1{\mathbbm{1}}

\def\Span{\mathrm{span}}
\def\rowspan{\mathrm{rowspan}}
\def\colspan{\mathrm{colspan}}
\def\Sym{\mathrm{Sym}}
\def\diag{\operatorname{diag}}

\def\Id{\mathrm{Id}}
\def\pMRA{\mathrm{pMRA}}
\def\dMRA{\mathrm{dMRA}}

\def\U{\mathcal U}

\def\mchoose#1#2{\ensuremath{\left(\kern-.3em\left(\genfrac{}{}{0pt}{}{#1}{#2}\right)\kern-.3em\right)}}

\newcommand{\dihgroup}[1]{\mathcal D_{#1}}
\newcommand{\cycgroup}[1]{\mathcal C_{#1}}
\newcommand{\projOp}{\Pi}
\def\shiftel{R}
\def\reflel{J}
\def\fshiftel{\widehat{R}}
\def\freflel{\widehat{J}}
\def\multi{\operatorname{multi}}
\newcommand{\fprojOp}{\widehat{\Pi}}
\def\bb1{\mathbbm{1}}
\def\sgn{\operatorname{sgn}}
\newcommand{\mset}[1]{\Lbag#1\Rbag}

\title{Beyond Frequency Marching: Orbit Recovery in Dihedral and Projected Multireference Alignment}

\usepackage{authblk}
\author{Tait Weicht\thanks{Email: \textit{tjweicht@ucdavis.edu}. Partially supported by NSF awards CCF-2338091 and DMS-2012764.} }
\author{Alexander S.\ Wein\thanks{Email: \textit{aswein@ucdavis.edu}. Partially supported by a Sloan Research Fellowship and NSF CAREER Award CCF-2338091.}}

\affil{UC Davis}

\date{}

\begin{document}

\maketitle
\begin{abstract}
\emph{Multireference alignment (MRA)} is the task of recovering a hidden ``signal'' vector, given many noisy copies that have been cyclically shifted by unknown offsets. This task belongs to the class of \emph{orbit recovery} problems, in which the observed samples are affected by some group action. These problems have a variety of practical motivations, including the reconstruction of 3-dimensional molecular structure from \emph{cryogenic electron microscopy (cryo-EM)} images. We consider two variants of MRA: \emph{dihedral MRA}, where the cyclic group is replaced by the dihedral group, allowing for reversals of the vector in addition to shifts; and \emph{projected MRA}, where the observations are passed through a projection operator akin to the tomographic projection present in cryo-EM. We apply the method of moments and aim to recover the signal from the third moment tensor of the samples. This inverse problem is well understood for basic MRA, but for the variants we consider there was previously no polynomial-time algorithm known to succeed for generic signals. We give an algorithm for both of these variants. Our method requires the signal length to be a power of two, and recursively subdivides the problem into smaller problems of half the size. The algorithm's success for generic signals is proven, conditional on a conjecture about the rank of a certain symbolic matrix of polynomials. For any given problem size, this conjecture can be verified on a computer.
\end{abstract}
\newpage
\tableofcontents

\newpage
\section{Introduction}

In the \emph{multireference alignment (MRA)} problem~\cite{bandeira_multireference_2014,bandeira_optimal_2020,perry_sample_2019,abbe_sample_2017}, we are given samples of the form
\begin{equation}
    \label{eq:MRA-model}
    y_i = \shiftel_i \cdot x + \sigma\xi_i \in \RR^d, \quad i = 1,\ldots,N
\end{equation}
where
\begin{itemize}
    \item $x\in\Rl^d$ is a ``signal'' we aim to recover;
    \item $\shiftel_i$ is an element of the cyclic group $\cycgroup{d} = \{R_0,R_1,\ldots,R_{d-1}\}$ drawn independently and uniformly at random;
    \item the cyclic group acts on $\RR^d$ via cyclic shifts, so $R_\ell \cdot x$ denotes the vector $x$ with entries cycled to the right by an offset of $\ell$;
    \item $\xi_i \sim \mathcal{N}(0,I)$ is Gaussian noise;
    \item $\sigma > 0$ is the noise level, which determines the \emph{signal-to-noise ratio} $\mathrm{SNR} := 1/\sigma^2$.
\end{itemize}
Given samples $\{y_i\}_{i=1}^N$, the goal is to recover $x$. However, $x$ is only identifiable up to cyclic shift, that is, we aim to recover the \emph{orbit} of $x$ under the group action of cyclic shifts, $\{R_\ell\cdot x : \ell=0,1,\dots,d-1\}$.

MRA is one canonical example of a task within the class of \emph{orbit recovery} problems~\cite{bandeira_estimation_2023}, where the observed samples are affected by some group action.
These problems have a variety of practical motivations, including the reconstruction of 3-dimensional molecular structure from \emph{cryogenic electron microscopy (cryo-EM)} images.
In cryo-EM, individual particles are frozen in a thin layer of ice through which images are taken by a scanning electron microscope.
Particle orientations are random relative to the viewing plane.
While random rotations may be considered nuisance parameters due to these rotations, we can only hope to recover a 3D structure up to a 3D rotation i.e., the structure's orbit under the action of $SO(3)$.
Two added complexities are that the images are tomographic projections of a 3D electro-static potential function onto a 2D imaging plane, and low amounts of radiation must be used to avoid damage and motion blur from melting ice, leading to noisy images.
The tomographic projection results in additional ambiguity since particles of different chirality will have the same distribution of projections.
With projection, one can only hope to recover structure up to handed-ness, or equivalently, up to the action of $O(3)$.
See~\cite{frank_generalized_2016} for a historical narrative of cryo-EM's experimental development and~\cite{singer_mathematics_2019} for a more detailed discussion of the mathematical details we mention here.

Collecting useful experimental cryo-EM data requires considerable expertise, but even after successful data collection, the mathematics of structure determination as a non-linear inverse problem is highly non-trivial.
Orbit recovery in the mathematical models like that those in~\cite{singer_mathematics_2019} is far from solved.
The MRA problem provides a toy model which captures some key features of the cryo-EM problem, including a phase transition first observed empirically by Sigworth~\cite{sigworth_maximum-likelihood_1998} and later proven~\cite{bandeira_optimal_2020,perry_sample_2019,abbe_sample_2017} in which the reconstruction error scales as $1/\mathrm{SNR}^3$ in the high-noise regime.
However, the MRA model simplifies several aspects present in cryo-EM since the action of cyclic shifts is diagonalizable while the action of $SO(3)$ is not.
In addition, in cryo-EM the tomographic projection operation has a non-trivial kernel so that each sample only contains partial information about the object even when $\sigma=0$.

We study two different multi-reference alignment (MRA) variants which aim to introduce these challenges to the MRA model.
The first variant is \emph{dihedral} multi-reference alignment (dMRA) \cite{bendory_dihedral_2022, edidin_reflection-invariant_2026} where samples are drawn from 
\begin{equation}
    \label{eq:dMRA-model}
    y_i = g_i\cdot x + \sigma\xi_i, \quad i=1,\dots, N
\end{equation}
where \eqref{eq:dMRA-model} modifies \eqref{eq:MRA-model} in that $g_i$ is an element of dihedral group $\dihgroup{d} = \{R_0,\ldots,R_{d-1}, JR_0,\ldots,JR_{d-1}\}$ drawn independently and uniformly at random.
Here, $J$ acts by reversing the order of entries in the vector.
The distribution of the samples $\{y\}_{i=1}^N$ is invariant if we exchange $x$ for any $g\cdot x$ with $g\in \dihgroup{d}$ so we can only expect to recover $x$ up to its $\dihgroup{d}$-orbit.

The second variant, a version of which was introduced in \cite[Section~4.4.1]{bandeira_estimation_2023}, is \emph{projected} multi-reference alignment (pMRA)
\begin{equation}
    \label{eq:pMRA-model}
    y_i = \Pi(R_i\cdot x) + \sigma\xi_i, \quad i=1,\dots, N
\end{equation}
where we modify \eqref{eq:MRA-model} by fixing $x\in \Rl^{2d}\cap\Span([1,-1,1,\dots,-1]^\top)^\perp$ and using the projection operation $\Pi: \Rl^{2d}\to\Rl^d$ given by
\begin{equation}
    (\Pi x)[j] = x[j] + x[2d-1-j],\quad j=1,2,\dots,d.
 \end{equation}
We call the application of $\Pi$ ``projection'' since it could be considered a discretization of the tomographic projection operation.
In this context, we can think of entries of $x$ as giving the density for patches of a thin-shell object.
Projection then has the effect of summing two entries when they align on an imaging axis (see Figure \ref{fig:pMRA-model-diagram}).
\begin{figure}[h]
    \centering
    \label{fig:pMRA-model-diagram}
    \includegraphics[width=0.4\linewidth]{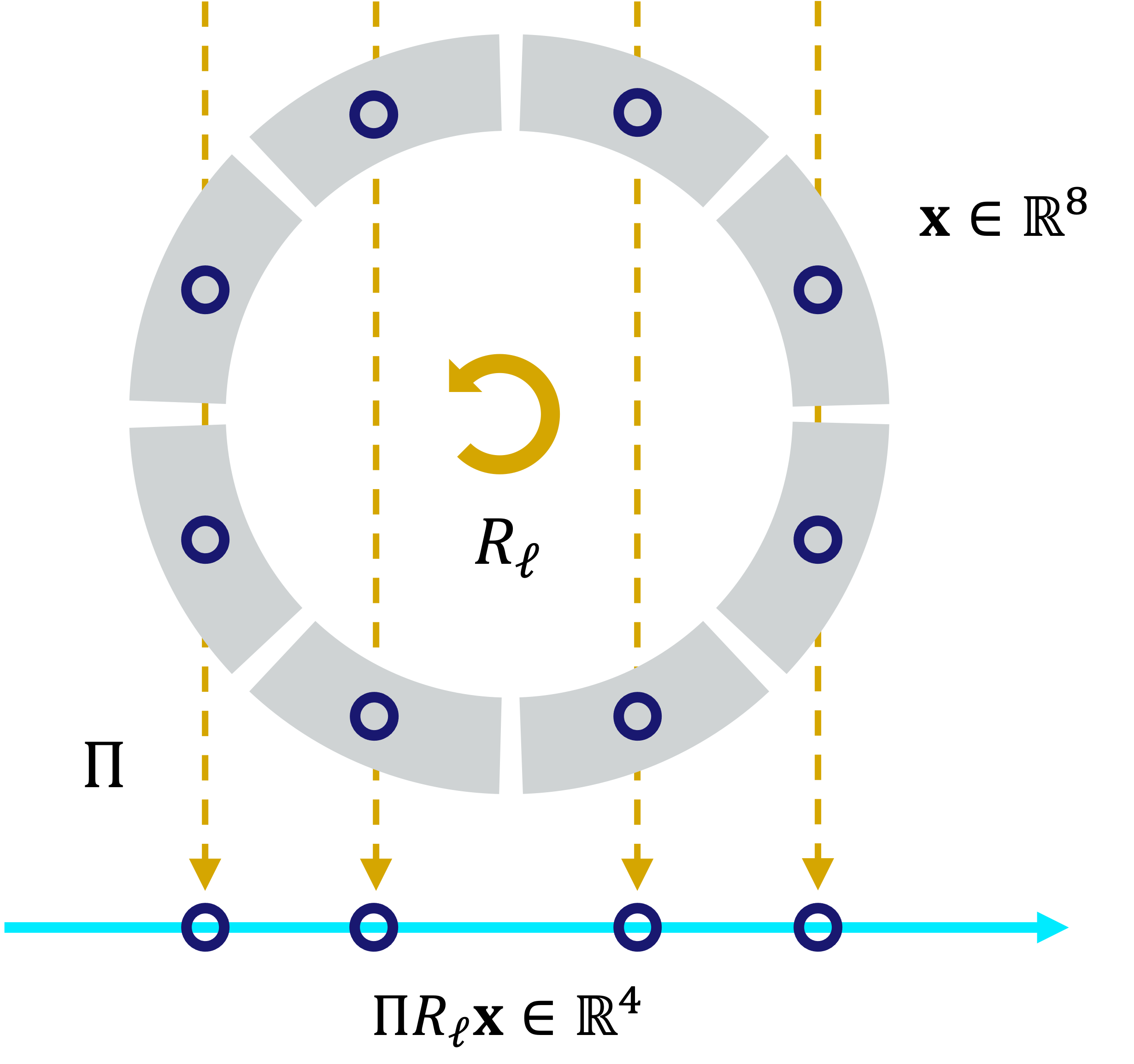}
    \caption{A visualization of the sampling procedure in the pMRA model as tomographic projection of some thin-shell 2D discretized object.}
\end{figure}
An effect of the projection operation is that we can only hope to recover $x$ up to its $\dihgroup{2d}$-orbit instead of its $\cycgroup{2d}$-orbit. 
In fact, the model is unchanged if our group were dihedral instead of cyclic.
The condition that the signal is orthogonal to $(1,-1,1,\dots,-1)^\top$ ensures the problem is well-posed since the vector $(1,-1,1,\dots,-1)^\top$ lies in the kernel of $\Pi \circ g$ for every $g\in\dihgroup{2d}$. This vector represents the highest frequency component present in the signal, so excluding this component amounts to a ``smoothness' assumption on the signal.
Projection has been a key difficulty in extending methods developed for orbit recovery for the MRA model to a model of cryo-EM~\cite{liu_algorithms_2022}.
Projection has also been shown to cause issues with methods such as maximum likelihood estimation (MLE) when the probability distribution on group actions is misspecified \cite{xu_misspecified_2025}.

\paragraph{Method of moments/invariants.}

In this work, we focus on the high noise regime where $\sigma$ is much larger than a typical entry of $x$. The strategy of \emph{synchronization} \cite{singer_three-dimensional_2011, singer_angular_2011, bandeira_non-unique_2020} uses comparisons between samples to label each with a unique group action (up to a global action), after which samples with the same label may be averaged to reduce noise.
In the high noise regime, we cannot hope to estimate the group element associated with each sample by pairwise or even $n$-way comparisons \cite{aguerrebere_fundamental_2016}.
Instead, one can aim to estimate features of the sample distribution that are invariant under the group action and then use these features to recover the signal's orbit under the group action.
We use \emph{moment tensors} of (the signal part of) the sample distribution which are one class of invariant features (capturing all the polynomial invariants).
These moments are given by
\begin{equation}
    \label{eq:invariant-tensors-with-projection}
    T^{(r)}_{G,P}(x) := \Eop_{g\sim\mathrm{Unif}(G)}[(P(g\cdot x))^{\otimes r}]
\end{equation}
where in the dMRA model $G=\dihgroup{d}$ and $P=I$, and in the pMRA model $G= \cycgroup{2d}$ and $P=\Pi$.
Whenever $P$ is omitted in the subscript, it is assumed to be the identity $P = I$. Each $T^{(r)}_{G,P}(x)$ is an order-$r$ tensor where each entry of the tensor is a homogeneous degree-$r$ $G$-invariant polynomial of $x$.
Each entry of the $r$-th order moment may be estimated to an accuracy $\delta$ with high probability using $O(\sigma^{2r}\log(d/\delta))$ samples \cite[Proposition 6.7]{bandeira_estimation_2023} using the empirical moments $\frac1N \sum_i y_i^{\otimes r}$ and some post-processing steps to remove bias caused by noise.
Since the number of samples needed grows exponentially in the order of the tensor we wish to estimate, one might hope that for small $R$ the tensors $\{T^{(r)}_{G,P}(x)\}_{r=1}^R$ \emph{separate} the $G$-orbits of $x$, i.e., contain enough information to identify the orbit of $x$.
Unfortunately, even for the original MRA model this does not hold; there exist signals $x\in\Rl^d$ that will require up to the $d$-th moment to separate orbits (see \cite{bandeira_optimal_2020} for this result in a continuous MRA model and a similar analysis in \cite{perry_sample_2019} for the discrete MRA model).
However, the signals of this type have Lebesgue measure zero.
If we are satisfied with recovering ``most'' signals, we can instead ask how many moments are needed to uniquely identify the $G$-orbit for a \emph{generic} signal, meaning we exclude some measure-zero set of ``bad'' signals (see Definition~\ref{def:generic} for the formal definition).
It has been shown in the MRA model that the first and second moment tensors are not sufficient to identify the orbit of $x$, but the third moment tensor separates the orbits of a generic signals.
As a result, the sample complexity of MRA scales like $\Theta(\sigma^6)$ in the high-noise limit (where throughout this discussion we will ignore the logarithmic dependence on other parameters such as dimension and accuracy)~\cite{bandeira_optimal_2020,perry_sample_2019,abbe_sample_2017}. Further, there is a fast algorithm to recover the orbit from the third moment tensor~\cite{perry_sample_2019}. Beyond MRA, for more general orbit recovery problems, $\Theta(\sigma^{2r})$ samples are both necessary and sufficient, information-theoretically, where $r$ is the number of moments needed to uniquely identify the signal~\cite{bandeira_estimation_2023}.

Both the models we study are strictly harder than the original MRA model in the sense that the invariant tensors in the dMRA and pMRA models have entries which are linear combination of the entries found in the MRA invariant tensors; hence, one obtains less information from these tensors.
One therefore needs at least three moments to achieve orbit recovery in both the pMRA and dMRA models.
Recently, Edidin and Katz~\cite{edidin_reflection-invariant_2026} proved that the first three moments do indeed separate generic $\dihgroup{d}$-orbits\footnote{When we say that certain moments ``separate generic orbits'' we mean that for a generic $x$, the only vectors $y$ that share the same moments as $x$ are the elements of $x$'s orbit.} so the sample complexity of dihedral MRA is $\Theta(\sigma^6)$.
However, the proof method does not result in a computationally efficient algorithm to recover the orbit.
For the pMRA model, even less was known prior to our work: Bandeira et al.\ showed that for a similar model~\cite[Problem 4.7]{bandeira_estimation_2023} --- in which projection is defined for odd-length signals --- the system of polynomials defined by the first three moment tensors will have a finite number of solutions for problems of dimension 3 to 21, but again this does not yield an efficient algorithm for recovery.

In a more general setting that includes cryo-EM, we note that Balanov, Bendory, and Edidin~\cite{balanov_group-invariant_2026} have shown that for an unknown $n$-dimensional object, the $r$-th moment of tomographic projection onto an $m$-dimension subspace determines the $r$-th moment of the un-projected object provided that $m \ge r$, leading to an algorithm that recovers certain un-projected moments from their projected counterparts. However, this does not handle the important case of $r=3$ and $m=2$ in cryo-EM, and does not seem to apply in the models we study here.

\paragraph{Our perspective.}

We work towards finding a polynomial-time algorithm that can provably recover orbits for generic signals in the dihedral and projected MRA models.
In this work we suppose that we have access to exact third order moments $T^{(3)}_{\dihgroup{d}}(x)$ or $T^{(3)}_{\cycgroup{2d},\Pi}(x)$, and we assume linear-algebraic computations can be done in exact arithmetic.
While questions of robustness and noise stability are pertinent (e.g.~\cite{liu_algorithms_2022}), we leave this direction for future work.
We insist on using only moments of order $\le 3$ (and in fact the third moment alone will be enough for our purposes) in order to achieve optimal sample complexity in the high-noise limit.

\paragraph{Known algorithms and their limitations.}
From an algebraic perspective, the entries of $T^{(r)}_{G,P}(x)$ are degree-$r$ homogeneous polynomials in the entries of the vector $x$, so orbit recovery may be achieved by solving a polynomial system of equations.
In principle, the solutions may be found by computing a Gr\"obner basis with an ordering on the monomials to allow for elimination.
In practice, this method is not efficient, and the worst-case time complexity of Gr\"obner basis computation is doubly exponential in the dimension $d$.
We briefly detail two approaches developed in prior work on the MRA model that motivate our approach.

First, \emph{frequency marching} may be used to achieve orbit recovery in the original MRA model \cite{matsuoka_phase_1984,giannakis_signal_1989} and in \emph{cryogenic electron tomography (cryo-ET)}~\cite{bandeira_estimation_2023,liu_algorithms_2022} --- i.e., ``unprojected'' cryo-EM with 3-dimensional observations rather than 2-dimensional projections --- using the first three moments of the sample distribution.
In both models, the second moment determines the magnitudes of the Fourier or Fourier--Bessel coefficients, also called the \emph{power spectrum}.
To recover phases of these coefficients, the third moment of the distribution may be used to set up a linear system of equations modulo $2\pi$ for the coefficients' phases.
The phases of coefficients may then be solved for sequentially, in a procedure called \emph{frequency marching} referring to the fact that one initializes a phase for the first Fourier coefficient (corresponding to a low frequency component of signal) then then uses relationships between phases to determine phases of subsequent Fourier coefficients (higher frequencies).
Bendory et al.~\cite{bendory_bispectrum_2018} give a version of frequency marching for the MRA problem and compare its performance to a number of other algorithms.
Liu and Moitra \cite{liu_algorithms_2022} modify the frequency marching procedure to achieve $SO(3)$-orbit recovery in cryo-ET in quasipolynomial time.

However, Edidin and Katz \cite{edidin_reflection-invariant_2026} showed that adapting the frequency marching method to dihedral MRA, one can only determine each linear equation needed to march frequencies up to an unknown sign.
The ambiguity in sign leads to a combinatorial explosion in which there are exponentially many linear systems, only one of which is guaranteed to have the true $x$ as a solution.
While in theory solving the right system recovers the orbit, the runtime of enumerating all systems is impractical.

Second, \emph{tensor decomposition} methods have been shown to achieve orbit recovery in the MRA model \cite{perry_sample_2019}.
For $x\in\Rl^d$ the third moment tensor takes the form
\[
    T^{(3)}_{\cycgroup{d}}(x) = \sum_{\ell=0}^{d-1} (R_\ell x)^{\otimes 3},
\]
which is the sum of $d$ different $d \times d \times d$ rank-1 tensors. Unlike a matrix factorization, the factorization of a tensor into its rank-1 summands (called its \emph{canonical polyadic (CP) decomposition}) does not require the component vectors to be orthogonal for a unique decomposition~\cite{kruskal_three-way_1977}.
Perry et al.~\cite{perry_sample_2019} showed that \emph{simultaneous diagonalization}, a classical tensor decomposition algorithm, may be used to factor the third moment in a robust way to recover the components $\{R_\ell x\}_{\ell=0}^{d-1}$. Crucially, simultaneous diagonalization requires the rank (number of rank-1 summands to recover) to be no larger than the dimension --- in the case of MRA, this holds with equality.
However, the third moments arising from the models we study take the form
\[
T^{(3)}_{\dihgroup{d}}(x) = \sum_{g\in \dihgroup{d}} (g\cdot x)^{\otimes 3} \qquad \text{and} \qquad T^{(3)}_{\cycgroup{2d}, \Pi}(x) = \sum_{g\in \cycgroup{2d}} (\Pi(g\cdot x))^{\otimes 3}.
\]
Each is the sum of $2d$ different $d \times d \times d$ rank-1 tensors, which is more components than simultaneous diagonalization can handle.
The best known methods for generic tensor decomposition can reach a rank of $(2-\epsilon)d$~\cite{kothari_overcomplete_2024} (see also~\cite{persu_tensors_2018,chen_overcomplete_2022,koiran_efficient_2024}), which barely falls short of our $2d$.
This means we cannot use known tensor methods ``out of the box'' and must take advantage of the special structure in our tensors to when devising an efficient algorithm.

We also note that methods for \emph{fourth-order} tensor decomposition can handle significantly more components, namely $\Theta(d^2)~\cite{de_lathauwer_fourth-order_2007,ma_polynomial-time_2016,hopkins_robust_2019,johnston_computing_2023}$. However, estimating the fourth moment requires sub-optimal sample complexity, so our focus in this work is on methods that only use the third moment and below.

\paragraph{Our contributions.}

Our main results partially resolve the open problem of orbit recovery in the dihedral and projected MRA models.
We state these here.
We first formally define the algebraic idea of generic vectors.

\begin{definition}\label{def:generic}
    A property $Q$ is said to ``hold generically'' for vectors $v\in\K^m$ (or to ``hold for generic choice of'' $v\in\K^m$) if there is a non-zero polynomial $p$ on $\K^m$ such that when $p(v)\neq 0$, $v$ has property $Q$.
\end{definition}

\begin{theorem}[Dihedral MRA]
    \label{thm:dMRA-main}
    Let $d=2^k$, for $k\ge 2$ and suppose a particular symbolic matrix $M = M(d)$ (see \eqref{eq:master-M-dMRA}) has full column-rank.
    There is a polynomial-time procedure that, for generic $x\in\Rl^{2d}$, takes as input the exact third moment $T^{(3)}_{\dihgroup{2d}}(x)$ and outputs an element in the $\dihgroup{2d}$-orbit of $x$.
\end{theorem}

\begin{theorem}[Projected MRA]
    \label{thm:pMRA-main}
    Let $d=2^k$, for $k\ge 3$, and suppose a particular symbolic matrix $M' = M'(d)$ (see \eqref{eq:master-M-pMRA}) has full column-rank. There is a polynomial-time procedure that, for generic $x\in\Rl^{2d}\cap \Span(1,-1,\dots,1,-1)^\perp$, takes as input the exact third moment $T^{(3)}_{\cycgroup{2d},\Pi}(x)$ and outputs an element in the $\dihgroup{2d}$-orbit of $x$.
\end{theorem}

For any specific dimension $d$, one can use a computer to verify that the symbolic matrices $M(d)$ and $M'(d)$ have full column-rank, implying that the algorithm works for generic inputs of that size. 
This check only needs to be done once and ensures our algorithm then works on generic inputs.
Each of these matrices has more rows than columns under the conditions given, and both classes of matrices have an aspect ratio tending toward $8:3$ as $d$ goes to infinity.
We ran numerical tests giving strong evidence that these matrices have full column-rank for the dimensions $d=4,8$ and $16$.
Further, we provide an implementation of the algorithms in Python and demonstrate their efficacy by recovering signals in $\Rl^{64}$ and $\Rl^{64}\cap\Span(1,-1,\dots, 1,-1)^\perp$ respectively.
In both cases the algorithms recover the signal up to cyclic shift and reversal with high accuracy from the moments described.
These numerical tests and further details may be found in Section~\ref{sec:numerical-experiments}.

The dimension $2d$ must be a power of two since our method uses recursion to extract and solve orbit recovery sub-problems.
We first recover the even Fourier coefficients of $x$ by recursively running the algorithm on a sub-tensor.
We can then recover the odd Fourier coefficients given the even ones using a different subtensor.
While this method is reminiscent of frequency marching, we solve for multiple coefficients simultaneously by setting up a \emph{variety-constrained linear system}, a method that synthesizes the work of Bouss\'e et al.~\cite{bousse_linear_2018} on CPD-constrained linear systems and Johnston et al.~\cite{johnston_computing_2023} on efficiently finding the intersection of a subspace and a conic variety.
Unlike frequency marching, our methodology is able to handle the projection operation in the projected MRA model.

\paragraph{Further discussion and related work.}
A significant body of literature studying orbit recovery exists.
We highlight a couple of key problems not addressed in this work that are particularly interesting practically and theoretically, especially with respect to cryo-EM.

In cryo-EM, often we cannot guarantee that all the particles in a sample are identical.
Either there may be gross differences, e.g., a protein-complex missing a sub-unit, or conformational differences where proteins flex or bend slightly.
Both situations present a form of signal \emph{heterogeneity}.
In the simplest form, one assumes a mixture of $K$ different signals.
Including heterogeneity in an MRA model as in~\cite{perry_sample_2019}, we assume that instead of one distinct signal $x$ there are $K$ different signals that each appear with some probability.
Boumal et al.\ \cite{boumal_heterogeneous_2018} show that fitting the first three moments, heterogeneity $K>\lceil d/6 \rceil$ causes the problem to be ill-posed but a non-convex optimization approach seems to handle up to (approximately) $K\leq\sqrt{d}$ different signals, when the signals are generated randomly (a more restrictive setting than that of generic signals).
Continuous heterogeneity is a considerably different challenge to handle; see \cite{toader_methods_2023} for the current state of the art. Heterogeneity is related to the problem of \emph{super-resolution} in MRA~\cite{bendory_super-resolution_2022}.

Another reality is that due to boundary effects in the ice layer, particle orientation in cryo-EM may not be uniformly distributed~\cite{sharon_method_2020}.
One can make a similar assumption that some cyclic shifts are more or less likely in the MRA model.
If the distribution of group elements is unknown, it must also be estimated.
Interestingly, in some orbit recovery models, non-uniformity of the group elements allows for generic recovery using lower order moments than in the uniform case.
For example, Abbe et al.~\cite{abbe_multireference_2019} give an exact method for recovering the signal and non-uniform distribution of cyclic shifts for the MRA model using the \emph{second} moment rather than the third.
Bendory et al.~\cite{bendory_dihedral_2022} give an information-theoretic proof that the second moment suffices in the dihedral MRA model with a generic distribution on group elements.
Similarly for cryo-EM, results by Sharon et al.~\cite{sharon_method_2020} show that the second moment suffices for accurate estimation, and determine bandlimit conditions on molecular structure and the distribution over viewing angles which ensure well-posedness.

There is also a significant body of research on non-convex methods for orbit recovery, e.g.~\cite{boumal_heterogeneous_2018, fan_likelihood_2023, fan_maximum_2024, xu_misspecified_2025}.
Notably, fitting moments by least squares is demonstrated by Edidin and Katz~\cite{edidin_reflection-invariant_2026} as a practical method for recovery in dihedral MRA, albeit one without guarantees of correctness.
In MRA with a non-uniform distribution, Abas et al.~\cite{abas_generalized_2022} develop a generalized method of moments which determines how to re-weight a least-squares fitting of moments to be asymptotically optimal and demonstrate this improves recovery of the signal and distribution as parameters.
Other directions explore maximum likelihood estimation (MLE) as a recovery technique.
Xu et al.~\cite{xu_misspecified_2025} observe in cryo-EM that misspecifying the probability distribution (e.g., assuming a uniform distribution when viewing angles are non-uniform) does not have a deleterious effect on the likelihood landscape in \emph{unprojected} cryo-EM (cryo-ET) but does result in incorrect estimation when projection is present.
Due to this effect, they propose an expectation maximization (EM) algorithm for estimating the signal and viewing-angle distribution jointly.
To our knowledge, none of the methods based on non-convex optimization provide a rigorous proof that the global optimum can be found in polynomial time. As a result, this line of work deviates from our primary objective of finding algorithms with provable guarantees.

\paragraph{Concurrent work.}
Finally, we point out the concurrent and independent work~\cite{proj-mra}, which also studies a variation of the projected MRA problem (with odd-length signals) and gives an algorithm for recovery from the third moment. Their algorithmic approach is rather different from ours: essentially they have adapted the frequency marching method in a way that avoids the combinatorial explosion we mentioned above. Their algorithm appears to also work for dihedral MRA.

\paragraph{Organization.}
In Section~\ref{sec:preliminaries} we define a number of constructions --- in particular, a symmetrized Kronecker product to produce linear maps between symmetric lifts of vector spaces.
In Section~\ref{sec:dMRA-and-pMRA}, we define the dMRA and pMRA models in the Fourier basis and give explicit expressions for the invariants that result.
Section~\ref{sec:method-for-recovery} covers the steps of our algorithm in detail.
In Section~\ref{sec:analysis-of-aided-recovery}, we give sufficient conditions for the recursive step to succeed and give evidence that these conditions hold for generic signals.
Since our algorithm is recursive, we discuss how orbit recovery may be achieved for the base case in Section~\ref{sec:analysis-base-case}.
In Section~\ref{sec:efficient-implementation} we give an efficient implementation of the recursive step and details for an implementation in Python of the full algorithm.
Finally, in Section~\ref{sec:future-directions} we discuss future directions.

\section{Preliminaries}
\label{sec:preliminaries}
In the following, we use $\K$ to denote the fields $\Rl$ or $\Cx$.
We will use the notation $[n]$ to denote the set of whole numbers from 1 to $n$,
\[
    [n] = \{1,2,\dots,n-1, n\},
\] and the set $[n)$ to denote whole numbers from 0 to $n-1$,
\[
    [n) = \{0,1,2,\dots,n-1\}.
\]
We will primarily use elements $[d]$ to index entries of vectors in the standard basis and elements of $[d)$ to index Fourier coefficients of a vector (entries of a vector in the Fourier basis).

\paragraph{Some combinatorial constructions.}
We will disambiguate between sets and multisets reserving the traditional curly braces for sets, e.g.\ $\{1,2,3\}$ is a set, and using the ``bag'' notation for multisets, e.g. $\mset{1,2,2}$ is a multiset.
We will use the notation $\mchoose{n}{k}$ to mean ``$n$ multi-choose $k$'' where \[\mchoose{n}{k} = \binom{n+k-1}{k}\] and counts the number of $k$-element multisets of $[n]$.
For a set $S$, use $\mchoose{S}{k}$ to denote the set of $k$-element multi-sets of $S$.
Finally, for a $k$-element multiset $K$, we will use $\multi(K)$ to indicate the multiset of multiplicities of the elements in $K$.
In conjunction with the notation that
\[
\binom{n}{m_1,m_2,\dots,m_k} = \frac{n!}{m_1!m_2!\cdots m_k!}
\]
we may count the number of distinct permutations of a multiset $K$ by $\binom{|K|}{\multi(K)}$.
It will also at times be useful to impose a particular total ordering on the set of multisets for which we use the Dershowitz--Manna (DM) ordering.
\begin{definition}[Dershowitz--Manna (DM) ordering]
    \label{def:Dershowitz-Manna-order}
    For two multisets $M$ and $M'$ of $[n]$, we say that $M' \prec M$ when there are two multisets $X,Y$ such that
    \begin{enumerate}[(i)]
        \item $\emptyset \neq X\subseteq M$,
        \item $ M' = (M\setminus X)\cup Y$,
        \item and for every $x\in X$ there is a $y\in Y$ so that $x > y$.
    \end{enumerate}
\end{definition}

\subsection{Tensors and Indexing}

\paragraph{Tensors and tensor spaces.}
We denote a space of $k$-order tensors with shape $(n_1, n_2,\dots, n_k)$ as
\[
\K^{n_1} \otimes \K^{n_2} \otimes \cdots \otimes \K^{n_k} = \Span(u^1\otimes u^2 \otimes \dots \otimes u^k : u^1\in\K^{n_1}, u^2\in\K^{n_2},\dots, u^k\in\K^{n_k}).
\]
The short-hand
\[
    (\K^n)^{\otimes k} := \underbrace{\K^n \otimes \K^n \otimes \cdots \otimes \K^n}_{k \text{ times}}
\]
indicates a repeated tensor product.
\paragraph{Indexing and reshaping.} 
We use hard brackets to indicate a specific entry of a vector or matrix, e.g.\ for a vector $x\in\K^n$, $x[j]$ is the $j$-th entry of $x$ and for a matrix $A \in \K^{m\times n}$, $A[j,k]$ is the entry in the $j$-th row and $k$-th column.
We also use the placeholder ``$\::\:$'' to indicate full range of indices so that e.g.\ for a matrix $A$, $A[j,:]$ is the $j$-th row as a vector and $A[:,k]$ is the $k$-th column.

If a vector or matrix has extra structure we may index into it using tuples to emphasize this structure.
Flattening (or generally, reshaping) using an isomorphism $\K^m\otimes \K^n \cong \K^{mn}$ may be indicated using our indexing system e.g.\ if $A\in \K^m\otimes \K^n$ and we wish to define a flattening to $v\in \K^{mn}$ then we use
\[
    v[(i,j)] := A[i,j]
\]
where a lexicographical ordering on $(i,j)$ gives a bijection between standard basis vectors of $\K^{mn}$ and $\K^m\otimes \K^n$. In general, the tensor space $\K^{n_1} \otimes \K^{n_2} \otimes \cdots \otimes \K^{n_k}$ is isomorphic to the vector space $\K^{n_1n_2\cdots n_k}$. 
The Kronecker product of vectors $u^1\in \K^{n_1}$, $u^2\in \K^{n_2}, \ldots , u^k\in \K^{n_k}$ is a vector $u^1\otimes u^2\otimes \cdots\otimes u^k \in \K^{n_1n_2\cdots n_k}$ defined entry-wise by
\[
    (u^1\otimes u^2\otimes \cdots\otimes u^k)[(j_1,j_2,\dots j_k)] = u^1[j_1]u^2[j_2]\cdots u^k[j_k],
\]
where our indexing by $(j_1,j_2,\dots j_k)\in[n_1]\times [n_2]\times\cdots \times [n_k]$ implicitly uses a fixed bijection with an element of $[n_1n_2\cdots n_k]$ e.g.\ by using the lexicographic order on the tuples.
We say that vectors arising from such Kronecker products have \emph{a tensor product structure}.

In the other direction, for $x\in \K^m$ and $y\in \K^n$ we may accomplish the reshaping of $v \in \K^{mn}$ to a matrix $A \in \K^m\otimes \K^n$ via
\[
    A[i,j] := v[(i,j)].
\]

\subsection{Symmetric Tensor Spaces}

\paragraph{Symmetric tensors and symmetric products.}
Let $S_k$ be the symmetric group on the elements of $[k]$.
A tensor $T\in (\K^n)^{\otimes k}$ is symmetric when for all indices $j_1, j_2, \dots, j_k \in [n]$ and any permutation $\sigma\in S_k$ we have
\[
T[j_1, j_2, \dots, j_k] = T[j_{\sigma(1)}, j_{\sigma(2)}, \dots, j_{\sigma(k)}].
\]
We denote the space of symmetric tensors $S^k(\K^n) \subseteq (\K^n)^{\otimes k}$.
For a tensor without symmetry its symmetric component is given by the projection $\Sym_k:(\K^n)^{\otimes k} \to (\K^n)^{\otimes k}$ defined index-wise for $j_1, j_2, \dots, j_k \in [n]$ by
\begin{equation}
\label{eq:tensor-symmetric-projection}
\Sym_k(T)[j_1, j_2, \dots, j_k] = \frac1{k!}\sum_{\sigma \in S_k} T[j_{\sigma(1)}, j_{\sigma(2)}, \dots, j_{\sigma(k)}].
\end{equation}
The image of $(\K^n)^{\otimes k}$ under $\Sym_k$ is $S^k(\K^n)$.
Symmetrization may be analogously defined for vectors that have appropriate tensor product structure.
A symmetric tensor in $(\K^n)^{\otimes k}$ (or the symmetrization of a vector with tensor product structure) may have many redundant entries.
We will describe linear operators on $S^k(\K^n)$ in a concise manner by specifying a canonical basis for $S^k(\K^n)$ and identifying elements by their coordinate vectors in this subspace.

For each multiset $\mset{j_1,j_2,\dots, j_k}$ of $[n]$, define the vector
\begin{equation}
    \label{eq:symmetric-vector}
    \vec{e}_{\mset{j_1,j_2,\dots,j_k}} = \frac{w_{\mset{j_1,j_2,\dots, j_k}}}{k!} \sum_{\sigma\in S_k} \vec{e}_{\sigma(j_1)} \otimes \vec{e}_{\sigma(j_2)}\otimes \cdots \otimes \vec{e}_{\sigma(j_k)}
\end{equation}
where
\begin{equation}
    \label{eq:multiset-weight}
    w_{\mset{j_1,j_2,\dots, j_k}} := \sqrt{\binom{k}{\multi(\mset{j_1,j_2,\dots, j_k})}}.
\end{equation}
\begin{lemma}
    \label{lemma:canonical-orthonormal-basis-symmetric-space}
    The set $\mathcal S_k = \{ \vec e_M : M \text{ a } $k$\text{-element multiset of } [n]\}$ is an orthonormal basis for $S^k(\K^n)$.
\end{lemma}
See Section~\ref{subsec:properties-symmetric-products} for the proof.
We will define $(\K^n)^{\oasterisk k}$ to be the space of coordinate vectors of $S^k(\K^n)$ in the basis $\mathcal S_k$ i.e. the vector of coefficients when writing a symmetric tensor as a sum of basis elements.
A vector in $(\K^n)^{\oasterisk k}$ then has $\mchoose{n}{k}$ entries which we index using $k$-element multisets of $[n]$.
One can then identify this space with $\K^{\mchoose{n}{k}}$ by sending $\vec{e}_M$ to the vector $\vec{e}_j$ where $j$ is the position of $M$ in DM order on $k$-element multisets of $[n]$.
The symbol ``$\oasterisk$'' will be used to indicate symmetric products of vectors.
\begin{definition}
    \label{def:symmetrized-Kronecker-product}
    We define a symmetrized Kronecker product of vectors $u^1, u^2,\dots, u^k \in \K^n$, denoted $u^1 \oasterisk u^2 \oasterisk \dots \oasterisk u^k \in (\K^n)^{\oasterisk k}$, as the coordinate-vector of $\Sym_k(u^1 \otimes u^2 \otimes  \dots \otimes u^k)$ in $\mathcal S$.
    Entry-wise,
    \begin{equation}
        \label{eq:symmetrized-Kronecker-product}
        (u^1 \oasterisk u^2 \oasterisk \dots \oasterisk u^k)[\mset{j_1,j_2,\dots, j_k}] := \frac{w_{\mset{j_1,j_2,\dots, j_k}}}{k!} \sum_{\sigma\in S_k} u^1[j_{\sigma(1)}] u^2[j_{\sigma(2)}] \cdots u^k[j_{\sigma(k)}]
    \end{equation}
    with $\vec{e}_M$ defined as in \eqref{eq:symmetric-vector}
    for each multiset $\mset{j_1,j_2,\dots, j_k}$ of $[n]$.
\end{definition}

We define the inner product on $(\K^n)^{\oasterisk k}$ for the standard basis vectors as usual where $\langle \vec{e}_M, \vec{e}_K\rangle = 1$ if $M=K$ and $0$ otherwise and extended by linearity to arbitrary coordinate vectors.
Since $\mathcal S$ is an orthonormal basis for $S^k(\K^n)$, taking coordinates in this basis preserves the inner product.

\begin{lemma}
    \label{lemma:symmetric-product-inner-product}
    For vectors $u^1,\dots, u^k, v^1,\dots, v^k\in\K^n$, $\Sym_k$ as given in \eqref{eq:tensor-symmetric-projection} and our inner product on $(\K^n)^{\oasterisk k}$, our symmetrized Kronecker product satisfies
    \[
    \langle u^1\oasterisk u^2\oasterisk \dots \oasterisk u^k, v^1 \oasterisk v^2\oasterisk \dots \oasterisk v^k  \rangle = \langle \Sym_k(u^{1}\otimes u^2 \otimes \cdots\otimes u^k) ,  \Sym_k(v^{1}\otimes v^2 \otimes \cdots\otimes v^k)\rangle.
    \]
\end{lemma}
\begin{corollary}
    \label{cor:symmetric-product-inner-product-expansion}
    For vectors $u^1,\dots, u^k, v^1,\dots, v^k\in\K^n$, and the inner product on $(\K^n)^{\oasterisk k}$,
    \[
    \langle u^1\oasterisk u^2\oasterisk \dots \oasterisk u^k, v^1 \oasterisk v^2\oasterisk \dots \oasterisk v^k  \rangle = \frac1{k!} \sum_{\sigma\in S_k}  \prod_{i=1}^k \langle u^{i}, v^{\sigma(i)}\rangle.
    \]
\end{corollary}
See Appendix~\ref{subsec:properties-symmetric-products} for proofs of these facts.

We may now identify linear operators $A : S^{k}(\K^m) \to S^{k}(\K^n)$ by identifying $A$ as a matrix in $\K^{\mchoose{m}{k}\times \mchoose{n}{k}}$.
We can construct such operators by defining a symmetric product on matrices.

\begin{definition}
    \label{def:symmetric-matrix-product}
    For matrices $A^1,A^2,\dots, A^k \in \K^{m\times n}$ define a fully symmetrized Kronecker product of these matrices, denoted $A^1 \oasterisk A^2 \oasterisk \dots \oasterisk A^k \in \K^{\mchoose{m}{k}\times \mchoose{n}{k}}$, entry-wise by
    \begin{equation}
        (A^1 \oasterisk A^2 \oasterisk \dots \oasterisk A^k)[\mset{j_1,j_2,\dots, j_k}, \mset{j'_1, j'_2,\dots, j'_k}] := 
        \frac{w_{\mset{j_1,j_2,\dots, j_k}}w_{\mset{j'_1, j'_2,\dots, j'_k}}}{(k!)^2} \sum_{\sigma\in S_k} \sum_{\tau\in S_k} \prod_{i=1}^k A^i[j_{\sigma(i)}, j'_{\tau(i)}],
    \end{equation}
    where $\mset{j_1,j_2,\dots, j_k}$ is a $k$-element multiset of $[m]$ and $\mset{j'_1, j'_2,\dots, j'_k}$ is a $k$-element multiset of $[n]$.
\end{definition}

\begin{lemma}
    \label{lemma:symmetric-product-equivalences}
    The fully symmetrized Kronecker product may be equivalently defined, using Definition~\ref{def:symmetrized-Kronecker-product} for the symmetric Kronecker product of vectors, either row-wise by
    \begin{equation}
        \label{eq:symmetric-matrix-product-row}
        (A^1 \oasterisk A^2 \oasterisk \dots \oasterisk A^k)[\mset{j_1,j_2,\dots, j_k}, :]= \frac{w_{\mset{j_1,j_2,\dots, j_k}}}{k!}\sum_{\sigma\in S_k}  A^1[j_{\sigma(1)}, :] \oasterisk A^2[j_{\sigma(2)}, :] \oasterisk \cdots \oasterisk A^k[j_{\sigma(k)}, :]
    \end{equation}
    for any $\mset{j_1,j_2,\dots, j_k}$ a $k$-element multiset of $[m]$,
    or column-wise by
    \begin{equation}
        \label{eq:symmetric-matrix-product-column}
        (A^1 \oasterisk A^2 \oasterisk \dots \oasterisk A^k)[:,\mset{j'_1, j'_2,\dots, j'_k}]= \frac{w_{\mset{j'_1, j'_2,\dots, j'_k}}}{k!}\sum_{\tau\in S_k}  A^1[:,j'_{\tau(1)}] \oasterisk A^2[:,j'_{\tau(2)}] \oasterisk \cdots \oasterisk A^k[:,j'_{\tau(k)}]
    \end{equation}
    for any $\mset{j'_1,j'_2,\dots,j'_k}$ a $k$-element multiset of $[n]$.
\end{lemma}
See Section~\ref{subsec:properties-symmetric-products} for the proof.
If $A^1 = A^2 = \dots = A^k = A$, then we may call  $A^{\oasterisk k}$ the \emph{$k$-th order symmetric lift} of the matrix $A$ and the formulas simplify considerably. Next, our symmetrized Kronecker product obeys a symmetrized version of the mixed-product rule.
\begin{lemma}
    \label{lemma:symmetric-matrix-products}
    For matrices $A^1, A^2, \dots, A^k \in \K^{m\times p}$ and $B^1, B^2, \dots, B^k \in \K^{p\times n}$,
    \[
        (A^1 \oasterisk A^2 \oasterisk \cdots \oasterisk A^k)(B^1 \oasterisk B^2 \oasterisk \cdots \oasterisk B^k) = \frac1{k!}\sum_{\rho\in S_k}(A^1 B^{\rho(1)})\oasterisk (A^2 B^{\rho(2)}) \oasterisk \cdots \oasterisk (A^k B^{\rho(k)}).
    \]
\end{lemma}
See Section~\ref{subsec:properties-symmetric-products} for the proof.
Applied to $k$-th order symmetric lifts this gives the following simpler corollary.
\begin{corollary}
    \label{cor:lift-of-product}
    For $A\in \K^{m\times p}$ and $B \in \K^{p\times n}$
    \[
        (AB)^{\oasterisk k} = A^{\oasterisk k} B^{\oasterisk k}.
    \]
\end{corollary}

\paragraph{Partial symmetries.} We say that a tensor in the space $(\K^{n_1})^{\oasterisk k_1} \otimes (\K^{n_2})^{\oasterisk k_2} \otimes \cdots \otimes (\K^{n_L})^{\oasterisk k_L}$ has $L$ different partial-symmetries along its modes.
The order of a tensor is this space is $k_1 + k_2 + \cdots + k_L$.
We define the standard basis for this space indexed by an $L$-tuple of multisets $M_i = \mset{j^i_1,j^i_2,\dots j^i_{k_i}}$ for $i=1,2,\dots,L$,
\[
    \vec{e}_{(M_1,M_2,\dots,M_L)} := \vec{e}_{M_1} \otimes \vec{e}_{M_2} \otimes \cdots \otimes \vec{e}_{M_L}.
\]
\textbf{Warning:} Unlike the tensor product, our symmetrized Kronecker product is not associative, e.g.\ $(w\oasterisk x) \oasterisk (y\oasterisk z) \neq w\oasterisk x \oasterisk y\oasterisk z$ since these belong to different tensor spaces.
Specifically, $(w\oasterisk x) \oasterisk (y\oasterisk z)\in \left((\K^n)^{\oasterisk 2}\right)^{\oasterisk 2} \cong \K^m$ where $m = \binom{\binom{n+1}{2}+1}{2}$ whereas $w\oasterisk x \oasterisk y\oasterisk z \in (\K^n)^{\oasterisk 4} \cong \K^k$ where $k = \binom{n+3}{4}$.

\paragraph{Reshaping tensors with symmetry.}
While the vectors of $(\K^n)^{\oasterisk k}$ are a convenient way of storing the entries of a symmetric tensor without redundancy, to use the multilinear structure of one must convert the tensor back to a matrix.
First, if $v\in (\K^n)^{\oasterisk 2}$ and gives the coordinates of a symmetric matrix in our canonical basis, we define the matricization entry-wise by
\begin{equation}
    \label{eq:matricization}
    \mathrm{matricize}(v)[j_1,j_2] = \frac1{w_{\mset{j_1,j_2}}} v[\mset{j_1,j_2}],\quad j_1,j_2\in[n].
\end{equation}
More generally for $\ell \leq k$, $\mathrm{matricize}_{\ell}:(\K^n)^{\oasterisk k}\mapsto (\K^n)^{\oasterisk \ell} \otimes (\K^n)^{\oasterisk k-\ell}$ reshapes a symmetric tensor $v$ to a matrix via
\begin{equation*}
    \mathrm{matricize}_{\ell}(v)[M,K] = \frac{w_M w_K}{w_{M\cup K}} v[M\cup K]
\end{equation*}
where $M$ is an $\ell$-element multiset of $[n]$ and $K$ is a $(k-\ell)$-element multiset of $[n]$.
This particular matrix is $\ell$-th \emph{catalecticant matrix} of $v$ expressed as coordinates in the bases $\mathcal S_\ell$ and $\mathcal S_{k-\ell}$ defined in Lemma~\ref{lemma:canonical-orthonormal-basis-symmetric-space} for the row space and column space respectively.

\section{Dihedral MRA and MRA with Projection}
\label{sec:dMRA-and-pMRA}
In this section, we define and compare the two main models of interest in this work: dihedral MRA (dMRA) and MRA with projection (pMRA).
\paragraph{Dihedral MRA.}
In the dMRA model, we observe $N$ samples according to
\begin{equation}
    y_i = g_i \cdot x + \sigma\xi_i
\end{equation}
where $x$ is some fixed signal in $\Rl^d$, each $g_i$ is a group element of the dihedral group $\dihgroup{d}$, each $\xi_i$ is independent standard Gaussian noise in $\Rl^d$, and $\sigma > 0$ is a scalar that determines the signal-to-noise ratio.
We will index entries in $x$ starting at zero so that $x = (x[0], x[1],\dots, x[d-1]).$

The dihedral group $\dihgroup{d}$ may be generated from cyclic shifts and reflections.
We will denote these $R_\ell$ and $J$ respectively with $\dihgroup{d} =\{\shiftel_0, \shiftel_1,\dots, \shiftel_{d-1}, J\shiftel_0, J\shiftel_1,\dots, J\shiftel_{d-1}\}$.
For $\ell \in [d)$, we define the action of $R_\ell \in \dihgroup{d}$ to be a cycling of entries in $x$ to the right by $\ell$ positions so that for each index $j\in [d)$,
\[
    (R_\ell \cdot x)[j] = x[j-\ell \hspace{-2mm} \pmod d].
\]
In general, we will assume that any indexing we do is done cyclically modulo the length of the vector involved so we may omit the modulus.
For $J\in \dihgroup{d}$, we define its action to be reversal of entries in a vector so that for an index $j\in [d]$, 
\[
(J\cdot x)[j] = x[d-j-1].
\]
Using the samples $\{y_i\}_{i=1}^N$, one aims to recover the signal $x$ up to the action of $\dihgroup{d}$.

\paragraph{MRA with projection.} Our second model is MRA with projection (pMRA). In the pMRA model, we observe $N$ samples according to
\begin{equation}
    y_i = \Pi(g_i \cdot x) + \sigma \xi_i.
\end{equation}
Here now $x$ is a fixed signal in $\Rl^{2d}$, $g_i$ are group elements of the cyclic group $\cycgroup{2d}=\{\shiftel_0, \shiftel_1,\dots, \shiftel_{2d-1}\}$, and again $\xi_i$ is independent standard Gaussian noise (now in $\Rl^{2d}$) with $\sigma > 0$ a scalar determining the signal-to-noise ratio.
We also define the linear operator, $\Pi:\Cx^{2d}\to \Cx^d$, defined entry-wise by for $j\in[d)$ by
\begin{equation}
    (\Pi x)[j] = x[j] + x[2d-j-1].
\end{equation}
This projection operation is a proxy for tomographic projection or the X-ray transform in the sense that if we view the object as being the discretization of some ``thin-shell'' circular 2D object with $x_j$ the density in the $j$-th patch, projection sums two entries when they are align vertically (see Figure \ref{fig:pMRA-model-diagram}).

In this model, one might aim to recover the signal $x$ up to the action of $\dihgroup{2d}$ from the projections $\{y_i\}_{i=1}^N$. Note that while the group actions are cyclic, we can only hope recover the signal up to action by the dihedral group. Indeed, a signal $x$ and its reflection $Jx$ will produce the same distribution over samples.
In this manner, the pMRA model captures the same ambiguity that occurs in tomography; when the orientations in three dimensions are unknown, one can only distinguish an object up to reflection.
However, in another regard the model suffers from a small issue.
Complete recovery for a signal is impossible since the vector $c = (1,-1,1,-1,\dots, -1)$ lies in the kernel of $\projOp R_\ell$ for every $\ell \in [d]$.
This component is the \emph{Nyquist component} and under the Fourier transform corresponds to the $d$-th Fourier coefficient of $x$.
Since this component is of highest frequency, the inability to recover this component could be taken as akin to an optical or electron-imaging system having some limited resolving power. In light of this, we refine the goal: using samples $\{y_i\}_{i=1}^N$, one aims to recover all Fourier coefficients of the signal $x$ except the $d$-th one, up to the action of $\dihgroup{2d}$.

\paragraph{Advantages of the Fourier Basis.}
The action of $\cycgroup{2d}$ on $\Cx^{2d}$ is diagonalized in the Fourier basis and gives the standard representation for the cyclic group.
However, $\dihgroup{2d}$ is not a commutative group, the actions of its elements cannot be simultaneously diagonalized, and the matrices giving the action in the Fourier basis are all sparse with non-zero entries only on the diagonal or anti-diagonal.

For a signal $x\in\Rl^d$, denote the Fourier coefficients of $x$ indexed from zero, by
\[
    \hat{x}[k] := \sum_{j=0}^{d-1} x[j]\cdot e^{-2\pi i j k/d}
\]
where $i=\sqrt{-1}$.
Define 
\[
    \hat{x} := (\hat{x}[k] : k=[d))\in \Cx^d.
\]
Implicitly $\hat{x} = \mathcal F_d x$ where $\mathcal F_d:\Cx^{d \times d}$ is a linear transformation giving the change of basis from the standard basis into the Fourier basis in $\Cx^d$.
In addition, since $x$ is a real signal, we have that $\hat x$ satisfies $\hat{x}[-k] = \overline{\hat{x}[k]}$ with indices taken modulo $d$ so that they lie in $[d)$. Next, define the following matrices $\fshiftel_{\ell,d}\in\Cx^{d\times d}$, $\ell\in[d)$, and $\freflel_d\in\Cx^{d\times d}$, where for $v\in\Cx^d$ with entries indexed by $j\in[d)$,
\begin{align}
    \label{eq:fourier-shift-operation}
    (\fshiftel_{\ell,d} \cdot v)[j] &= e^{-2\pi i \ell j / (2d)}v[j], \\
    \label{eq:fourier-reflect-operation}
    (\freflel_d \cdot v)[j] &= v[-j],
\end{align}
and the matrix $\fprojOp_{2d} \in \Cx^{d\times 2d}$, where for $v\in \Cx^{2d}$,
\begin{equation}
    \label{eq:fourier-projection-operation}
    (\fprojOp_{2d} v)[j] =  v[j] + e^{2\pi i j/{2d}}v[-j].
\end{equation}
We will omit the sub-scripted dimensions $d$ and $2d$ where it may be inferred from context.

\subsection{Invariants from the dMRA and pMRA Models}
\label{subsec:invariants-dMRA-pMRA-models}
Using the linear operators defined in \eqref{eq:fourier-shift-operation}-\eqref{eq:fourier-projection-operation} we define the following invariant tensors of a complex vector $\hat x \in \Cx^{2d}$,
\begin{align}
     T_{\dMRA,2d}^{(r)}(\hat x) &:= \frac1{4d} \left(\sum_{\ell=0}^{2d-1}  (\fshiftel_{\ell,2d}\cdot \hat{x})^{\oasterisk r} + \sum_{\ell=0}^{2d-1}  (\freflel\fshiftel_{\ell,2d} \cdot \hat{x})^{\oasterisk r}\right), \label{eq:dMRA-invariant-tensor} \\
     T_{\pMRA,2d}^{(r)}(\hat x) &:= \frac1{2d} \left(\sum_{\ell=0}^{2d-1}  (\fprojOp_{2d} \fshiftel_{\ell,2d} \cdot \hat{x})^{\oasterisk r} \right). \label{eq:pMRA-invariant-tensor}
\end{align}

As suggested by our notation, we claim that these tensors are the same invariants from the dMRA and pMRA models given by \eqref{eq:invariant-tensors-with-projection} under a change of basis.
First we give the following lemma.
\begin{lemma}
    \label{lemma:almost-diagonalize-groups}
    For any $d$,
    \[
        \fshiftel_\ell = \mathcal F \shiftel_\ell \mathcal F^{-1},
    \]
    for any $\ell\in[d)$,
    \[
        \freflel = \mathcal F \reflel\shiftel_{-1} \mathcal F^{-1},
    \]
    and there is a full-rank square matrix $C$ giving a change of basis such that
    \[
        \fprojOp_{2d} = C\projOp_{2d} \mathcal F^{-1}.
    \]
\end{lemma}
Using this we have the following result.
\begin{proposition}
    \label{prop:invariants-change-of-basis}
    Let $x\in\Rl^{2d}$ and let $T^{(r)}_{\dihgroup{2d}}(x)$ and $T^{(r)}_{\cycgroup{2d},\Pi}(x)$ be $r$-order invariant tensor of $x$ as defined in \eqref{eq:invariant-tensors-with-projection}.
    For $\mathcal F\in \Cx^{2d\times 2d}$ the matrix giving the change of basis from the standard basis into the Fourier basis in $\Cx^{2d}$,
    \[
        T_{\dMRA,2d}^{(r)}(\mathcal F x) = \mathcal F^{\oasterisk r} \left(T^{(r)}_{\dihgroup{2d}}(x)\right).
    \]
    For $C$ the change-of-basis matrix in Lemma~\ref{lemma:almost-diagonalize-groups}, 
    \[
        T_{\pMRA,2d}^{(r)}(\mathcal F x) = C^{\oasterisk r} \left(T^{(r)}_{\cycgroup{2d},\Pi}(x)\right).
    \]
\end{proposition}
The subscript $2d$ may be dropped when the dimension is clear from context.
See Section~\ref{subsec:fourier-basis} for the proof of Lemma~\ref{lemma:almost-diagonalize-groups} and Proposition~\ref{prop:invariants-change-of-basis}, including the explicit construction of the change-of-basis matrix $C$.
Lemma~\ref{lemma:almost-diagonalize-groups} is also useful in proving the following fact.
\begin{proposition}
    \label{prop:two-clean-projections-sufficient}
    In the pMRA model, two projections $\{\projOp\shiftel_{\ell_1}x, \projOp\shiftel_{\ell_2}x\}$ from known directions $\ell_1\neq \ell_2$ are sufficient to recover
    $x\in\Rl^{2d}\cap\Span(1,-1,1,\dots,-1)^\perp$ when $\ell_1-\ell_2$ does not divide $d$.
\end{proposition}
See Section~\ref{subsec:fourier-basis} for the proof.
As we mentioned, strategies to label samples provably fail in the high-noise regime; however, Proposition~\ref{prop:two-clean-projections-sufficient} demonstrates the irony of the pMRA model: the third moment is an overcomplete tensor and we do not have an algorithm to decompose it; we also do not need all the factors from its decomposition, just two!

The careful reader may wonder why we give our dMRA invariant tensor in \eqref{eq:dMRA-invariant-tensor} for a signal of length $2d$ instead of $d$.
By defining both models with a common dimension we may give a relationship between the invariants:
\begin{lemma}
    \label{lemma:pMRA-invariants-from-dMRA-invariants}
    For $\hat x \in \Cx^{2d}$,
    \begin{equation*}
        (\fprojOp^{\oasterisk r}_{2d}) T_{\dMRA,2d}^{(r)}(\hat x) = T_{\pMRA,2d}^{(r)}(\hat x).
    \end{equation*}
\end{lemma}

\begin{proof}
    By linearity and Lemma~\ref{cor:lift-of-product} we may move the projection operator into the summands
    \[
    (\fprojOp^{\oasterisk r}) T_{\dMRA}^{(r)}(\hat x) = \frac1{4d} \left(\sum_{\ell=0}^{2d-1}  (\fprojOp\fshiftel_\ell \hat{x})^{\oasterisk r} + \sum_{\ell=0}^{2d-1}  (\fprojOp \freflel\fshiftel_\ell \hat{x})^{\oasterisk r}\right).
    \]
    Substituting the relations from Proposition~\ref{lemma:almost-diagonalize-groups} we have that
    \[
    \fprojOp\freflel = C\projOp \reflel \shiftel_{-1} \mathcal F^{-1} = C\projOp \shiftel_{-1} \mathcal F^{-1} = C\projOp\mathcal F^{-1}  \fshiftel_{-1} = \fprojOp\fshiftel_{-1}
    \]
    since $\projOp \reflel = \projOp$. If we re-index the second sum we may combine, yielding the result:
    \[
        (\fprojOp^{\oasterisk r}) T_{\dMRA}^{(r)}(\hat x) = \frac1{4d} \left(\sum_{\ell=0}^{2d-1}  (\fprojOp\fshiftel_\ell \hat{x})^{\oasterisk r} + \sum_{\ell=0}^{2d-1}  (\fprojOp\fshiftel_{\ell-1} \hat{x})^{\oasterisk r}\right) = \frac1{2d}\sum_{\ell=0}^{2d-1}  (\fprojOp\fshiftel_\ell \hat{x})^{\oasterisk r} = T^{(r)}_\pMRA(\hat x). \qedhere
    \]
\end{proof}

As a result, the $\dihgroup{2d}$-invariant polynomials that make up the entries of the $T_{\pMRA}^{(r)}(\hat x)$ tensor are linear combinations of the $\dihgroup{2d}$-invariant polynomials in $T_{\dMRA}^{(r)}(\hat x)$. In other words, the invariant dMRA tensor contains at least as much information as the pMRA tensor.
The advantage of our change of basis by the Fourier transform is apparent in computations giving each entry in the invariant tensors for both models.
\begin{proposition}
    \label{prop:invariant-tensor-structure}
    For $\hat{x} \in \Cx^{2d}$,
    \[
    T_\dMRA^{(r)}(\hat x)[\Lbag j_1,\dots, j_r\Rbag] = \frac{w_{\mset{j_1,\dots, j_r}}}2\left(\prod_{k=1}^r \hat x[j_k] + \prod_{k=1}^r \hat x[-j_k]\right) \bb1\left\{\sum_{k=1}^r j_k \equiv 0 \hspace{-2mm} \pmod{2d}\right\}
    \]
    for $\mset{j_1,\dots, j_r}$ a multiset of $[2d)$.
    
    For the same signal,
    \begin{align*}
        &T_\pMRA^{(r)}(\hat x)[\mset{j_1,\dots, j_r}] = \\
        &\qquad\frac{w_{\mset{j_1,\dots, j_r}}}2\sum_{s\in\{-1,1\}^r}  \phi(s,\mset{j_1,\dots, j_r}) \left(\prod_{k=1}^r  \hat x[s_k j_k] + \prod_{k=1}^r  \hat x[-s_k j_k]\right)\bb1\left\{\sum_{k=1}^r s_kj_k \equiv 0 \pmod{2d}\right\}
    \end{align*}
    for $\mset{j_1,\dots, j_r}$ a multiset of $[d)$ and where $\phi$ is a complex phase given by
    \[
     \phi(s, \mset{j_1,\dots, j_r}) := (-1)^{\frac1{2d}\sum_{k=1}^r s_k j_k}\exp\left(\frac{\pi i}{2d}\sum_{k=1}^r j_k \right).
    \]
\end{proposition}
See Section~\ref{subsec:invariant-tensor-structure} for the computations that verify these facts.
Observe that each entry of the $T_\dMRA^{(r)}(\hat x)$ tensor is either zero or simply the sum of two monomials up to scale.
Further, it is relatively simple to verify these polynomials are invariant under the actions given in \eqref{eq:fourier-shift-operation} and \eqref{eq:fourier-reflect-operation}.
In contrast, in the standard basis the entries of $T^{(r)}_{\dihgroup{2d}}(x)$ and $T^{(r)}_{\cycgroup{2d},\Pi}(x)$ will be the sum of many monomials.
We can tie sparsity directly to computational efficiency later and it also eases several future constructions.

\section{Method for Recovery}
\label{sec:method-for-recovery}
In this section, we first define generally the conditions which allow us to extract a subtensor of an invariant tensor that is also an invariant tensor for some model.
We demonstrate that this is possible for both the dMRA and pMRA invariant tensors and that the procedure produces smaller dMRA and pMRA invariant tensors such that an orbit recovery problem's dimension may be recursively halved.
First, we introduce some notation we use to indicate extracting subtensors.

\paragraph{Restriction and subtensors.}
Given a vector $x\in\K^n$ and a set of indices $S\subseteq [n]$ we define the restriction of $x$ to $S$ using either the notation $x|_S$ or $x[S] := (x[s] : s\in S) \in \K^{|S|}$.
For a matrix $A\in\K^{m\times n}$, $R\subseteq [m]$, and $S\subseteq[n]$ subsets.
We will use $A[R,S]$ to denote the submatrix of $A$ with rows in $R$ and columns in $S$.
When $R=[m]$ or $S=[n]$ we will use $A[:,S]$ or $A[R, :]$ respectively.

Similarly, for a tensor $ T \in \Rl^{d_1}\otimes \Rl^{d_2}\otimes \cdots \otimes \Rl^{d_r}$ and subsets $S_j\subseteq [d_j]$, $j=1,\dots,r$ we use the following to denote restriction along multiple modes:
\[
 T |_{S_1,S_2,\dots, S_r} \text{ or }  T[S_1, S_2,\dots, S_r] := [  T[s_1, s_2,\dots s_r] : s_1 \in S_1,s_2\in S_2, \dots, s_r\in S_r] \in \K^{|S_1|\times \cdots \times |S_r|}
\]
is a subtensor of $\mathcal T$.
For symmetric tensors $\mathcal T \in (\Rl^n)^{\oasterisk r}$ we may use a single restriction to indicate restriction on all modes by the same subset of indices:
\[
 T|_S :=  [T[\mset{s_1, s_2,\dots, s_r}] : s_1, s_2,\dots, s_r \in S] \in (\K^{|S|})^{\oasterisk r},
\]
or when the subsets $\{S_i\}_{i=1}^r$ are pairwise equal or disjoint, then the subtensor of $T$ defined by
\[
T[S_1, S_2,\dots, S_r] := [  T[\mset{s_1, s_2,\dots, s_r}] : s_1 \in S_1,s_2\in S_2, \dots, s_r\in S_r]
\]
is partially symmetric along modes where the subsets $S_j$ and $S_k$ are equal for $j,k\in[r]$.

We quickly give some implications for restriction on matrix-vector products and matrix-matrix products, that are handy.
\begin{lemma}
    \label{lemma:restriction-across-products}
    Let $A\in \K^{m\times n}$, $B\in \K^{m\times n}$ be matrices and $v\in\K^n$ be a vector.
    Let $R\subseteq [m]$, and $S\subseteq[n]$ be subsets of indices.
    Then
    \[
        (Av)[R] = A[R, :] v \qquad \text{and} \qquad (AB)[R,S]= A[ R, :] B[:, S].
    \]
\end{lemma}
This follows directly from our definition of restriction and matrix multiplication.

\subsection{Invariant Subtensors}
\label{subsec:invariant-subtensors}

\begin{lemma}[Moment subtensors from restriction]
    \label{lemma:invariant-subtensors}
    Let $G$ be a finite group with some action on $\K^m$, $h:\K^m \to \K^n$ be a linear transformation, and $R\subseteq [m]$, $S \subseteq[n]$ both be sets of indices. If $\tilde G$ is a finite group with some action on $\K^{|R|}$ and $
    \tilde h: \K^{|R|} \to \K^{|S|}$ a linear transformation such that for all $v\in \K^m$,
    \begin{equation}
        \label{eq:multiset-equality-condition}
        \mset{ (h (g\cdot v))|_S : g\in G} = k\times \mset{\tilde h (\tilde g\cdot v|_R) : \tilde g\in \tilde G},
     \end{equation}
    where for a multiset $M$, $k\times M$ indicates a multiset with each item in $M$ repeated $k$ times.
    Then restricting each mode of a $G$-invariant tensor to $S$ produces a $\tilde G$-invariant tensor under the projection $\tilde h$ up to scale,
    \[
        \left(T_{G,h}^{(r)}(v)\middle)\right|_S = k\frac{|\tilde G|}{|G|} T_{\tilde G, \tilde h}^{(r)}(v|_R).
    \]
\end{lemma}

\begin{proof}{Proof of Lemma~\ref{lemma:invariant-subtensors}.}
    First,
    \[
        \left[T_{G,h}^{(r)}(v)\middle]\right|_S = \frac1{|G|}\sum_{g\in\mathcal G} [(h (g\cdot v))^{\oasterisk r}]|_S = \frac1{|G|}\sum_{g\in\mathcal G} (h (g\cdot v)|_S)^{\oasterisk r}
    \]
    by our definition of restriction for a symmetric tensor.
    Next, applying the multiset condition we have
    \[
        \frac1{|G|}\sum_{g\in\mathcal G} (h (g\cdot v)|_S)^{\oasterisk r} = \frac{k}{|G|}\sum_{\tilde{g}\in\tilde{\mathcal G}} (\tilde h (\tilde g\cdot v|_R))^{\oasterisk r},
    \]
    so with some manipulation,
    \[
        \frac{k}{|G|}\sum_{\tilde{g}\in\tilde{\mathcal G}} (\tilde h (\tilde g\cdot v|_R))^{\oasterisk r} = \frac{k|\tilde G|}{|G|}\cdot \frac1{|\tilde G|}\sum_{\tilde{g}\in\tilde{\mathcal G}} (\tilde h (\tilde g\cdot v|_R))^{\oasterisk r} = \frac{k|\tilde G|}{|G|} T_{\tilde G, \tilde h}^{(r)}(v|_R),
    \]
    and $\left[T_{G,h}^{(r)}(v)\middle]\right|_S = k\frac{|\tilde G|}{|G|} T_{\tilde G, \tilde h}^{(r)}(v|_R)$.
\end{proof}
The property above is useful since when it holds, a particular subtensor is actually an invariant tensor in its own right and --- if $k>1$ --- is an invariant over a smaller group.

We now show that the invariant tensors from the dMRA and pMRA models both satisfy Lemma~\ref{lemma:invariant-subtensors} when restricting to even Fourier coefficients.
To denote the even numbers in $[d)$ we use
\begin{equation}
    \label{eq:even-numbers}
    E_d = \{j \in [d) : j\equiv 0 \pmod{2}\}.
\end{equation}
We will omit the subscript and simply write ``$E$'' if the dimension can be inferred from context.
We also use $\fprojOp_{2d} = \fprojOp$ to indicate the projection operation defined in \eqref{eq:fourier-projection-operation} and $\fprojOp_{4d}$ an analogous operator defined in the same manner but with the domain $\Cx^{4d}$.

\begin{proposition}
    \label{prop:recursive-invariants}
    For $d\ge 1$, the following invariant tensors have invariant subtensors found by restriction:
    \begin{enumerate}
        \item For $\hat{x} \in \Cx^{2d}$, \[
            \left( T^{(r)}_{\dMRA, 2d}(\hat{x})\middle)\right|_{E_{2d}} =  T^{(r)}_{\dMRA, d}(\hat{x}|_{E_{2d}}).
        \]
        \item For $\hat{x} \in \Cx^{4d}$, \[
            \left(T^{(r)}_{\pMRA,4d}(\hat{x})\middle)\right|_{E_{2d}} =  T^{(r)}_{\pMRA,2d}(\hat{x}|_{E_{4d}}).
        \]
    \end{enumerate}
\end{proposition}

\begin{proof}
    For each case we show that the multiset condition \eqref{eq:multiset-equality-condition} holds so that Lemma~\ref{lemma:invariant-subtensors} applies.
    Since we deal with the cyclic group and dihedral group over vector spaces of different sizes we will include the dimension in the subscript on elements, i.e.\ $\cycgroup{d}=\{\fshiftel_{\ell,d}\}_{\ell=1}^{d-1}$ where
    \[
        (\fshiftel_{\ell,d} \cdot \hat{x})[j] = e^{-2\pi i j\ell/d} x[j]
    \]
    for $j\in[d)$. Similarly, we will use $\freflel_d$ for the reflection element on $\Cx^d$.
    \begin{enumerate}
        \item Let $\hat x\in \Cx^{2d}$ be any vector.
        Given $m \in [d)$ we define $\ell_k=m+dk \in [2d)$ for $k\in\{0,1\}$. 
        We claim that for each $k$, $(\fshiftel_{\ell_k,2d} \cdot \hat{x})|_{E_{2d}} = \fshiftel_{m,d} \cdot (\hat{x}|_{E_{2d}})$ and $(\freflel_{2d} \cdot \fshiftel_{\ell_k,2d} \cdot \hat{x})|_{E_{2d}} = \freflel_d \cdot\fshiftel_{m,d} \cdot (\hat{x}|_{E_{2d}})$ so the multiset condition holds.
        To this end, notice that since $\ell_k=m+kd$,
        \begin{align*}
            e^{-2\pi i \ell_k/d} &= e^{-2\pi i m/d} & \implies \\
            \text{for any } j=[d),\qquad e^{-2\pi i j\ell_k/d} &= e^{-2\pi ji m/d} & \implies \\
            e^{-2\pi i j\ell_k/d}\hat x[2j] &= e^{-2\pi ji m/d}\hat x[2j] & \implies \\
            e^{-2\pi i (2j)\ell_k/(2d)}\hat x[2j] &= e^{-2\pi ji m/d}(\hat x|_{E_{2d}})[j] & \implies \\
            (\fshiftel_{\ell_k,2d}\cdot\hat x)[2j] &= \left(\fshiftel_{m,d}\cdot  (\hat x|_{E_{2d}})\right)[j] & \implies \\
            \left((\fshiftel_{\ell_k,2d}\cdot\hat x)|_{E_{2d}}\right)[j] &= \left(\fshiftel_{m,d}\cdot  (\hat x|_{E_{2d}})\right)[j],
        \end{align*}
        and so $(\fshiftel_{\ell_k,2d} \cdot \hat{x})|_{E_d} = \fshiftel_{m,d} \cdot (\hat{x}|_{E_d})$.
        Similarly,
        \begin{align*}
            e^{2\pi i \ell_k/d} &= e^{2\pi i m/d} & \implies \\
            \text{for any } j=[d),\qquad e^{2\pi i j\ell_k/d} &= e^{2\pi ji m/d} & \implies \\
            e^{2\pi i j\ell_k/d}\hat x[-2j] &= e^{2\pi ji m/d}\hat x[-2j] & \implies \\
            e^{-2\pi i (-2j)\ell_k/(2d)}\hat x[-2j] &= e^{-2\pi (-j)i m/d}(\hat x|_{E_{2d}})[-j] & \implies \\
            (\freflel_{2d}\cdot \fshiftel_{\ell_k,2d}\cdot\hat x)[2j] &= \left(\fshiftel_{m,d}\cdot  (\hat x|_{E_{2d}})\right)[-j] & \implies \\
            \left((\freflel_{2d}\cdot\fshiftel_{\ell_k,2d}\cdot\hat x)|_{E_{2d}}\right)[j] &= \left(\freflel_d\cdot\fshiftel_{m,d}\cdot  (\hat x|_{E_{2d}})\right)[j],
        \end{align*}
        and so $(\freflel_{2d} \cdot \fshiftel_{\ell_k,2d} \cdot \hat{x})|_{E_d} = \freflel_d \cdot\fshiftel_{m,d} \cdot (\hat{x}|_{E_d})$.
        This completes the association of elements between the sets so that 
        \[
            \mset{(g\cdot \hat x)|_{E_d} : g\in \dihgroup{2d}} = 2\times \mset{(\tilde g \cdot \hat x|_{E_d}) : \tilde g \in \dihgroup{d}}.
        \]
        \item Let $\hat x\in \Cx^{4d}$. Set $m\in [2d)$ and let $\ell_k=m+(2d)k \in [4d)$ for $k\in\{0,1\}$.
        We aim to show that
        \[
            \left.\left(\fprojOp_{4d}(\fshiftel_{\ell_k,4d}\cdot \hat x)\right)\right|_{E_{2d}}= \fprojOp_{2d}\left(\fshiftel_{m,2d}\cdot (x|_{E_{4d}})\right).
        \]
        Reusing analysis from the first part, we know that $e^{-2\pi i j\ell_k/(2d)} \hat x[2j] = e^{-2\pi ijm/(2d)} \hat x[2j]$ for all $j\in[2d)$.
        Rewriting this, we have that $e^{-2\pi i j\ell_k/(2d)} \hat x[2j] = e^{-2\pi i (2j)\ell_k/(4d)} \hat x[2j] = (\fshiftel_{\ell_k,4d}\cdot \hat x)[2j]$ and that $e^{-2\pi ijm/(2d)} \hat x[2j] = e^{-2\pi ijm/(2d)} (\hat x|_{E_{4d}})[j] = \left(\fshiftel_{m,2d}\cdot (\hat x|_{E_{4d}})\right)[j]$ for all $j\in[2d)$.
        
        Now we take linear combinations to produce the $d$ different equations
        \[
            e^{-2\pi i j\ell_k/(2d)} \hat x[2j] + e^{2\pi ij/(2d)}e^{2\pi i j\ell_k/(2d)} \hat x[-2j] = e^{-2\pi i jm/(2d)} \hat x[2j] + e^{2\pi ij/(2d)}e^{2\pi i jm/(2d)} \hat x[-2j]
        \]
        with $j\in[d)$.
        Applying the rewritings gives
        \[
        (\fshiftel_{\ell_k,4d}\cdot \hat x)[2j] + \underbrace{e^{2\pi i j/(2d)}}_{=e^{2\pi i (2j)/(4d)}}(\fshiftel_{\ell_k,4d}\cdot \hat x)[-2j] = \left(\fshiftel_{m,2d}\cdot (\hat x|_{E_{4d}})\right)[j] + e^{2\pi i j/(2d)}\left(\fshiftel_{m,2d}\cdot (\hat x|_{E_{4d}})\right)[-j]. \]
        We then have that by the definition of the $\fprojOp$ operator that
        \begin{align*}
            \left(\fprojOp_{4d}(\fshiftel_{\ell_k,4d}\cdot \hat x)\right)[2j] &=  \left(\fprojOp_{2d}\left(\fshiftel_{m,2d}\cdot (\hat x|_{E_{4d}})\right)\right)[j] & \implies \\
            \left(\left(\fprojOp_{4d}(\fshiftel_{\ell_k,4d}\cdot \hat x)\right)\middle|_{E_{2d}}\right)[j] &= \left(\fprojOp_{2d}\left(\fshiftel_{m,2d}\cdot (\hat x|_{E_{4d}})\right)\right)[j].
        \end{align*}
        Then
        \[
            \mset{(\fprojOp_{4d}(g\cdot \hat x))|_{E_d} : g\in \cycgroup{4d}} = 2\times \mset{\fprojOp_{2d}\left(\tilde g\cdot (\hat x|_{E_{2d}})\right): \tilde g \in \cycgroup{2d}}. \qedhere
        \]
    \end{enumerate}
\end{proof}

Proposition~\ref{prop:recursive-invariants} shows that we can recover part of the signal by extracting a subtensor and then solving (recursively) the orbit recovery problem associated to it.
Once we recover an element $y$ in the $\dihgroup{d}$-orbit of $\hat x[E]$, our next goal will be to extend $y$ to a vector double its length that lies in the $\dihgroup{2d}$-orbit of $\hat x$.

\subsection{Variety-Constrained Linear Systems}
\label{subsec:variety-constrained-linear-systems}
In this section, we introduce a broad class of problems we call \emph{variety-constrained linear systems}.
The construction of such systems, an algorithm for solving them, and conditions that guarantee uniqueness of solutions, will allow us to complete the second stage of recovery.
First, we introduce some definitions.

A \emph{variety}  $\mathcal X\subseteq \K^n$ is the zero set of a collection of polynomials $p_1, p_2,\dots, p_k$, in the entries of a vector $x\in \K^n$ i.e.,
\[
    \mathcal X = \{ x\in \K^n : p_1(x) = p_2(x) = \cdots= p_k(x) = 0\}.
\]
The polynomials $p_1, p_2,\dots, p_k$ are said to \emph{cut out} $\mathcal X$.
A \emph{conic variety} is a variety which is closed under scalar multiplication and is cut out by homogeneous polynomials which can be chosen to have the same degree.
A polynomial $p_j$ that is homogeneous of degree $b$ is determined by the coefficients on each degree-$b$ monomial of $x$ and we may view each one as a vector in $(\K^n)^{\oasterisk b}$ such that $p_j(y) = \langle p_j, y^{\oasterisk b} \rangle$ for any $y\in \K^n$.
A variety is said to be \emph{irreducible} if it cannot be written as the union of smaller varieties.

Since our symmetric product will be unfamiliar, we will gives some examples of varieties, the polynomials that cut them out, and how to express these polynomials as coordinate vectors in $\mathcal S$.
\begin{example}
    \label{ex:variety-rank-1 matrices}
    Define the variety of $m\times n$ matrices of rank at most 1 as
    \[
        \mathcal X_1 = \{B\in \K^{mn} : \rank(B) \leq 1\}
    \]
    where $B$ is considered a matrix when finding $\rank(B)$.
    It is a well known that the $2\times 2$ minors of a matrix cut out $\mathcal X_1$.
    So if $B\in \mathcal X_1$, for $1\leq a < c \leq m$ and $1\leq b < d \leq n$, $B[a,b] B[c,d] - B[a,d]B[c,b] = 0$.
    We may write this as
    \[
        \langle \vec e_{(a,b)} \otimes \vec e_{(c,d)} - \vec e_{(a,d)} \otimes \vec e_{(c,b)} , B^{\otimes 2}\rangle =B[(a,b)] B[(c,d)] - B[(a,d)]B[(c,b)]= 0
    \]
    where to take the inner product we now treat $B$ as a vector with extra structure on its entries.
    However, observe that $B^{\otimes 2} \in S^2(\K^{mn})$ so 
    \begin{align*}
        \langle \vec e_{(a,b)} \otimes \vec e_{(c,d)} - \vec e_{(a,d)} \otimes \vec e_{(c,)} , B^{\otimes 2}\rangle &= \langle \Sym_2(\vec e_{(a,b)} \otimes \vec e_{(c,d)} - \vec e_{(a,c)} \otimes \vec e_{(b,d)}) , \Sym_2(B^{\otimes 2})\rangle \\
        &=\left\langle \frac{\vec e_{\mset{(a,b), (c,d)}}}{w_{\mset{(a,b), (c,d)}}} - \frac{\vec e_{\mset{(a,d), (c,b)}}}{w_{\mset{(a,d), (c,b)}}}, B^{\oasterisk 2}\right\rangle
    \end{align*}
    where the second step follows from the fact that $w_{\mset{X,Y}}\Sym_2(\vec{e}_X \otimes \vec{e}_Y) = \vec e_{\mset{X,Y}}$ and Lemma~\ref{lemma:symmetric-product-inner-product}.
    So each minor may be identified with a vector of the form
    \[
        \frac{\vec e_{\mset{(a,b), (c,d)}}}{w_{\mset{(a,b), (c,d)}}} - \frac{\vec e_{\mset{(a,d), (c,b)}}}{w_{\mset{(a,d), (c,b)}}} \in (\K^{mn})^{\oasterisk 2}.
    \]
\end{example}
\begin{example}
    \label{ex:variety-symmetric-rank-1 matrices}
    Define the variety of symmetric matrices of rank at most 1 as
    \[
        \mathcal X^\vee_1 = \{B\in (\K^n)^{\oasterisk 2} : \rank(B) \leq 1\}
    \]
    where since $B$ is given in the coordinate space of symmetric matrices, we define $\rank(B) := \rank(\mathrm{matricize}(B))$ with the matricization operation defined in \eqref{eq:matricization}.
    If $B\in (\K^n)^{\oasterisk 2}$, $\mathrm{matricize}(B)[a,b] = \frac1{w_{\mset{a,b}}}B[\mset{a,b}]$ so we may write a minor directly as
    \[
        \frac1{w_{\mset{a,b}}}B[\mset{a,b}]\frac1{w_{\mset{c,d}}} B[\mset{c,d}] - \frac1{w_{\mset{a,d}}}B[\mset{a,d}]\frac1{w_{\mset{c,b}}}B[\mset{c,b}]
    \]
    for $a,b,c,d\in[n]$.
    Following the same process as in Example~\ref{ex:variety-rank-1 matrices}, we may write a vector of the form
    \begin{equation}
        \label{eq:symmetric-minors-as-vectors}
        \frac{\vec e_{\mset{\mset{a,b}, \mset{c,d}}}}{w_{\mset{\mset{a,b}, \mset{c,d}}}w_{\mset{a,b}} w_{\mset{c,d}}} - \frac{\vec e_{\mset{\mset{a,d}, \mset{c,b}}}}{w_{\mset{\mset{a,d}, \mset{c,b}}}w_{\mset{a,d}} w_{\mset{c,b}}} \in ((\K^n)^{\oasterisk 2})^{\oasterisk 2}.
    \end{equation}
    When this is not the zero vector (i.e.\ when $\mset{\mset{a,b}, \mset{c,d}} \neq \mset{\mset{a,d}, \mset{c,b}}$) the inner product against $B^{\oasterisk 2}$ computes a minor of $\mathrm{matricize}(B)$.
\end{example}

\begin{definition}[Variety-constrained linear system]
    We define a variety-constrained linear system (VC-LS) to be
    \[
    Au = b \quad \text{subject to }\quad  u \in \mathcal X
    \]
    where $A\in \K^{m\times n}$ is a coefficient matrix, $u\in \K^n$ is the solution vector, $b\in \K^m$ is another vector and $\mathcal X \subseteq\K^n$ is an irreducible conic variety.
\end{definition}
Notably these systems may have unique solutions even when $m < n$.
The condition $u\in \mathcal X$ could be considered a form of ``sparsity'' and as such, algorithms to solve VC-LSs can be considered a form of compressive sensing~\cite{bousse_linear_2018}.
Problems of this form include the \emph{canonical polyadic decomposition (CPD)-constrained linear systems} studied in~\cite{bousse_linear_2018}, where the authors give an algebraic algorithm when $\mathcal X$ is a variety of rank-1 tensors and the constrained system has a single unique solution.

In a different direction, recent work by Johnston, Lovitz, and Vijayaraghavan develops an algorithm~\cite[Algorithm~1]{johnston_computing_2023} that can be used to find multiple intersections between a linear subspace $\mathcal U$ and a conic variety $\mathcal X$, corresponding to a VC-LS with $b=0$.
We observe that the machinery of this intersection algorithm greatly clarifies the general principles in Bouss\'e et al.'s algebraic algorithm and provides a method for dealing with multiple solutions.
In Algorithm~\ref{alg:algorithm1-prime} below, we give a method for solving VC-LSs, synthesizing~\cite{bousse_linear_2018} and~\cite{johnston_computing_2023}. This can potentially handle multiple solutions, although there will only be a single solution in our use cases. We note that the algorithm is not always guaranteed to succeed, and this will be discussed later.

\begin{algorithm}[H]
\caption{
Solve a variety-constrained linear system. \\
\textbf{Input:} Matrix $A \in \K^{m\times n}$, vector $b\in\K^m$, and irreducible conic variety $\mathcal X \subseteq \K^n$. \\
\textbf{Output:} Solution(s) to variety-constrained linear system $Au = b$ subject to $u\in \mathcal X$, or ``Fail.''
}\label{alg:algorithm1-prime}
\begin{algorithmic}[1]
\Procedure{FindVarietyConstrainedSolutions}{$A, b, \mathcal X$}
\State Find the subspace $\mathcal U = \{ u\in \K^n : Au = \gamma b, \gamma \in \K\}$ of solutions to $Au=b$ up to scale.
\State \label{step:run-JLV} Run the subspace conic-variety intersection algorithm (\cite[Algorithm 1]{johnston_computing_2023}) to find the vectors \[
v_1,v_2,\dots, v_{s'} \in \mathcal U \cap \mathcal X
\] (or report `Fail' if the algorithm fails).
\ForAll{$v_k$} 
    \State $t \gets \langle b, Av_k\rangle/\|b\|^{2}_2$ 
    \If{$t\neq 0$} 
    \State $u_k \gets \frac1{t}v_k$
    \Else
    \State discard $u_k$
    \EndIf
\EndFor
\State \Return $\{u_1,u_2,\dots, u_s\}$
\EndProcedure
\end{algorithmic}
\end{algorithm}

Johnston, Lovitz, and Vijayaraghavan prove that their intersection algorithm succeeds on a certain class of generic inputs with planted solutions, provided the dimension of $\mathcal U$ is not too large.
In the following sections we will be constructing specialized VC-LS problems, which are not covered by the existing analysis. Towards analyzing these, we next formulate more explicit sufficient conditions for Algorithm~\ref{alg:algorithm1-prime} to recover a unique solution to a VC-LS.
We ask the following:
\begin{quote}
    Given a particular $A\in\K^{m\times n}$ and $u\in \mathcal X \subseteq \K^m$, and defining $b := Au$, what are explicit conditions on $A$, $u$, $\mathcal X$ to guarantee that Algorithm~\ref{alg:algorithm1-prime} recovers $u$?
\end{quote}

In this context, we will call $u\in\mathcal X$ the \emph{planted solution} to the VC-LS.
First, observe that if $b=0$ then we are finding $\ker A \cap \mathcal X$ so the intersection algorithm of JLV may be used without modification and solutions are only unique up to scaling.
We will then only consider the case where $u\notin \ker A$ so that $b\neq 0$ and it is possible for a solution to be unique, not just unique up to scale.
We will also restrict ourselves here to varieties $\mathcal X \subseteq \K^n$ cut out by degree-2 homogeneous polynomials and no linear functions so we may consider $\Span\{p_1,p_2,\dots,p_k\} \subseteq (\K^n)^{\oasterisk 2}$.
Varieties of this type include the variety of rank-at-most-1 tensors of any shape and include the the variety of symmetric matrices of rank-at-most-1 which is our primary use case.
We now define a construction which may be used to certify unique recovery.
To do this, we first extend our definition for the symmetric product from matrices and vectors to subspaces.
\begin{definition}
    For subspaces $\mathcal A, \mathcal B \subseteq \K^n$, define
    \[
        \mathcal A^{\oasterisk k} = \Span\{ u^{\oasterisk k} : u\in \mathcal A\} \subseteq (\K^n)^{\oasterisk k}
    \]
    and 
    \[
        \mathcal A \oasterisk \mathcal B = \Span\{ u\oasterisk v : u\in \mathcal A, v\in\mathcal B\} \subseteq (\K^n)^{\oasterisk 2}.
    \]
\end{definition}

Next, we generalize a construction used in prior work by the authors with Dastidar \cite{dastidar_improving_2025} to improve a bound the maximum dimension of a subspace in which one can find planted rank-1 matrices.

\begin{definition}
    \label{def:equiv-M-matrix-set}
    Define a set of matrices denoted by $\mathcal M(A,u, \mathcal X)$ associated to the variety-constrained linear system with planted solution $u$ where
    \[
        M \in \mathcal M(A,u, \mathcal X)
    \]
    if we may write $M = \Phi U$ for $\Phi$ a matrix whose rows form a basis for $\Span\{p_1,p_2,\dots,p_k\} \subseteq (\K^n)^{\oasterisk 2}$ where $p_1,p_2,\dots,p_k$ cut out the irreducible conic variety $\mathcal X$, and $U$ a matrix whose columns form a basis for $(\Span(u) \oasterisk \ker A) + (\ker A)^{\oasterisk 2} \subseteq (\K^n)^{\oasterisk 2}$.
\end{definition}

Note that each matrix in $\mathcal M(A,u, \mathcal X)$ is related by \emph{matrix equivalence} in that if $M,M' \in \mathcal M(A,u, \mathcal X)$ then there exist invertible matrices $P$ and $Q$ such that $P^{-1}M'Q = M$.

\begin{proposition}
    \label{prop:full-rank-unique-recovery}
    If $M \in \mathcal M(A,u, \mathcal X)$ has full column-rank then Algorithm~\ref{alg:algorithm1-prime} recovers the planted vector $u$ given inputs $A,b,\mathcal X$, where $b = Au$.
\end{proposition}

Before we prove this, observe that for a representative $M$ to have full column-rank, it must have at least as many rows as columns.
This gives a necessary condition that may be used to determine when the kernel of $A$ is too large. 
If $L=\dim(\ker A)$ and $k=\dim(\Span\{p_1,\dots, p_k\})$, then $M$ has $k$ rows and  $\mchoose{L}{2} - 1$ columns.
It is then necessary that
\begin{equation}
    \label{eq:M-matrix-dimension-bound}
    L \leq \sqrt{2(k+1) + \frac14} - \frac12,
\end{equation}
or else $M$ has fewer rows than columns and necessarily has a non-trivial nullspace.
The bound we give in~\eqref{eq:M-matrix-dimension-bound} is a tightened version of the necessary condition given in~\cite[Section 3.1.2]{bousse_linear_2018} which is based off counting the number of entries in second-compound matrices (a matrix containing all $2\times 2$ minors).
The use of second-compound matrices is a convenient tool but will produce linearly dependent rows in certain cases.
See Section~\ref{sec:analysis-of-aided-recovery} for more discussion.

In general, the bound~\eqref{eq:M-matrix-dimension-bound} indicates that the more degree-2 homogeneous polynomials that cut out $\mathcal X$, the larger the kernel of $A$ may be while allowing for the possibility of unique recovery.
Now we give a proof of the proposition.

\begin{proof}[Proof of Proposition~\ref{prop:full-rank-unique-recovery}]
    First, the bulk of our work will be showing that the intersection algorithm in Step~\ref{step:run-JLV} will return a single vector $v_1$ in the direction of $u$.
    Then we show that our rescaling produces a vector $u_1$ such that $Au_1=b$ ensuring that $u_1 = u$.
    
    We will show by way of contraposition that if $M\in\mathcal M(A,u,\mathcal X)$ has full column-rank then for the subspace $\mathcal U = \{ u\in \K^n : Au = \gamma b, \gamma \in \K\}$,
    \begin{equation*}
        \Span\{p_1,p_2,\dots,p_k\}^\perp \cap \U^{\oasterisk 2} = \Span(u^{\oasterisk 2}).
    \end{equation*}

    Suppose that $\Span\{p_1,p_2,\dots,p_k\}^\perp \cap \U^{\oasterisk 2} \neq \Span(u^{\oasterisk 2})$.
    Now since $u$ was planted, $u\in \mathcal X$ and $u \in \U$, $\Span\{p_1,p_2,\dots,p_k\}^\perp \cap \U^{\oasterisk 2} \supseteq \Span(u^{\oasterisk 2})$
    so it must be that $\Span\{p_1,p_2,\dots,p_k\}^\perp \cap \U^{\oasterisk 2} \supsetneq \Span(u^{\oasterisk 2})$.
    Suppose then there is some $w\in (\K^n)^{\oasterisk 2}$ such that $w\in \Span\{p_1,p_2,\dots,p_k\}^\perp$, $w\in \U^{\oasterisk 2}$ and $w\notin \Span(u^{\oasterisk 2})$.
    First,  if $\{n_1,\dots, n_L\}$ is a basis for $\ker A$, then $\{u, n_1,\dots, n_L\}$ is a basis for $\mathcal U$.
    
    Then $\{u^{\oasterisk 2}\} \cup \{u\oasterisk n_j: j\in[L]\} \cup \{n_j\oasterisk n_k: 1\leq j\leq k \leq L\}$ is a basis for $\mathcal U^{\oasterisk 2}$ and we may write 
    \[
    w = c_{\{0,0\}} u^{\oasterisk 2} + \sum_{j=1}^L c_{\{0,j\}} (u\oasterisk n_j) + \sum_{1\leq j\leq k \leq L} c_{\{j,k\}} (n_j\oasterisk n_k)
    \]
    for coefficients $c_{\{j,k\}}\in\K$, $0\leq j\leq k\leq L$.
    Now it must be that $c_{\{j,k\}} \neq 0$ for some $\{j,k\}\neq \{0,0\}$, otherwise $w = c_{\{0,0\}} (u^{\oasterisk 2})^{\oasterisk 2} \in \Span((u^{\oasterisk 2})^{\oasterisk2})$.
    Next, observe that if $\Phi$ is a matrix with rows a basis for $\Span\{p_1,p_2,\dots,p_k\}$ then $\Phi w=\vec 0$.
    Since $\Phi u^{\oasterisk 2} = \vec 0$ we have
    \begin{align*}
         \vec 0 = \Phi w &= c_{\{0,0\}} \Phi u^{\oasterisk 2} + \Phi\left(\sum_{j=1}^L c_{\{0,j\}} (u\oasterisk n_j) + \sum_{1\leq j\leq k \leq L} c_{\{j,k\}}(n_j\oasterisk n_k) \right) \\
        &= \Phi\left(\sum_{j=1}^L c_{\{0,j\}} (u\oasterisk n_j) + \sum_{1\leq j\leq k \leq L} c_{\{j,k\}} (n_j\oasterisk n_k)\right).
    \end{align*}
    The set $\{u\oasterisk n_j: j\in[L]\} \cup \{n_j\oasterisk n_k: 1\leq j\leq k \leq L\}$ is a basis for the subspace $(\Span(u)\oasterisk \ker A) \cup (\ker A)^{\oasterisk 2}$.
    Collecting these vectors as columns of a matrix $N$ and collecting coefficients into a column vector $c = \left(c_{\{j,k\}}:  0\leq j\leq k\leq L, \{j,k\}\neq \{0,0\}\right)$, we have 
    \[
    \vec 0 = \Phi Nc.
    \]
    Observe that $\Phi N \in \mathcal M(A,u,\mathcal X)$ and by the previous observation that at least one $c_{\{j,k\}} \neq 0$, the vector $c \neq \vec 0$ so $\Phi N$ does not have full column-rank.
    Since column-rank is preserved by matrix equivalence, every matrix $M\in \mathcal M(A,u,\mathcal X)$ lacks full column-rank.

    In its second step, the intersection algorithm of JLV runs simultaneous diagonalization on a basis for the subspace $\Span\{p_1,p_2,\dots,p_k\}^\perp \cap \U^{\oasterisk 2}$ which we have shown is equal to $\Span(u^{\oasterisk 2})$ under the assumption $M$ has full column-rank.
    However, reshaping any $w\in\Span(u^{\oasterisk 2})$ to an $n\times n\times 1$ tensor and
    running simultaneous diagonalization in this special case degenerates to the rank-1 factorization of the $n\times n$ matrix which has the form $\alpha uu^\top$ for some $\alpha\in\K$.
    It follows that Step~\ref{step:run-JLV} recovers a vector in the direction of $u$.

    So the intersection algorithm of JLV in step~\ref{step:run-JLV} will output a vector $v_1$ that is a nonzero scalar multiple of $u$, i.e.\ $v_1 = \gamma u$ for some $\gamma\in\K$.
    Finally, we can check that $u_1=v_1/t$ for $t=\langle b, Au_1\rangle /\|b\|_2^2$ obeys $Av_1=b$:
    \[
        Av_1 = \frac1 tAu_1 = \frac{\|b\|_2^2}{\langle b, Au_1\rangle}Au_1 = \frac{\|b\|_2^2}{\langle b, \gamma Au\rangle}\gamma Au = \frac{\|b\|_2^2}{\langle b, b\rangle}b = b.
    \]
    It follows that $u_1 = u$, meaning Algorithm~\ref{alg:algorithm1-prime} outputs the correct planted solution.
\end{proof}

\subsection{Aided Recovery}
\label{subsec:aided-recovery}
In this section, we will show how one may set up a variety-constrained linear system using either dMRA or pMRA invariant tensors such that the symmetric lift of the odd Fourier coefficients $(\hat x[E^c])^{\oasterisk 2}$ is the planted solution (or more precisely these coefficients after some group action).
First, we note that second order invariants are insufficient to recover the signal in the dMRA model and hence the pMRA model as well.
\begin{proposition}[\cite{bendory_dihedral_2022, edidin_reflection-invariant_2026}]
    \label{prop:power-spectrum-dMRA}
    For $x\in \Rl^d$, the second order invariant $T^{(2)}_\dMRA(\hat x)$ only determines $|\hat x[j]|$ for $j\in [d)$ (the power spectrum of $x$).
\end{proposition}
Proof of this follows almost directly from Proposition~\ref{prop:invariant-tensor-structure}.
We show an analogous result holds for the pMRA model.
\begin{proposition}
    \label{prop:power-spectrum-pMRA}
    For $x\in \Rl^{2d}$, the second order invariant $T^{(2)}_\pMRA(\hat x)$ only determines $|\hat x[j]|$ for $j\in [2d)\setminus\{d\}$ (the power spectrum of $x$ minus the Nyquist component).
\end{proposition}
\noindent See Section~\ref{subsec:invariant-tensor-structure} for the proof, which follows almost directly from Proposition~\ref{prop:invariant-tensor-structure}.

Since in neither model can we recover the phases of the Fourier coefficients from the second order invariants, we can examine third order invariants.
Consider for $T=T^{(3)}_{\dMRA}(\hat x)$ or $T=T^{(3)}_{\dMRA}(\hat x)$ the following four (partially) symmetric subtensors, recalling that $E$ denotes the even Fourier coefficients:
\begin{equation*}
    \begin{array}{cccc}
       T[E,E,E], & T[E^c,E,E], & T[E^c,E^c,E], & \text{and }T[E^c,E^c,E^c].
    \end{array}
\end{equation*}
Since $T$ is symmetric and $E$ and $E^c$ are a partition of $[2d)$, these subtensors represent a disjoint partition of the invariant polynomials that are entries of $T$.
The first subtensor is the subject of Proposition~\ref{prop:recursive-invariants} where we give the conditions under which it is itself an invariant subtensor of the same type as $T$.
One might hope that the fourth subtensor might yield some useful information about $\hat x[E^c]$ in a similar manner but this is not the case.
Instead, $T[E^c,E^c,E]$ contains the remaining non-zero invariant polynomials.

\begin{proposition}
    \label{prop:zero-subtensors}
    For $x\in\Rl^{2d}$ and $T=T^{(3)}_{\dMRA}(\hat x)$ or $T=T^{(3)}_{\pMRA}(\hat x)$, 
    \[
    T[E^c,E,E]=0\qquad \text{and} \qquad T[E^c,E^c,E^c]=0.
    \]
\end{proposition}
\begin{proof}
    This follows relatively directly from Proposition~\ref{prop:invariant-tensor-structure} as well.
    Observe that entries in $T[E,E,E^c]$ are indexed by two even indices and one odd, so any sum or difference of the three will be odd. 
    Similarly, $T[E^c,E^c,E^c]$ is indexed by three odd indices, so any sum or difference of the three will be odd.
    Consider the indicator variables in the formulas found in Proposition~\ref{prop:invariant-tensor-structure} for $r=3$, namely
    \[
        \bb1\left\{ j_1 + j_2 +j_3 \equiv 0 \hspace{-2mm} \pmod{2d}\right\}\qquad \text{and} \qquad \bb1\left\{ s_1j_1 + s_2j_2 +s_3j_3 \equiv 0 \hspace{-2mm} \pmod{2d}\right\}.
    \]
    By the parity argument, the first can never be satisfied, and similarly the second cannot since the sign variables $s_1,s_2,s_3\in\{-1,1\}$ do not alter parity.
    Hence, these subtensors must be identically zero.
\end{proof}
This raises the question: 
\begin{quote}
    If we can recover $y\in\Cx^d$ in the $\dihgroup{d}$-orbit of $\hat x[E]$ from the subtensor $T[E,E,E]$, is the remaining non-zero subtensor $T[E^c,E^c, E]$ enough to extend $y$ to a vector $x'$ in the $\dihgroup{2d}$-orbit of $\hat x$?
\end{quote}

We will see the answer is yes. Observe that by similar parity arguments to those used in Proposition~\ref{prop:zero-subtensors}, the entries of $T[E^c, E^c, E]$ are bihomogeneous polynomials, being degree-2 in the entries of $\hat x[E^c]$ and degree-1 in the entries of $\hat x[E]$.
Following this observation, we may set up a VC-LS problem with a planted solution providing such an extension.

\paragraph{Overarching Strategy.} Suppose that $x\in \K^m$ is some signal, $G$ is a finite group with some action on $\K^m$, $h:\K^m \to \K^n$ is a linear transformation, $R\subseteq [m]$, $S \subseteq[n]$ are sets of indices, $\tilde G$ is a finite group with some action on $\K^{|R|}$, $\tilde h: \K^{|R|} \to \K^{|S|}$ is also a linear transformation such that the conditions of Lemma~\ref{lemma:invariant-subtensors} hold and we may extract an invariant subtensor
\[
    \left(T_{G,h}^{(3)}(x)\right)[S,S,S] = \frac{k|\tilde G|}{|G|} T_{\tilde G, \tilde h}^{(3)}(x[R]).
\]
Next, define
\[
\mathcal X_1^\vee = \{x^{\oasterisk 2} : x\in \K^{|R^c|}\} \subseteq \left(\K^{|R^c|}\right)^{\oasterisk 2}
\]
to be the variety of symmetric $|R^c| \times |R^c|$ rank-1 matrices.

The following strategy may be applied recursively to the orbit recovery problem:
\begin{enumerate}
    \item Find a representative $y$ of the orbit of $x[R]$ under $\tilde G$ by solving the orbit recovery problem associated with $T_{\tilde G, \tilde h}^{(3)}(x)$.
    \item \label{item:partial-operator-condition} 
    Define \[
    b:= \left(T_{G,h}^{(3)}(x)\right)[S^c, S^c, S]\in \K^{\mchoose{|S^c|}{2}|S|}
    \] to be the vectorization of a particular subtensor of $T_{G,h}^{(3)}(x)$.
    Next, construct a linear operator using $y$ so that $A:=A(y)\in \K^{\mchoose{|R^c|}{2}\times \mchoose{|S^c|}{2}|S|}$ and for some $x'$ in the $G$-orbit of $x$ with component $z=x'[R^c]$,
    \begin{equation}
        \label{eq:step2-variety-contrained-system}
    A (z^{\oasterisk 2}) =  b.
    \end{equation}
    \item \label{item:unique-solution-condition} Solve the variety-constrained linear system constructed in \eqref{eq:step2-variety-contrained-system}, given generally as
    \begin{equation}
        \label{eq:step3-abstraction}
        Au= b \qquad \text{subject to} \qquad u\in \mathcal X^\vee_1,
    \end{equation}
    using Algorithm~\ref{alg:algorithm1-prime}.
\end{enumerate}
We call this procedure \emph{aided recovery} since a specific representative $y$ is used to construct $A$ and aids in the recovery of the rest of the signal.
This ensures that only particular representatives of $x$'s orbit satisfy \eqref{eq:step3-abstraction}.
The subsequent portions of this section show how to construct a linear operator as suggested in Step~\ref{item:partial-operator-condition} for the tensors $ T^{(3)}_\dMRA(\hat{x})$ and $T^{(3)}_\pMRA(\hat{x})$ with a planted solution as given in~$\eqref{eq:step2-variety-contrained-system}$.

\paragraph{Aided Recovery for Dihedral MRA.}

\begin{definition}
    \label{def:dMRA-recursive-linear-system}
    For $y \in \Cx^d$, define the linear operator $A_\dMRA(y):(\Cx^{d})^{\oasterisk 2} \to (\Cx^{d})^{\oasterisk 2} \otimes \Cx^d$ column-wise by
    \begin{equation}
        \label{eq:columns-A-dMRA}
        A_\dMRA[:,\mset{j,k}] = \frac{\sqrt{3}}{2} y[-(j+k+1)] ( \vec{e}_j\oasterisk \vec{e}_k \otimes \vec{e}_{-j-k-1}  +  \vec{e}_{-j-1}\oasterisk \vec{e}_{-k-1} \otimes \vec{e}_{j+k+1})
    \end{equation}
    for $\mset{j,k}$ a multiset of $[d)$. \\
    For $\hat{x}\in\Cx^{2d}$, define a vector $b_\dMRA(\hat x) := b_\dMRA\in (\Cx^{d})^{\oasterisk 2} \otimes \Cx^d$ given entry-wise by
    \begin{equation}
        \label{eq:define-b-dMRA}
        b_\dMRA[(\mset{j,k},\mset{m})] = \left(T_\dMRA^{(3)}(\hat x)\right)[\mset{2j+1,2k+1,2m}]
    \end{equation}
    where $0\leq j\leq k< d$ and $0\leq m<d$.
\end{definition}
The vector $b_\dMRA$ is a particular vectorization of $\left(T_\dMRA^{(3)}(\hat x)\right)[E^c, E^c, E]$, which is a partially symmetric tensor with symmetry along the first two modes.
\begin{proposition}
    \label{prop:dMRA-variety-constrained-linear-system}
    Let $x\in \Rl^{2d}$ with Fourier transform $\hat x \in \Cx^{2d}$, and let $y\in \Cx^d$ belong to the $\dihgroup{d}$-orbit of $\hat x[E]$.
    For $A=A_\dMRA(y)$ and $b=b_\dMRA(\hat x)$, there is some $g\in \dihgroup{2d}$ such that $(g\cdot \hat{x})|_{E} = y$ and the vector $z = (g\cdot \hat x)|_{E^c}$ satisfies $A z^{\oasterisk 2} = b$.
\end{proposition}
\begin{remark}
    There are generically two $g_1,g_2\in\dihgroup{2d}$ that meet these conditions and produce vectors $z$ that are negations of each other.
\end{remark}
The crux is that the VC-LS
\[
    (A_\dMRA) u =b_\dMRA \quad \text{subject to} \quad u \in \mathcal X_1^\vee
\]
contains a planted solution which may be used to extend $y$ to an element in the $\dihgroup{2d}$-orbit of $\hat x$.
Applying such extensions recursively gives us our strategy for orbit recovery of signals $x\in\Rl^{2^{k+1}}$ for $k\ge 2$.

Since much of our indexing is modulo $d$ we will also find it helpful to say that two multisets of the same size, $M$ and $M'$, are equivalent modulo $d$ denoted by $M \equiv M' \pmod{d}$ when there exists a bijection $\pi : M \to M'$ such that for every $m \in M$, $m \equiv \pi(m) \pmod{d}$.

\begin{proof}[Proof of Proposition~\ref{prop:dMRA-variety-constrained-linear-system}]
    To begin let $x \in \Rl^{2d}$, $\hat x \in \Cx^{2d}$ be $x$'s Fourier coefficients. We first consider the special case $y = \hat x[E]$. Using the shorthand $A=A_\dMRA(\hat x[E])$ and $b=b_\dMRA(\hat x)$, we will show that the choice $g = \Id$, i.e.\ $z=\hat x[E^c]$, satisfies 
    \begin{equation}
        \label{eq:check-A-dMRA-planted-solution}
        Az^{\oasterisk 2} = b.
    \end{equation}
    Then, more generally, we will show that this holds over the orbit of $\hat x[E]$ for a corresponding transformation on $\hat x[E^c]$.

    First, observe that by Definition~\ref{def:symmetrized-Kronecker-product}, for $j',k'\in[d)$,
    \begin{equation}
        \label{eq:entries-of-symmetric-product-odd-coefs}
        (z^{\oasterisk 2})[\mset{j',k'}] = w_{\mset{j',k'}} \hat{x}[2j'+1]\hat{x}[2k'+1]
    \end{equation}
    and by Definition~\ref{def:dMRA-recursive-linear-system}, for $j,k,m\in[d)$ we may find the entries in column of $A$ using indicator functions,
    \begin{equation}
        \label{eq:entries-of-A-dMRA}
        \begin{split}
            A[(\mset{j,k}, \mset{m}), \mset{j',k'}] &= \frac{\sqrt{3}}{2} x[-2(j'+k'+1)] \times \\
        (&\bb1\{ \mset{j,k} \equiv \mset{j',k'}\pmod{d}\}\cdot \bb1\{m\equiv -j'-k'-1\pmod{d}\}
        \\
        &+ \bb1\{ \mset{-j-1,-k-1} \equiv \mset{j',k'} \pmod{d}\} \cdot \bb1\{ m\equiv j'+k'+1\pmod{d}\}).
        \end{split}
    \end{equation}
    So then entries of $Az^{\oasterisk 2}$ --- the left-hand side of \eqref{eq:check-A-dMRA-planted-solution} --- have the form
    \begin{align*}
        (Az^{\oasterisk 2})[(\mset{j,k}, \mset{m})] &= A[(\mset{j,k}, \mset{m}), :] z^{\oasterisk 2} \\
        &= \sum_{0\leq j'\leq k'< d} A[(\mset{j,k}, \mset{m}), \mset{j',k'}] (z^{\oasterisk 2})[\mset{j',k'}].
    \end{align*}
    Using \eqref{eq:entries-of-symmetric-product-odd-coefs} and \eqref{eq:entries-of-A-dMRA}, we simplify this to
    \begin{align*}
        (Az^{\oasterisk 2})[(\mset{j,k}, \mset{m})] =& \bb1\{m\equiv -j-k-1 \pmod{d}\} \cdot \left(\frac{\sqrt{3}w_{\mset{j,k}}}{2}  \hat{x}[2j+1]\hat{x}[2k+1]\hat{x}[2(-j-k-1)] \right.\\
        & \left.+\frac{\sqrt{3} w_{\mset{-j-1,-k-1}} }{2} \hat{x}[-2j-1]\hat{x}[-2k-1]\hat{x}[2(j+k+1)]\right).
    \end{align*}
    Now $w_{\mset{j,k}} = w_{\mset{-j-1,-k-1}}$ since if $j=k$ then $-j-1=-k-1$.
    We may further simplify to get
    \begin{equation}
        \label{eq:entries-of-LHS}
        \begin{split}
            (A(\hat x[E^c])^{\oasterisk 2})[(\mset{j,k},\mset{m})] =& \frac{\sqrt{3} w_{\mset{j,k}}}{2}  \bb1{\{m\equiv -j-k-1 \pmod{d}\}} \times \\ 
        & \hspace{-25pt} (\hat{x}[2j+1]\hat{x}[2k+1]\hat{x}[-2(j+k+1)] + \hat{x}[-2j-1]\hat{x}[-2k-1]\hat{x}[2(j+k+1)]).
        \end{split}
    \end{equation}
    Next, we examine the entries of $b$ --- the right-hand side of \eqref{eq:check-A-dMRA-planted-solution} --- using Proposition~\ref{prop:invariant-tensor-structure}.
    We can expand the entries of the third order invariant indexed as
    \begin{equation}
        \label{eq:raw-entries-of-RHS}
        \begin{split}
            b[(\mset{j,k},\mset{m})] =& \left(T_\dMRA^{(3)}(\hat x)\right)[\mset{2j+1,2k+1,2m}] \\
        =&  \frac{w_{\mset{2j+1,2k+1,2m}}}2\bb1\{2j+1+2k+1+2m \equiv 0 \pmod{2d}\} \times\\
        &\ (\hat x[2j+1] \hat x[2k+1]\hat x[2m] + \hat x[-2j-1] \hat x[-2k-1]\hat x[-2m]).
        \end{split}
    \end{equation}
    The indicator conditions in \eqref{eq:entries-of-LHS} and \eqref{eq:raw-entries-of-RHS} are the same, i.e.,
    \[
    \bb1\{2j+1+2k+1+2m \equiv 0 \pmod{2d}\} =\bb1\{m \equiv -j-k-1 \pmod{d}\}.
    \]
    Next, by parity, $2m\not\equiv 2j+1\pmod{2d}$ and $2m\not\equiv 2k+1 \pmod{2d}$ for any assignment of indices so we may simplify $w_{\mset{2j+1,2k+1,2m}} = \sqrt{3}w_{\mset{2j+1,2k+1}}$.
    Similar to our observations before, $w_{\mset{2j+1,2k+1}} = w_{\mset{j,k}}$ so we get
    \begin{equation}
        \label{eq:entries-of-RHS}
        \begin{split}
            b[(\mset{j,k},\mset{m})] =& \frac{\sqrt{3} w_{\mset{j,k}}}{2}  \bb1{\{m\equiv -j-k-1 \pmod{d}\}} \times \\ 
        & (\hat{x}[2j+1]\hat{x}[2k+1]\hat{x}[-2(j+k+1)] + \hat{x}[-2j-1]\hat{x}[-2k-1]\hat{x}[2(j+k+1)]).
        \end{split}
    \end{equation}
    Comparing \eqref{eq:entries-of-LHS} to \eqref{eq:entries-of-RHS}, they are equal and so $A(\hat x[E^c])^{\oasterisk 2} = b$.

    Next, suppose that $y=\tilde g\cdot \hat x[E]$ for $\tilde g\in \dihgroup{d}$ and set $A'= A_\pMRA(y)$.
    The group element $\tilde g$ is one of $\{\fshiftel_{\ell,d}, \freflel\fshiftel_{\ell,d}\}_{\ell=0}^{d-1}$.
    If $\tilde g= \fshiftel_{\ell,d}$, set $g= \fshiftel_{\ell,2d}$ or $\fshiftel_{\ell+d,2d}$ and
    if $\tilde g= \freflel\fshiftel_{\ell,d}$, set $g= \freflel\fshiftel_{\ell,2d}$ or $\freflel\fshiftel_{\ell+d,2d}$.
    From the proof of Lemma~\ref{lemma:invariant-subtensors}, one can verify that these choices of group elements will ensure that $(\tilde g\cdot \hat x)|_{E} = y$.
    Set $z= (g\cdot \hat x)|_{E^c}$.
    We will show that $A'(z^{\oasterisk 2}) = A(\hat x[E^c])^{\oasterisk 2}$.
    First, observe that we may factor coefficients out of the columns of $A$ by defining a diagonal matrix $X \in \Cx^{\mchoose{d}{2} \times \mchoose{d}{2}}$,
    \[
        X = \diag\left(\hat x[-2(j+k+1)] : 0\leq j \leq k <d\right),
    \]
    and a matrix $B$ such that $B X = A$.
    Rows and columns of $X$ are indexed by multisets $\mset{j,k}$ of $[d)$ under the DM order, and columns of $B$ are
    \[
        B[:, \mset{j,k}] = \frac{\sqrt{3}}{2} ( \vec{e}_j\oasterisk \vec{e}_k \otimes \vec{e}_{-j-k-1}  +  \vec{e}_{-j-1}\oasterisk \vec{e}_{-k-1} \otimes \vec{e}_{j+k+1}).
    \]
    Observe that in the first case $\tilde g = \fshiftel_{\ell,d}$ for some $\ell\in[d)$,
    \begin{align*}
        A_\dMRA(\tilde g \cdot (\hat x[E])) = B  X \underbrace{\diag\left(e^{2\pi i(2(j+k+1))\ell/(2d)} : 0\leq j \leq k <d\right)}_{=: D_1} = A D_1.
    \end{align*}
    With a small simplification we have $D_1 = \diag\left(e^{2\pi i(j+k+1)\ell/d} : 0\leq j \leq k <d\right)$.
    Next, consider that for $\ell'\in[2d)$, if $z=(\fshiftel_{\ell',2d} \cdot \hat x)|_{E^c}$, then for $j,k\in[d)$,
    \begin{align*}
        z^{\oasterisk 2}[\mset{j,k}] &= w_{\mset{j,k}} e^{-2\pi i (2j+1)\ell'/(2d)} \hat x[2j+1] e^{-2\pi i (2k+1)\ell'/(2d)}\hat x[2k+1] \\
        &= e^{-2\pi i(2j+1 + 2k+1)\ell'/(2d)} w_{\mset{j,k}} \hat x[2j+1]\hat x[2k+1].
    \end{align*}
    So when $\ell'=\ell$ or $\ell' = \ell + d$,
    \[
        z^{\oasterisk 2} = \underbrace{\diag\left(e^{-2\pi i (j+k+1)\ell/d} : 0\leq j \leq k <d\right)}_{=: D_2}(x[E^c])^{\oasterisk 2}.
    \]
    Observe that the phases in $D_1$ and $D_2$ are conjugates so $D_1 D_2 = I$ and
    \[
    A'z^{\oasterisk 2} = A D_1 D_2 (x[E^c])^{\oasterisk 2} = A (x[E^c])^{\oasterisk 2} = b.
    \]
    In the second case, we set $\tilde g = \freflel_d\fshiftel_{\ell,d}$ for some $\ell\in[d)$.
    Let us consider the effect of the exchange matrix without phase shifting ($\ell=0$) first.
    Observe that while $(\freflel_{2d} \cdot \hat x)[E] = \freflel_d \cdot (\hat x[E])$,  $(\freflel_{2d} \cdot \hat x)[E^c] = \reflel_d \cdot (\hat x[E^c])$ where
    \[
        \reflel_d[j,j'] = \begin{cases}
        1 & j\equiv -j'-1 \pmod{d} \\
        0 & \text{otherwise}.
    \end{cases}
    \]
    In other words, the action of reflection around the zero index for a vector of length $2d$ (the action of $\freflel_{2d})$ has the effect of reversal $(\reflel_d)$ on the odd coefficients.
    Then $(J_d \cdot (\hat x[E^c]))^{\oasterisk 2}[\mset{j,k}] = (\hat x[E^c])^{\oasterisk 2}[\mset{-j-1,-k-1}]$ so we may write
    \[
        \reflel_d^{\oasterisk 2}[\mset{j,k}, \mset{j',k'}] = \begin{cases}
        1 & \mset{j,k}\equiv \mset{-j'-1,-k'-1} \pmod{d} \\
        0 & \text{otherwise}.
        \end{cases}
    \]
    Since $\reflel_d^{\oasterisk 2}$ swaps entries in the vector, it is an involution, i.e.\ $(\reflel_d^{\oasterisk 2})^2=I$.
    Next notice that
    \[
        \reflel_d^{\oasterisk 2}  X \reflel_d^{\oasterisk 2} = \diag\left(\hat x[-2((-j-1)+(-k-1)+1)] : 0\leq j \leq k <d\right) = \diag\left(\hat x[2(j+k+1)] : 0\leq j \leq k <d\right)
    \]
    so we have that $A_\dMRA( \freflel_d \cdot \hat x[E]) = B \reflel_d^{\oasterisk 2}  X \reflel_d^{\oasterisk 2}$.
    Also note that $B[:, \mset{j,k}] = B[:, \mset{-j-1,-k-1}]$ so $B \reflel_d^{\oasterisk 2} = B$.
    Then for any $\ell\in[d)$ with $\tilde g = \freflel_d\fshiftel_{\ell,d}$,
    \[
        A_\dMRA( \tilde g \cdot \hat x[E]) = B \reflel_d^{\oasterisk 2} X D_1 \reflel_d^{\oasterisk 2} = B X D_1(\reflel_d)^{\oasterisk 2} = A D_1(\reflel_d)^{\oasterisk 2}.
    \]
    Next, consider that for $\ell'= \ell$ or $\ell' = \ell + d$, if $g=\freflel_{2d}\fshiftel_{\ell',2d}$ and $z=(g \cdot \hat x)|_{E^c}$,
    \[
        z^{\oasterisk 2} = ((g \cdot \hat x)[E^c])^{\oasterisk 2} = \reflel_d^{\oasterisk 2} D_2 (\hat x[E^c])^{\oasterisk 2}.
    \]
    Combining these two expressions we have 
    \[
    A'z^{\oasterisk 2} = A D_1\reflel_d^{\oasterisk 2} \reflel_d^{\oasterisk 2} D_2(\hat x[E^c])^{\oasterisk 2} = A D_1 D_2(\hat x[E^c])^{\oasterisk 2} = A (\hat x[E^c])^{\oasterisk 2} = b. \qedhere
    \]
\end{proof}

\paragraph{Aided Recovery for MRA with Projection.} 

The relationship between the MRA invariants given in Lemma~\ref{lemma:pMRA-invariants-from-dMRA-invariants} implies something similar for subtensors of $T^{(3)}_\pMRA$.
We define the following shorthand for a particular restriction of $\fprojOp^{\oasterisk 3}$ along each mode of its domain and codomain:
\begin{equation}
    \label{eq:fourier-projection-suboperator}
\fprojOp^{(3)}_{E^c,E^c,E} := (\fprojOp_{2d}[E_d^c, E_{2d}^c])^{\oasterisk 2} \otimes \fprojOp_{2d}[E_{d}, E_{2d}].
\end{equation}
By construction this is a linear transformation from a tensor with partial symmetry in $(\Cx^d)^{\oasterisk 2} \otimes \Cx^d$ to a tensor with partial symmetry in $(\Cx^{\lfloor (d+1)/2 \rfloor})^{\oasterisk 2} \otimes \Cx^{\lfloor d/2 \rfloor}$.
\begin{definition}
    \label{def:pMRA-recursive-linear-system}
    For $y \in \Cx^d$, define the linear operator $A_\pMRA(y):(\Cx^{2d})^{\oasterisk 2} \to (\Cx^{\lfloor d/2 \rfloor})^{\oasterisk 2} \otimes \Cx^{\lfloor d/2 \rfloor}$ by
    \begin{equation}
        \label{eq:columns-A-pMRA}
        A_\pMRA(y) = \fprojOp^{(3)}_{E^c,E^c,E} A_\dMRA(y),
    \end{equation}
    with $\fprojOp^{(3)}_{E^c,E^c,E}$ given in~\eqref{eq:fourier-projection-suboperator} and $A_\dMRA$ given in Definition~\ref{def:dMRA-recursive-linear-system}. \\
    For $x\in\Cx^{2d}$, define a vector $b_\pMRA(\hat x) := b_\pMRA\in (\Cx^{\lfloor (d+1)/2 \rfloor})^{\oasterisk 2} \otimes \Cx^{\lfloor d/2 \rfloor}$ given entry-wise by
    \begin{equation}
        \label{eq:define-b-pMRA}
        b_\pMRA[(\mset{j,k},\mset{m})] = \left(T_\pMRA^{(3)}( x)\right)[\mset{2j+1,2k+1,2m}]
    \end{equation}
    where $0\leq j\leq k < d$ and $0\leq m<d$.
\end{definition}

\begin{proposition}
    \label{prop:pMRA-variety-constrained-linear-system}
    Let $x\in \Rl^{2d}\cap \Span(1,-1,\dots, 1,-1)^{\perp}$ with Fourier transform $\hat x \in \Cx^{2d}$, and let $y\in \Cx^d$ belong to the $\dihgroup{d}$-orbit of $\hat x[E]$.
    For $A=A_\pMRA(y)$ and $b=b_\pMRA(\hat x)$, there is some $g\in \dihgroup{2d}$ such that $(g\cdot \hat{x})|_{E} = y$ and the vector $z = (g\cdot \hat x)|_{E^c}$ satisfies $A z^{\oasterisk 2} = b$.
\end{proposition}
Again, the main point is that
\[
    (A_\pMRA) u =b_\pMRA \quad \text{subject to} \quad u \in \mathcal X_1^\vee
\]
also contains a planted solution which may be used to extend $y$ to an element in the $\dihgroup{2d}$-orbit of $\hat x$ and so solving this VC-LS can help us achieve orbit recovery.

The proposition follows almost directly from Proposition~\ref{prop:dMRA-variety-constrained-linear-system}, though we clarify some of the more complicated steps through smaller lemmas.
\begin{lemma}
    \label{lemma:inner-restriction-valid}
    For any $x\in \Cx^{2d}$,
    \[
        b_\pMRA = \fprojOp^{(3)}_{E^c,E^c,E} \left(T_\dMRA^{(3)}(\hat x)\right)[E_{2d}^c, E_{2d}^c, E_{2d}].
    \]
\end{lemma}
\begin{proof}
    First, observe that
    \[
    b_\pMRA=\left(T_\pMRA^{(3)}(\hat x)\right)[E_d^c, E_d^c, E_d] = \left(\fprojOp^{\oasterisk 3}T_\pMRA^{(3)}(\hat x)\right)[E_d^c, E_d^c, E_d]
    \]
    by Lemma~\ref{lemma:pMRA-invariants-from-dMRA-invariants}.
    Then by Lemma~\ref{lemma:restriction-across-products}, we may move the restriction to the rows of $\fprojOp^{\oasterisk 3}$ yielding
    \[
        b_\pMRA = \left((\fprojOp[E_d^c, :])^{\oasterisk 2}\otimes \fprojOp[E_d, :]\right)T_\pMRA^{(3)}(\hat x).
    \]
    Finally, notice that for any vector $v\in\Cx^{2d}$, $\fprojOp[E_d, :] v = \fprojOp[E_d, E_{2d}] v[E_{2d}]$ since by the sparsity structure of $\fprojOp$ in $\eqref{eq:fourier-projection-operation}$, restriction to even rows $\fprojOp[E_d, :]$ eliminates all non-zero entries in odd-indexed columns.
    The odd columns are all then zero so restriction to non-zero columns has no effect on the matrix-vector product.
    For the analogous reasons, $\fprojOp[E_d^c, :] v = \fprojOp[E_d^c, E_{2d}^c] v[E_{2d}^c]$.
    Finally, applying this gives
    \begin{align*}
        \left(T_\pMRA^{(3)}(\hat x)\right)[E_d^c, E_d^c, E_d] &= \left((\fprojOp[E_d^c, E_{2d}^c])^{\oasterisk 2}\otimes \fprojOp[E_d, E_{2d}]\right)\left(T_\pMRA^{(3)}(\hat x)\right)[E_{2d}^c, E_{2d}^c, E_{2d}] \\
        &= \fprojOp^{(3)}_{E^c,E^c,E}\left(T_\pMRA^{(3)}(\hat x)\right)[E_{2d}^c, E_{2d}^c, E_{2d}]. \qedhere
    \end{align*}
\end{proof}

\begin{proof}[Proof of Proposition~\ref{prop:pMRA-variety-constrained-linear-system}]
Suppose $x\in \Rl^{2d}$, let $\hat x \in \Cx^{2d}$ be its Fourier transform, and fix some $y\in \Cx^d$ in the $\dihgroup{d}$-orbit of $\hat x[E]$.
Then $y = \tilde g\cdot \hat x[E]$ for some $\tilde g\in \dihgroup{d}$.
Pick $g \in \dihgroup{2d}$ as in Proposition~\ref{prop:pMRA-variety-constrained-linear-system} and set $z=(g\cdot x)[E^c]$.
Then we know that $z$ satisfies
\[
A_\dMRA z^{\oasterisk 2} = b_\dMRA.
\]
But then we have
\begin{align*}
(\fprojOp^{(3)}_{E^c,E^c,E}A_\dMRA )z^{\oasterisk 2} &= \fprojOp^{(3)}_{E^c,E^c,E}b_\dMRA, \\
(A_\pMRA) z^{\oasterisk 2} &=b_\pMRA,
\end{align*}
where the left-hand side follows from the definition of $A_\pMRA$ and the right-hand side from Lemma~\ref{lemma:inner-restriction-valid}.
Hence, $z^{\oasterisk 2}$ is still a planted solution for the variety constrained linear system
\[
    (A_\pMRA) u =b_\pMRA \quad \text{subject to} \quad u \in \mathcal X_1^\vee. \qedhere
\]
\end{proof}
Careful readers may notice what appears to be an inconsistency: in the pMRA model we cannot recover the Nyquist component $\hat x[d]$ when $x\in\Rl^{2d}$.
It appears that we need this component in the construction of $A_\pMRA$.
We clarify this apparent issue.

\begin{lemma}
    \label{lemma:pMRA-zero-columns}
    When $d$ is even and $j,k\in[d)$ such that $j+k+1\equiv d/2 \equiv -d/2 \pmod{d}$,
    \[
        A_\pMRA[:,\mset{j,k}] = \vec{0},
    \]
    so $A_\pMRA(y)$ does not depend on the value of $y[d/2]$.
\end{lemma}
\begin{proof}
    When $d$ is even, the matrix $\fprojOp_{2d}[E_{d}, E_{2d}]$ used in the construction of  $\fprojOp^{(3)}_{E^c,E^c,E}$ in \eqref{eq:fourier-projection-suboperator} is simply the matrix $\fprojOp_{d}$.
From \eqref{eq:fourier-projection-operation} one can determine that $\fprojOp_{d}$ has a zero-column at index $d/2$, so $\fprojOp_{d}\vec e_{d/2} = \vec{0}$.
In Definition~\ref{def:dMRA-recursive-linear-system}, it would appear we need to know 
$y[d/2]$ to construction a column $A_\dMRA[:,\mset{j,k}]$ when $j+k+1\equiv d/2\equiv -d/2\pmod{d}$, which we use in the construction of $A_\pMRA$.
However, in this case the column has the form
\[
    A_\dMRA[:,\mset{j,k}] = \frac{\sqrt{3}}{2} y[d/2] ( \vec{e}_j\oasterisk \vec{e}_k \otimes \vec{e}_{d/2}  +  \vec{e}_{-j-1}\oasterisk \vec{e}_{-k-1} \otimes \vec{e}_{d/2}).
\]
Consider the partially symmetric standard basis tensors appearing above. If follows that
\begin{align*}
    \fprojOp^{(3)}_{E^c,E^c,E}( \vec{e}_j\oasterisk \vec{e}_k \otimes \vec{e}_{d/2}) &= (\fprojOp_{2d}[E_d^c, E_{2d}^c]^{\oasterisk 2}(\vec{e}_j\oasterisk \vec{e}_k))   \otimes (\fprojOp_{2d}[E_{d}, E_{2d}]\vec{e}_{d/2}) \\
    &= (\fprojOp_{2d}[E_d^c, E_{2d}^c]^{\oasterisk 2}(\vec{e}_j\oasterisk \vec{e}_k))   \otimes (\fprojOp_d\vec{e}_{d/2}) \\
    &= (\fprojOp_{2d}[E_d^c, E_{2d}^c]^{\oasterisk 2}(\vec{e}_j\oasterisk \vec{e}_k))   \otimes \vec{0} \\
    &= \vec{0}.
\end{align*}
Similarly $\fprojOp^{(3)}_{E^c,E^c,E}( \vec{e}_{-j-1}\oasterisk \vec{e}_{-k-1} \otimes \vec{e}_{d/2}) =\vec{0}$.
Thus, the corresponding columns in $A_\pMRA$ that would appear to use $y[d/2]$ are zero-columns.
\end{proof}
Observe that since $A_\pMRA$ has some zero-columns it will certainly have a non-trivial nullspace.
However, this does not mean the VC-LS is impossible to solve.

\subsection{Extending Representatives}
\label{subsec:extending-representatives}
We have now shown how the subtensor $T[E^c, E^c, E]$ can be used to produce variety-constrained linear systems that contain the remaining information about $\hat x$. We now state the full routine for extracting this information, in both the dihedral and projected MRA settings.

\begin{algorithm}[H]
\caption{
Aided recovery for dihedral MRA. \\
\textbf{Input:} A vector $y$ in the $\dihgroup{d}$-orbit of $\hat{x}[E]$ (even Fourier coeffs.\ of $x$) and the tensor $T_\dMRA^{(3)}(\hat x)$ from \eqref{eq:dMRA-invariant-tensor}. \\
\textbf{Output:} Either reports an element $w$ in the $\dihgroup{2d}$-orbit of $\hat x$ or ``Fail'' if no element was found.
}
\label{alg:extend-representative-dMRA}
\begin{algorithmic}[1]
\Procedure{ExtendRepresentativeDihedralMRA}{$y, T$}
\State $A \gets A_\dMRA(y)$ from Definition~\ref{def:dMRA-recursive-linear-system}
\State $b \gets T[E^c,E^c,E]$
\State $u \gets \Call{FindVarietyConstrainedSolutions}{A, b, \mathcal X^\vee_1}$ (see Algorithm~\ref{alg:algorithm1-prime})
\State Reshape $u$ to a complex symmetric matrix using \eqref{eq:matricization}: $U \gets \mathrm{matricize}(u)$
\State Find any $z$ such that $zz^\top=U$
\State $w[E] \gets y$
\State $w[E^c] \gets z$
\State \Return $w$
\EndProcedure
\end{algorithmic}
\end{algorithm}

\begin{algorithm}[H]
\caption{
Aided recovery for projected MRA. \\
\textbf{Input:} A vector $y$ in the $\dihgroup{d}$-orbit of $\hat{x}[E]$ and the tensor $T_\pMRA^{(3)}(\hat x)$ from \eqref{eq:pMRA-invariant-tensor}. \\
\textbf{Output:} Either reports an element $w$ in the $\dihgroup{2d}$-orbit of $\hat x$ or ``Fail'' if no element was found.
}
\label{alg:extend-representative-pMRA}
\begin{algorithmic}[1]
\Procedure{ExtendRepresentativeProjectedMRA}{$y, T$}
\State $A \gets A_\pMRA(y)$ from Definition~\ref{def:pMRA-recursive-linear-system}
\State $b \gets T[E^c,E^c,E]$
\State $u \gets \Call{FindVarietyConstrainedSolutions}{A, b, \mathcal X^\vee_1}$ (see Algorithm~\ref{alg:algorithm1-prime})
\State Reshape $u$ to a complex symmetric matrix using \eqref{eq:matricization}: $U \gets \mathrm{matricize}(u)$
\State Find any $z$ such that $zz^\top=U$
\State $w[E] \gets y$
\State $w[E^c] \gets z$
\State \Return $w$
\EndProcedure
\end{algorithmic}
\end{algorithm}

We note that matricizing $u$ produces a complex symmetric matrix, not a Hermitian one. To find $z$, one can simply take the first column of $U$ scaled by the appropriate complex scalar. There is a unique choice of $z$, up to negation.

The success of these algorithms hinge on the success of the \textsc{FindVarietyConstrainedSolutions} subroutine applied in this case.
In Section~\ref{sec:analysis-of-aided-recovery}, we analyze the two sets of equivalent matrices $\mathcal M(A, u, \mathcal X)$ corresponding to these two problems, to determine when these methods succeed.

\section{Analysis of the Aided Recovery Procedure}
\label{sec:analysis-of-aided-recovery}

In Section~\ref{subsec:variety-constrained-linear-systems} we repurposed the subspace--conic variety intersection algorithm from~\cite{johnston_computing_2023} via our modification in Algorithm~\ref{alg:algorithm1-prime} to solve VC-LS's.
We now apply our general construction from Definition~\ref{def:equiv-M-matrix-set} for the dMRA and pMRA models.

In Example~\ref{ex:variety-symmetric-rank-1 matrices} we introduced the variety of symmetric matrices of rank-at-most 1.
We now formally describe

\begin{definition}
    For the variety of symmetric rank-1 matrices $\mathcal X_1^\vee \subseteq (\Cx^d)^{\oasterisk 2}$, define 
    \[
    I(\mathcal X_1^\vee)_2 := \Span(v^{\oasterisk 2} : v\in\mathcal X_1^\vee)^\perp \subseteq ((\Cx^d)^{\oasterisk 2})^{\oasterisk 2}.
    \]
\end{definition}
$I(\mathcal X_1^\vee)_2$ is the subspace spanned by the degree-2 homogeneous polynomials cutting out $\mathcal X_1^\vee$.

\begin{proposition}[\cite{johnston_computing_2023}]
    $\dim(I(\mathcal X_1^\vee)_2) = \frac1{12}(d+1)d^2(d-1)$.
\end{proposition}
\begin{proof}
    This follows from the observations that $\Span(v^{\oasterisk 2} : v\in\mathcal X_1^\vee) \cong (\Cx^d)^{\oasterisk 4}$.
    Then
    \[
    \dim(I(\mathcal X_1^\vee)_2) = \binom{\binom{d+1}{2}+1}{2} - \binom{d+3}{4} = \frac1{12}(d+1)d^2(d-1).
    \]
\end{proof}

\begin{definition}
    \label{def:basis-of-symmetric-minors}
    Let $\mathcal P$ be the set of vectors
    \[
        p_{j_1j_2k_1k_2} =\frac{\vec e_{\mset{\mset{j_1,k_1},\mset{j_2,k_2}}}}{ w_{\mset{\mset{j_1,k_1},\mset{j_2,k_2}}} w_{\mset{j_1,k_1}}  w_{\mset{j_2,k_2}}} - \frac{\vec e_{\mset{\mset{j_1,k_2},\mset{j_2,k_1}}}}{ w_{\mset{\mset{j_1,k_2},\mset{j_2,k_1}}} w_{\mset{j_1,k_2}}  w_{\mset{j_2,k_1}}} \in ((\Cx^d)^{\oasterisk 2})^{\oasterisk 2}
    \]
    for $j_1,j_2,k_1,k_2 \in [d)$ such that $j_1\leq k_1$, $j_2\leq k_2$, $j_1 < j_2$ and $k_1 < k_2$, and $w_K$ the weight of a multiset $K$ defined in \eqref{eq:multiset-weight}.
\end{definition}
As discussion in Section~\ref{subsec:variety-constrained-linear-systems}, these basis vectors can be identified with the $2\times 2$ minors of a $d\times d$ symmetric matrix by taking $\mset{j_1,k_1}$ to be coordinates for the upper left entry and $\mset{j_2,k_2}$ be coordinates for the lower right entry of a $2\times 2$ submatrix whose determinant is being computed.
The weights on each partially symmetric standard basis tensor ensure correct normalization given our definition of symmetric lift in Definition~\ref{def:symmetrized-Kronecker-product}.
We show that the conditions on the indices ensure this is a set of linearly independent vectors.

\begin{proposition}
    \label{prop:basis-for-cutout}
    $\mathcal P$ is a basis for $I(\mathcal X_1^\vee)_2$.
\end{proposition}
See Section~\ref{subsec:polynomial-subspaces} for the proof.

\subsection{Sufficient Conditions for Aided Recovery in the dMRA Model}
\label{subsec:dMRA-sufficient-conditions}
Since this section and the next concern aided recovery in the context of recursion we will assume that $d$ is even.

Let
\begin{equation}
    \label{eq:dMRA-column-labels}
    C := \{\mset{j,k} : j,k\in[d), j + k + 1 < d\}
\end{equation}
be a set of multisets.
Treated as coordinates in the plane $(j,k)$ with $j\leq k$ for even $d$ these make up one fourth of the square $[d)\times[d)$, so $|C|=d^2/4$.
We use the Dershowitz--Manna order to construct a map $\beta':C\to [d^2/4]$ which maps a multiset in $C$ to its index when all are arranged in increasing order.
Next, we will define
\[
-C := \{\mset{d-1-j-1,d-1-k} : \{j,k\}\in C\}
\]
which we may think of as the reflection of the original multisets-as-coordinates in $C$ over the line $k = d-j$.
Define the set
\[
L_1 := \{\mset{j,k} : j,k\in[d), j+k+1=d\}.
\]
These three sets form a partition of the set of all two-element multisets of $[d)$, i.e.\ $C \sqcup -C \sqcup L_1 = \mchoose{[d)}{2}$.
Using this partition we extend our original ordering to a map from all multisets $\beta: \mchoose{[d)}{2} \to [d^2/4] \cup\{0\}$ by $\beta(C)=\beta(-C)=\beta'(C)$ and $\beta(L_1) = 0$ for any $j\in[d)$.
Finally define a function $\sgn : \mchoose{[d)}{2} \to \{-1,0,1\}$ by $\sgn(C) = 1$, $\sgn(-C)=-1$ and $\sgn(L_1) = 0$.

\begin{definition}
    \label{def:dMRA-M-matrix}
    For $\hat x\in \Cx^{2d}$ the matrix $M_\dMRA(\hat x)$ is defined to have the rows
    \begin{equation}
        \begin{split}
        (M_\dMRA)&[\mset{\mset{j_1, k_1},\mset{j_2,k_2}},:] =  \\
            &\begin{array}{rll}
               \quad \hat{x}[2j_1+1]\cdot\hat{x}[2k_1+1] \cdot\hat{x}[2(j_2 + k_2 + 1)] \cdot \sgn(\mset{j_2, k_2}) & \vec{e}_0 \oasterisk \vec{e}_{\beta(\mset{j_2, k_2})} \\
              -\: \hat{x}[2j_1+1]\cdot\hat{x}[2k_2+1]\cdot\hat{x}[2(j_2 + k_1 + 1)] \cdot \sgn(\mset{j_2, k_1}) & \vec{e}_0 \oasterisk \vec{e}_{\beta(\mset{j_2, k_1})} \\
              +\: \hat{x}[2j_2+1]\cdot\hat{x}[2k_2+1]\cdot\hat{x}[2(j_1 + k_1 + 1)] \cdot \sgn(\mset{j_1, k_1})& \vec{e}_0 \oasterisk \vec{e}_{\beta(\mset{j_1, k_1})} \\
              -\: \hat{x}[2j_2+1]\cdot\hat{x}[2k_1+1]\cdot\hat{x}[2(j_1 + k_2 + 1)] \cdot \sgn(\mset{j_1, k_2}) & \vec{e}_0 \oasterisk \vec{e}_{\beta(\mset{j_1, k_2})} \\
              +\: \hat{x}[2(j_1 + k_1 + 1)]\cdot\hat{x}[2(j_2 + k_2 + 1)] \cdot \sgn(\mset{j_1, k_1})\cdot \sgn(\mset{j_2, k_2})& \vec{e}_{\beta(\mset{j_1, k_1})} \oasterisk \vec{e}_{\beta(\mset{j_2, k_2})} \\
              -\: \hat{x}[2(j_1 + k_2 + 1)]\cdot\hat{x}[2(j_2 + k_1 + 1)]\cdot \sgn(\mset{j_1, k_2})\cdot \sgn(\mset{j_2, k_1}) & \vec{e}_{\beta(\mset{j_1, k_2})} \oasterisk \vec{e}_{\beta(\mset{j_2, k_1})}
        \end{array}
        \end{split}
    \end{equation}
    with $0\leq j_1 \leq k_1 < d$, $0\leq j_2 \leq k_2 < d$, $j_1 < j_2$ and $k_1 < k_2$.
    Indices are taken modulo $2d$ when accessing entries of $\hat{x}$ and the rowspace of $M$ should be considered a subspace of $(\Cx^{d^2/4+1})^{\oasterisk 2} \cap \Span(\vec e_0 \oasterisk \vec e_0)^\perp$.
\end{definition}

For $\hat x\in \Cx^{2d}$, $M_\dMRA(\hat x)$ has $\frac1{12}(d+1)d^2(d-1)$ rows and $\binom{ d^2/4 + 2}{2} - 1$ columns.
As a result, it has more rows than columns when $d \ge 4$.
It is also a very sparse matrix, having at most 6 non-zero entries per row.

\begin{proposition}
    \label{prop:M-dMRA-is-instance}
    For $\hat x\in \Cx^{2d}$ and any $g\in \dihgroup{2d}$,
    \[
        M_\dMRA(g\cdot \hat x) \in \mathcal M\left(A_\dMRA(\hat x[E]), (\hat x[E^c])^{\oasterisk 2}, \mathcal X_1^\vee\right)
    \]
    where $\mathcal M(A, u,\mathcal X)$ is the set of equivalent matrices in Definition~\ref{def:equiv-M-matrix-set} for a variety-constrained linear system $Au=b$ with planted solution $u\in\mathcal X$.
\end{proposition}

\begin{conjecture}
    \label{conj:full-rank-dMRA-M-matrix} For $d=2^k$ with $k \ge 2$, $M_\dMRA(\hat x)$ has full column-rank as a symbolic matrix of polynomials where the variables are the $d+1$ real and $d-1$ imaginary components of the Fourier coefficients of $x\in\Rl^{2d}$.
\end{conjecture}
See Section~\ref{subsec:numerical-evidence-conjectures} for numerical evidence that these matrices have full column-rank symbolically.
We will also drop the dependence on $\hat x$, replacing it with the parameter $d$, and writing $M_\dMRA(d)$, to emphasize when we are considering the construction as a symbolic matrix.

\begin{remark}[Symbolic Matrices]
    \label{rem:symbolic-real-and-imaginary-components}
    While we define $M_\dMRA(\hat x)$ in terms of Fourier coefficients of $x$, the signal $x$ is real so $\hat x[j]=\overline{\hat x[2d-j]}$ for $j\in[2d)$.
    Due to this conjugate relationship, $\hat x[j]$ and $\hat x[2d-j]$ cannot be considered distinct complex variables symbolically.
    Instead we consider symbolic variables $r_0,r_1,\dots,r_d$ and $c_1,c_2,\dots,c_{d-1}$ and define $\hat x$ by
    \begin{equation*}
    \begin{split}
        \Re(\hat x) &=: (r_0, r_1, r_2,\dots, r_{d-1}, r_d, r_{d-1},\dots,r_2, r_1) \quad \text{and}\\
        \Im(\hat x) &=: (0, c_1,c_2,\dots, c_{d-1}, 0, -c_{d-1}, -c_{d-2},\dots, -c_2, -c_1).
    \end{split}
    \end{equation*}
    Now $M_\dMRA(\hat x)$ is a matrix of symbolic polynomials in these underlying symbolic variables. Such a matrix is said to have full column-rank if some maximal minor evaluates to a nonzero polynomial.
\end{remark}

If $M_\dMRA(d)$ has full column-rank symbolically, this ensure the success of Algorithm~\ref{alg:extend-representative-dMRA} in retrieving the odd Fourier coefficients for generic $x\in\Rl^{2d}$.

\begin{theorem}[Aided recovery for dihedral MRA]
    \label{thm:dMRA-correct-extension}
    Fix even $d \ge 4$ and suppose $M_\dMRA(d)$ has full column-rank as a symbolic matrix.
    Then for generic $x\in\Rl^{2d}$, if $y$ is an arbitrary element in the $\dihgroup{d}$-orbit of $\hat x[E]$, Algorithm~\ref{alg:extend-representative-dMRA} runs in polynomial time on inputs $y$ and $T^{(3)}_\dMRA(\hat x)$ and recovers some $\hat x'$ in the $\dihgroup{2d}$-orbit of $\hat x$.
\end{theorem}
Note that the algorithm is guaranteed to succeed for a specific $\hat x$ whenever the instance $M_\dMRA(\hat{x})$ has full column-rank.

As one can see in Definition~\ref{def:equiv-M-matrix-set}, the construction of $M_\dMRA(\hat x)$ requires a basis for the kernel of $A_\dMRA(\hat x[E])$.
To prove Proposition~\ref{prop:M-dMRA-is-instance} (and subsequently Theorem~\ref{thm:dMRA-correct-extension}), we give an explicit and sparse basis for this subspace.
For $y\in \Cx^d$, define the vectors
\begin{equation}
    \label{eq:general-null-vectors}
    n_{\mset{j,k}}(y) := y[j+k+1]  \vec{e}_j \oasterisk \vec{e}_k - y[-j-k-1]  \vec{e}_{-k-1} \oasterisk \vec{e}_{-j-1}
\end{equation}
where indices on the standard basis vectors and indexing entries of $y$ are taken modulo $d$.
Define
\begin{equation}
        \label{eq:dMRA-nullset}
        N_\dMRA(y) := \{n_{m}(y): m\in C\} \subseteq (\Cx^d)^{\oasterisk 2}
\end{equation}
where $C$ is the set defined in \eqref{eq:dMRA-column-labels}.
\begin{lemma}
    \label{lemma:characterize-nullspace-dMRA}
    When $y\in\Cx^d$ with each $y[j]\neq 0$ for $j\in[d)$, $N_\dMRA(y)$ is a basis for $\ker A_\dMRA(y)$
\end{lemma}
The non-zero condition prevents $A_\dMRA(y)$ from having zero-columns.
Zero-columns would increase the dimension of the nullspace of $A_\dMRA$ which --- as we have noted --- does not necessarily break our method unless the bound in~\eqref{eq:M-matrix-dimension-bound} is exceeded.
However, generic signals will have non-zero Fourier coefficients so we do not consider these cases.

As suggested by the lemma, $C$, $\beta$ and $\sgn$ are helpful construct this specific basis.
Note that for any $j\in[d)$, $n_{\mset{j,d-j-1}}(y)$ is $\vec{0}$, hence it is an edge-case we handle along this anti-diagonal and the reason we split our set of multisets into $C$, $-C$ and $L_1$.

\begin{proof}[Proof of Lemma~\ref{lemma:characterize-nullspace-dMRA}]
    For some $y\in\Cx^d$ set $A=A_\dMRA(y)$.
    Consider the inner product of two columns of $A$ given by Definition~\ref{def:dMRA-recursive-linear-system}:
    \begin{align*}
        \langle A[:,\mset{j,k}], A[:,\mset{j',k'}] \rangle &= \frac{3}4 w^2_{\mset{j,k}} y[-(j+k+1)] \overline{ y[-(j'+k'+1)]} \\
        &\quad \langle ( \vec{e}_j\oasterisk \vec{e}_k \otimes \vec{e}_{-j-k-1}  +  \vec{e}_{-j-1}\oasterisk \vec{e}_{-k-1} \otimes \vec{e}_{j+k+1}), \\
        &\qquad ( \vec{e}_{j'}\oasterisk \vec{e}_{k'} \otimes \vec{e}_{-j'-k'-1}  +  \vec{e}_{-j'-1}\oasterisk \vec{e}_{-k'-1} \otimes \vec{e}_{j'+k'+1})\rangle.
    \end{align*}
    The columns will be orthogonal unless they have overlapping support.
    Observe that by Corollary~\ref{cor:symmetric-product-inner-product-expansion},
    \begin{align*}
        \langle \vec e_a \oasterisk \vec e_b \otimes \vec e_c, \vec e_{a'} \oasterisk \vec e_{b'} \otimes \vec e_{c'}\rangle &= \langle\vec e_a \oasterisk \vec e_b,  \vec e_{a'} \oasterisk \vec e_{b'} \rangle \langle\vec e_c, e_{c'} \rangle \\
        &= \frac12\left(\langle\vec e_a, e_{a'} \rangle \langle\vec e_b, e_{b'} \rangle + \langle\vec e_a, e_{b'} \rangle\langle\vec e_b, e_{a'} \rangle\right) \langle\vec e_c, e_{c'} \rangle \\
        &= \frac{|\mset{a,b}|}2  \bb1( \mset{a,b} = \mset{a',b'}) \bb1(c = c').
    \end{align*}
    By linearity of the inner product and since are indices are taken modulo $d$ we have
    \begin{align*}
        &\langle ( \vec{e}_j\oasterisk \vec{e}_k \otimes \vec{e}_{-j-k-1}  +  \vec{e}_{-j-1}\oasterisk \vec{e}_{-k-1} \otimes \vec{e}_{j+k+1}), (\vec{e}_{j'}\oasterisk \vec{e}_{k'} \otimes \vec{e}_{-j'-k'-1}  +  \vec{e}_{-j'-1}\oasterisk \vec{e}_{-k'-1} \otimes \vec{e}_{j'+k'+1})\rangle \\
        &= \left(\frac{|\mset{j,k}|}2  \bb1( \mset{j,k} \equiv \mset{j',k'} \pmod{d}) \bb1(-j-k-1 \equiv -j'-k'-1 \pmod{d})\right. \\
         &\quad + \frac{|\mset{j,k}|}2  \bb1( \mset{j,k} \equiv \mset{-j'-1,-k'-1} \pmod{d}) \bb1(-j-k-1 \equiv j'+k' +1 \pmod{d}) \\
        & \quad +\frac{|\mset{-j-1,-k-1}|}2  \bb1( \mset{-j-1,-k-1} \equiv \mset{j',k'}\pmod{d}) \bb1(j+k+1 \equiv -j'-k'-1\pmod{d}) \\
        & \left.\quad + \frac{|\mset{-j-1,-k-1}|}2  \bb1( \mset{-j-1,-k-1} \equiv \mset{-j'-1,-k'-1}\pmod{d}) \bb1(-j-k-1 \equiv j'+k' +1\pmod{d}) \right)\\
        &= |\mset{j,k}|  \bb1( \mset{j,k} \equiv \mset{j',k'} \pmod{d}) + |\mset{j,k}|  \bb1( \mset{j,k} \equiv \mset{-j'-1,-k'-1} \pmod{d}).
    \end{align*}
    Hence, columns are orthogonal unless $\mset{j,k} \equiv \mset{j',k'}$ (i.e.\ they are the same column) or $\mset{j,k} \equiv \mset{-j'-1,-k'-1}$.
    Notice that in the second case, these columns are scalar multiples of each other.
    We then have that
    \[
        y[j+k+1] \cdot A[:,\mset{j,k}] - y[-(j+k+1)] \cdot A[:,\mset{-j-1,-k-1}] = \vec 0.
    \]
    Hence, all null vectors take the form given in \eqref{eq:general-null-vectors}.
    Further, for $m_1,m_2\in C$ with $m_1\neq m_2$, $\langle n_{m_1}, n_{m_2}\rangle = 0$ since the null vectors have non-overlapping support.
    Hence, $N_\dMRA(y)$ is a orthogonal basis for $\ker A$.
\end{proof}

Having determined a sparse basis for the nullspace of $A_\dMRA(y)$ in terms of the entries of $y$ we may now show that $M_\dMRA(\hat x)$ lies in the equivalence class of matrices that determine whether the variety-constrained linear system has a unique planted solution recoverable by Algorithm~\ref{alg:extend-representative-dMRA}.

\begin{proof}[Proof of Proposition~\ref{prop:M-dMRA-is-instance}]
    We will proceed in two stages: first we will show that the matrix $M_\dMRA(\hat x) \in \mathcal M\left(A_\dMRA(\hat x[E]), (\hat x[E^c])^{\oasterisk 2}, \mathcal X_1^\vee\right)$, and then we will show that for $g\in\dihgroup{2d}$, $M_\dMRA(g\cdot \hat x) = W_g M_\dMRA(\hat x)$ where $W_g$ is some change-of-basis matrix acting on the rows of $M_\dMRA(\hat x)$ depending on $g$.
    
    Consider a matrix $N\in \Cx^{\mchoose{d}{2}\times |C|}$ with rows
    \[
        N[\mset{j,k}, :] = \sgn(\mset{j,k}) w_{\mset{j,k}} \hat{x}[2(j+k+1)] \vec{e}_{\beta(\mset{j,k})}.
    \]
    The columns of $N$ are the null vectors of $N_\dMRA$ scaled by the weight $w_{\mset{j,k}}$ so by Lemma~\ref{lemma:characterize-nullspace-dMRA}, $\operatorname{colspan} N = \ker A_\dMRA$.
    Next, we can expand the solution vector $\hat x[E^c]^{\oasterisk 2}$ by examining its entries,
    \[
        (\hat x[E^c]^{\oasterisk 2})[\mset{j,k}] = w_{\mset{2j+1,2k+1}}\hat x[2j+1]\hat x[2k+1] = w_{\mset{j,k}}\hat x[2j+1]\hat x[2k+1].
    \]
    If we form the $\Cx^{\mchoose{d}{2}\times (|C|+1)}$ block matrix
    \[
        U = \begin{bmatrix}
            (\hat x[E^c])^{\oasterisk 2} & N
        \end{bmatrix},
    \]
    where columns are indexed from $0$ to $|C|$, the rows have the form
    \[
        U[\mset{j,k}, :] = w_{\mset{j,k}}\left(\hat x[2j+1]\hat x[2k+1] \vec{e}_0 + s(\{j,k\}) \hat{x}[2(j+k+1)] \vec{e}_{\beta(\mset{j,k})}\right).
    \]
    Define the matrix $\tilde U$ to be $U^{\oasterisk 2}$ with the first column dropped.
    Using Lemma~\ref{lemma:symmetric-product-equivalences} we expand along the rows but exclude vectors in the direction $\vec{e}_0 \oasterisk \vec{e}_0$.
    This gives the three terms
    \begin{equation}
        \label{eq:lift-minus-planted}
        \begin{split}
            \tilde U[\mset{\mset{j_1,k_1},&\mset{j_2,k_2}}, :] = w_{\mset{\mset{j_1,k_1},\mset{j_2,k_2}}} w_{\mset{j_1,k_1}} w_{\mset{j_2,k_2}} \times \\
            &(\hat x[2j_1+1]\hat x[2k_1+1]  \hat{x}[2(j_2+k_2+1)] s(\mset{j_2,k_2})  \vec{e}_0 \oasterisk\vec{e}_{\beta(\mset{j_2,k_2})} \\
            &+ \:\hat x[2j_2+1]\hat x[2k_2+1]  \hat{x}[2(j_1+k_1+1)] s(\mset{j_1,k_1}) \vec{e}_0 \oasterisk\vec{e}_{\beta(\mset{j_1,k_1})} \\
            &+\:  \hat{x}[2(j_1+k_1+1)]  \hat{x}[2(j_2+k_2+1)]s(\mset{j_1,k_1})s(\mset{j_2,k_2})  \vec{e}_{\beta(\mset{j_1,k_1})} \oasterisk\vec{e}_{\beta(\mset{j_2,k_2})}).
        \end{split}
    \end{equation}
    Finally, define $\Phi$ to have rows given by the vectors in $\mathcal P$ from Definition~\ref{def:basis-of-symmetric-minors}. Left-multiplication by $\Phi$ has the effect of taking a linear combination of exactly two rows of $\tilde U$.
    The weights cancel and result in the final expression given in Definition~\ref{def:dMRA-M-matrix} which by construction is in $\mathcal M\left(A_\dMRA(\hat x[E]), (\hat x[E^c])^{\oasterisk 2}, \mathcal X_1^\vee\right)$.

    Next, we will show that $M_\dMRA(\hat x)$ transforms equivariantly in that for $g\in\dihgroup{2d}$, $M_\dMRA(g\cdot \hat x) = g\cdot M_\dMRA(\hat x)$ for some group action on the column space.
    First, consider the action of $\fshiftel_\ell \in \dihgroup{2d}$.
    Observe that for each monomial in the row, the sum of indices used to access entries is always $2(j_1 + k_1 + j_2 + k_2 + 2)$.
    Then
    \[
        M_\dMRA(\fshiftel_\ell \cdot \hat x)[\mset{\mset{j_1, k_1},\mset{j_2,k_2}},:] = e^{-2\pi i\ell (j_1 + k_1 + j_2 + k_2 + 2)/d} M_\dMRA(\hat x)[\mset{\mset{j_1, k_1},\mset{j_2,k_2}},:].
    \]
    If we define a matrix $D_\ell := \diag(e^{-2\pi i\ell (j_1 + k_1 + j_2 + k_2 + 2)/d} : \mset{\mset{j_1, k_1},\mset{j_2,k_2}})$,
    where $j_1, j_2, k_1, k_2\in[d]$ and $j_1\leq k_1$, $j_2\leq k_2$, $j_1 < j_2$ and $k_1 < k_2$,
    then
    \[
        M_\dMRA(\fshiftel_\ell\hat x) = D_\ell M_\dMRA(\hat x) \in \mathcal M\left(A_\dMRA(\hat x[E]), (\hat x[E^c])^{\oasterisk 2}, \mathcal X_1^\vee\right).
    \]
    since $D_\ell$ is an invertible matrix.
    It remains to show that exchange also causes a simple change of basis.
    Recall that $(\freflel\cdot  \hat x)[j] = \hat x[-j]$ for $j\in[2d)$ and with indices take modulo $d$.
    Then 
    \[
    M_\dMRA(\freflel \cdot \hat x)[\mset{\mset{j_1, k_1},\mset{j_2,k_2}},:] = M_\dMRA(\hat x)[\mset{\mset{d-j_2-1, d-k_2-1},\mset{d-j_1-1,d-k_1-1}},:].
    \]
    Observe that if $j_1 \leq k_1$, $j_2 \leq k_2$ and $k_1 < k_2$, then the mapping $j \mapsto d-j-1$ flips each inequality.
    Thus, the action of $\freflel$ simply permutes the rows of $M_\dMRA$ by swapping many pairs of rows.
    Hence, $\freflel$ induces some change of basis so in general
    \[
        M_\dMRA(g\cdot \hat x) \in \mathcal M\left(A_\dMRA(\hat x[E]), (\hat x[E^c])^{\oasterisk 2}, \mathcal X_1^\vee\right)
    \]
    for any $g\in\dihgroup{2d}$.
\end{proof}

Next, we will briefly give an elementary result in invariant theory about the decomposition of a symmetric rank-1 matrix.
\begin{lemma}
    \label{lemma:rank-1-up-to-sign}
    If $z\in \K^d$ and $Z= zz^\top$ then $Z$ determines $z$ up to sign.
\end{lemma}
\begin{proof}
    If $Z$ is the zero matrix then $z=\vec{0}$.
    Suppose $Z$ is not the zero matrix and both $Z=z_1z_1^\top$ and $Z=z_2z_2^\top$ for $z_1,z_2\in\K^d$ are rank-1 decompositions.
    Then there exists some $j\in[d]$ so that $z[j]\neq 0$.
    The $j$-th column of $Z$ is $Z[:,j] = z_1[j] z_1 = z_2[j]z_2$ Since $z_1[j]\ne 0$ we know $z_1 = \alpha z_2$ for some $\alpha \in \K$. But then $z_1 z_1^\top = Z = \alpha^2 z_2 z_2^\top$ so $1=\alpha^2$ and $\alpha=\pm1$. We can then conclude that $z_2 = \pm z_1$.
\end{proof}

We can now prove the correctness of the recovery step in Algorithm~\ref{alg:extend-representative-dMRA}.

\begin{proof}[Proof of Theorem~\ref{thm:dMRA-correct-extension}]
    Since $y$ is in the $\dihgroup{d}$-orbit of $\hat x[E]$, by Lemma~\ref{prop:dMRA-variety-constrained-linear-system} the linear system
    \[
        Au=b \quad \text{with}\quad u\in \mathcal X^\vee_1
    \]
    with $A=A_\dMRA(\hat y)$ and $b = b_\dMRA(\hat x)$ has a planted solution of the form $u= z^{\oasterisk 2}$ such that $z = (g\cdot \hat x)|_{E^c}$ and $(g\cdot \hat x)|_E = \hat y$ for some $g\in\dihgroup{2d}$.
    Define $M = M_\dMRA(\hat x)$ to be the specific instance of $M_\dMRA(d)$ using the coefficients of $\hat x$.
    Since symbolically $M_\dMRA(d)$ has full column-rank and $\hat x$ is generic, the specific instance $M$ has full column-rank.
    Then Proposition~\ref{prop:M-dMRA-is-instance} and Proposition~\ref{prop:full-rank-unique-recovery} together imply the subroutine \textsc{FindVarietyConstrainedSolutions} returns $u$.
    Now $u = z^{\oasterisk 2}$ so the matricization of $u$ may be written $U = zz^\top$.
    By Lemma~\ref{lemma:rank-1-up-to-sign} there is a $z'$ such that $U=z'z'^\top$ and $z'=z$ or $z'= -z$.
    We will show either choice of sign in the rank-1 decomposition is valid.
    If the sign is positive and $z'=z$, by assumption $z$ extends $y$ to some $w\in\Cx^{2d}$ in the $\dihgroup{2d}$-orbit of $\hat x$.
    If $z' = -z$ observe that $\fshiftel_d = \diag((-1)^j : j\in[2d))$ so $(\fshiftel_d g\cdot z)[E^c] = -(g\cdot z)[E^c]$ and $(\fshiftel_d g\cdot z)[E] = (g\cdot z)[E] = y$. Since $\fshiftel_d g \in \dihgroup{2d}$, either sign results in an element in the $\dihgroup{2d}$-orbit of $\hat x$.
\end{proof}

\subsection{Sufficient Conditions for Aided Recovery in the pMRA Model}
\label{subsec:pMRA-sufficient-conditions}

As in the previous section, we assume $d$ is even throughout.
The construction of $M_\pMRA$ is very similar to that of $M_\dMRA$ with a small modification to the column indexing function $\beta$.
Define the sets of multisets
\begin{align*}
    L_2 &:= \{\mset{j,k} : j,k\in[d), j+k+1\equiv d/2 \pmod{d}\}, \\
    C' &:= C\setminus L_2, \text{ and} \\
    -C' & := \{\mset{d-1-j,d-1-k} :\mset{j,k}\in C'\},
\end{align*}
where $C$ is defined in \eqref{eq:dMRA-column-labels}.
The accounting is somewhat trickier but $|L_2\cap C|=\lceil d/4\rceil$ (see Figure~\ref{fig:nullspace-diagram} for a visualization).
We can then determine that $|C' \sqcup L_2 | = \lceil d(d+1)/4 \rceil$.
Let $\gamma' : C' \cap L_2 \to  [\lceil d(d+1)/4 \rceil]$ give the DM order on $C' \cap L_2$.
We extend this using the partition given by four sets now, $\mchoose{[d)}{2} = L_1\sqcup C'\sqcup-C'\sqcup L_2$, so that
$\gamma : \mchoose{[d)}{2} \to [\lceil d(d+1)/4 \rceil]\cup\{0\}$ is a map such that $\gamma(C') = \gamma(-C') = \gamma'(C)$, $\gamma(L_2) = \gamma'(L_2)$ and $\gamma(L_1) = 0$.
We may reuse the original sign function defined at the start of Section~\ref{subsec:dMRA-sufficient-conditions}.

Finally, recall that $\hat x[d]$ is indeterminate in the pMRA model.
However, in the following we will set $\hat x[d] = 1$ to avoid dealing with a special case in construction.
Lemma~\ref{lemma:characterize-nullspace-pMRA} and the proof of Proposition~\ref{prop:M-pMRA-is-instance} will clarify why this is used.
With these modification, our $M$-matrix representative for the pMRA model is effectively the same as in Definition~\ref{def:dMRA-M-matrix}.

\begin{definition}
    \label{def:pMRA-M-matrix}
    For $\hat x\in \Cx^{2d}$ the matrix $M_\pMRA(\hat x)$ has rows
    \begin{equation}
        \label{eq:pMRA-M-matrix}
        \begin{split}
        (M_\pMRA)&[\mset{\mset{j_1, k_1},\mset{j_2,k_2}},:] =  \\
            &\begin{array}{rll}
               \quad \hat{x}[2j_1+1]\cdot\hat{x}[2k_1+1] \cdot\hat{x}[2(j_2 + k_2 + 1)] \cdot \sgn(\mset{j_2, k_2}) & \vec{e}_0 \oasterisk \vec{e}_{\gamma(\mset{j_2, k_2})} \\
              -\: \hat{x}[2j_1+1]\cdot\hat{x}[2k_2+1]\cdot\hat{x}[2(j_2 + k_1 + 1)] \cdot \sgn(\mset{j_2, k_1}) & \vec{e}_0 \oasterisk \vec{e}_{\gamma(\mset{j_2, k_1})} \\
              +\: \hat{x}[2j_2+1]\cdot\hat{x}[2k_2+1]\cdot\hat{x}[2(j_1 + k_1 + 1)] \cdot \sgn(\mset{j_1, k_1})& \vec{e}_0 \oasterisk \vec{e}_{\gamma(\mset{j_1, k_1})} \\
              -\: \hat{x}[2j_2+1]\cdot\hat{x}[2k_1+1]\cdot\hat{x}[2(j_1 + k_2 + 1)] \cdot \sgn(\mset{j_1, k_2}) & \vec{e}_0 \oasterisk \vec{e}_{\gamma(\mset{j_1, k_2})} \\
              +\: \hat{x}[2(j_1 + k_1 + 1)]\cdot\hat{x}[2(j_2 + k_2 + 1)] \cdot \sgn(\mset{j_1, k_1})\cdot \sgn(\mset{j_2, k_2})& \vec{e}_{\gamma(\mset{j_1, k_1})} \oasterisk \vec{e}_{\gamma(\mset{j_2, k_2})} \\
              -\: \hat{x}[2(j_1 + k_2 + 1)]\cdot\hat{x}[2(j_2 + k_1 + 1)]\cdot \sgn(\mset{j_1, k_2})\cdot \sgn(\mset{j_2, k_1}) & \vec{e}_{\gamma(\mset{j_1, k_2})} \oasterisk \vec{e}_{\gamma(\mset{j_2, k_1})}
        \end{array}
        \end{split}
    \end{equation}
    with $0\leq j_1 \leq k_1 < d$, $0\leq j_2 \leq k_2 < d$, $j_1 < j_2$ and $k_1 < k_2$.
    Indices are taken modulo $2d$ when accessing entries of $\hat{x}$ and the rowspace of $M_\pMRA$ should be considered a subspace of $(\Cx^{\lceil d(d+1)/4 \rceil+1})^{\oasterisk 2} \cap \Span(\vec e_0 \oasterisk \vec e_0)^\perp$.
    Further, one should replace $\hat x[d]$ with the constant 1 anywhere it appears as the expression above as $M_\pMRA$ does not depend on $\hat x[d]$'s value but the definition requires a non-zero constant in these positions to be correct.
\end{definition}

For $\hat x\in \Cx^{2d}$, $M_\pMRA(\hat x)$ has $\frac1{12}(d+1)d^2(d-1)$ rows and $\binom{\lceil d(d+1)/4 \rceil + 2}{2} - 1$ columns.
As a result, it is square when $d=4$ and has more rows than columns when $d > 4$.
It is also a very sparse matrix, having at most 6 non-zero entries per row.

We have an analogue to Proposition~\ref{prop:M-dMRA-is-instance}.
\begin{proposition}
    \label{prop:M-pMRA-is-instance}
    For $\hat x\in \Cx^{2d}$ and any $g\in \dihgroup{2d}$,
    \[
        M_\pMRA(g\cdot \hat x) \in \mathcal M\left(A_\pMRA(\hat x[E]), (\hat x[E^c])^{\oasterisk 2}, \mathcal X_1^\vee\right)
    \]
    where $\mathcal M(A, u,\mathcal X)$ is the set of equivalent matrices in Definition~\ref{def:equiv-M-matrix-set} for a variety-constrained linear system $Au=b$ with planted solution $u\in\mathcal X$.
\end{proposition}

\begin{conjecture}
    \label{conj:full-rank-pMRA-M-matrix} For $d=2^k$ with $k \ge 3$, $M_\pMRA(\hat x)$ has full column-rank as a symbolic matrix in the $d$ real and $d-1$ imaginary components of the Fourier coefficients of $x\in\Rl^{2d}\cap\Span(1,-1,\dots,1,-1)^\perp$.
\end{conjecture}
See Section~\ref{subsec:numerical-evidence-conjectures} for evidence toward this conjecture.
Similar to the notation $M_\dMRA(d)$, we will write $M_\pMRA(d)$ and drop the dependence on $\hat x$ and reference the parameter $d$ when considering this matrix as a symbolical construct.

\begin{remark}
    \label{rem:symbolic-real-and-imaginary-components-no-nyq}
    See Remark~\ref{rem:symbolic-real-and-imaginary-components} for how these symbolic variables are formally described in the case of the $M_\dMRA$ matrix in Definition~\ref{def:dMRA-M-matrix}.
    The condition that $x\in\Rl^{2d}\cap\Span(1,-1,\dots,1,-1)^\perp$ results in one less variable since now
    \[
    \Re(\hat x) =: (r_0, r_1, r_2,\dots, r_{d-1}, 0, r_{d-1},\dots,r_2, r_1)
    \]
    in the underlying symbolic variables.
    However, in Definition~\ref{def:pMRA-M-matrix}, to ensure that Proposition~\ref{prop:M-pMRA-is-instance} holds, the dependence on $\hat x[d] = r_d$ should formally be removed by replacing it with the constant 1 (or any non-zero constant).
\end{remark}

\begin{theorem}[Aided recovery for projected MRA]
    \label{thm:pMRA-correct-extension}
    Fix even $d \ge 4$ and suppose $M_\pMRA(d)$ has full column-rank as a symbolic matrix.
    Then for generic $x\in\Rl^{2d}\cap\Span(1,-1,\dots,1,-1)^\perp$, if $y$ is an arbitrary element in the $\dihgroup{d}$-orbit of $\hat x[E]$, Algorithm~\ref{alg:extend-representative-pMRA} runs in polynomial time on inputs $y$ and $T^{(3)}_\pMRA(\hat x)$ and recovers some $\hat x'$ in the $\dihgroup{2d}$-orbit of $\hat x$.
\end{theorem}

The construction requires a basis for the kernel of $A_\pMRA(\hat x[E])$, which we now characterize.
Define the set of vectors 
\begin{equation}
    \label{eq:pMRA-nullset}
    N_\pMRA(y) := \{n_m(y): m\in C'\} \cup \{\vec{e}_j\oasterisk \vec{e}_{d/2-j} : j\in[d)\}
\end{equation}
where $n_m(y)$ is a null-vector of the form define in \eqref{eq:general-null-vectors} where indices are taken modulo $d$ on the standard basis vectors.

\begin{lemma}
    \label{lemma:characterize-nullspace-pMRA}
    When $d$ is even and $y\in\Cx^d$ with each $y[j]\neq 0$ for $j\in[d)$, $N_\pMRA(y)$ is a basis for $\ker A_\pMRA(y)$.
\end{lemma}

Comparing constructions in Lemma~\ref{lemma:characterize-nullspace-dMRA} and Lemma~\ref{lemma:characterize-nullspace-pMRA}, $\ker A_\dMRA \subseteq \ker A_\pMRA$ which can also be inferred from the definition of $A_\pMRA$.
In the pMRA model, without the Nyquist component a small set of the dMRA null vectors (specifically the vectors indexed by multisets in the set $L_2$) to split into two orthogonal standard basis vectors in $(\Cx^d)^{\oasterisk 2}$.

\begin{proof}[Proof of Lemma~\ref{lemma:characterize-nullspace-pMRA}]
    Set $A=A_\pMRA(y)$ for $y\in \Cx^d$.
    By Definition~\ref{def:pMRA-recursive-linear-system},
    $A = \fprojOp^{(3)}_{E^c,E^c,E} A_\dMRA(y)$, and in Lemma~\ref{lemma:pMRA-zero-columns} we determined that $A$ has zero-columns in the $\mset{j,k}$-th position when $j+k\equiv d/2$.
    By these observations we may conclude that $\Span(N_\pMRA(y)) \subseteq \ker A_\pMRA$.
    Our main goal is to show that these are the only such vectors.
    First, recall that our indices are taken modulo $2d$ or $d$ depending on the dimension of the matrices involved. If we use $j\in\{-d+1,\dots,0,\dots,d-1, d\}$ to index columns of $\fprojOp$ then observe that the $d$-th column is $\vec 0$ and every other column may be written as
    \[
        \fprojOp[:,j] = e^{\pi i (j-|j|)/2d}\vec e_{|j|}.
    \]
    For $j,k,\ell\in\{-d/2+1,\dots,d-1,d/2\}$ the $(\mset{j,k},\mset{\ell})$-th column of $\fprojOp^{(3)}_{E^c, E^c, E}$ is given by
    \[
        \fprojOp^{(3)}_{E^c, E^c, E}[:, (\mset{j,k},\mset{\ell})] = \begin{cases}
            \phi(j,k)\sqrt{3} w_{\mset{j,k}} \left(\vec{e}_{|j+\frac12| -\frac12 }\oasterisk \vec{e}_{|k+\frac12| -\frac12}\otimes \vec{e}_{|\ell|}\right) & \ell \not\equiv d/2 \pmod{d} \\
            \vec{0} & \ell \equiv d/2
        \end{cases}
    \]
    where $\phi(j,k) =\exp\left(\frac{\pi i (2j+1-|2j+1| + 2k+1-|2k+1| - 2(j+k+1) + 2|j+k+1|)}{2d}\right) = \exp\left(\frac{\pi i (2|j+k+1| -|2j+1|-|2k+1|)}{2d}\right).$
    It then follows that the columns of $A$ have the form 
    \[
        A[:, \mset{j,k}] = \left(\phi(j,k)+\phi(-j-1,-k-1)\right)\frac{\sqrt{3}w_{\mset{j,k}}}2 \left(\vec{e}_{|j+\frac12| -\frac12 }\oasterisk \vec{e}_{|k+\frac12| -\frac12} \otimes\vec{e}_{|j+k+1|}\right)
    \]
    except when $j+k+1\equiv d/2$, in which case the phase factors cancel to form a zero-column, producing the second set of null vectors in~\eqref{eq:pMRA-nullset}.
    Next, for a column of $A$ with index $\mset{j',k'}$ to have the same support as the $\mset{j,k}$-th column of $A$, it must be that
    \begin{align*}
        |j+1/2|-1/2 &\equiv |j'+1/2|-1/2 \pmod{d}, \\
        |j+1/2|-1/2 &\equiv |j'+1/2|-1/2 \pmod{d}, \\
        |j+k+1| &\equiv |j'+k'+1| \pmod{d}.
    \end{align*}
    The first two equations force $j'\equiv j$ or $-j-1 \pmod{d}$ and $k'=k$ or $-k-1\pmod{d}$.
    We know from Lemma~\ref{lemma:characterize-nullspace-dMRA} that when $\mset{j',k'} = \mset{-j-1,-k-1}$, this produces the set of vectors we index by $C'$ in \eqref{eq:pMRA-nullset}.
    For the other two pairs, if $\{j',k'\} = \{-j-1,k\}$ or if $\{j',k'\} = \{j,k-1\}$ then the third modular equation requires either $2j\equiv 1\pmod{d}$ or $2k\equiv 1\pmod{d}$, both of which are impossible since $d$ is even.
    Since the columns have non-overlapping support, they are orthogonal and there can be no additional linearly independent null vectors beyond the ones in $N_\pMRA(y)$. 
\end{proof}

\begin{figure}
    \centering
    \includegraphics[width=0.33\linewidth]{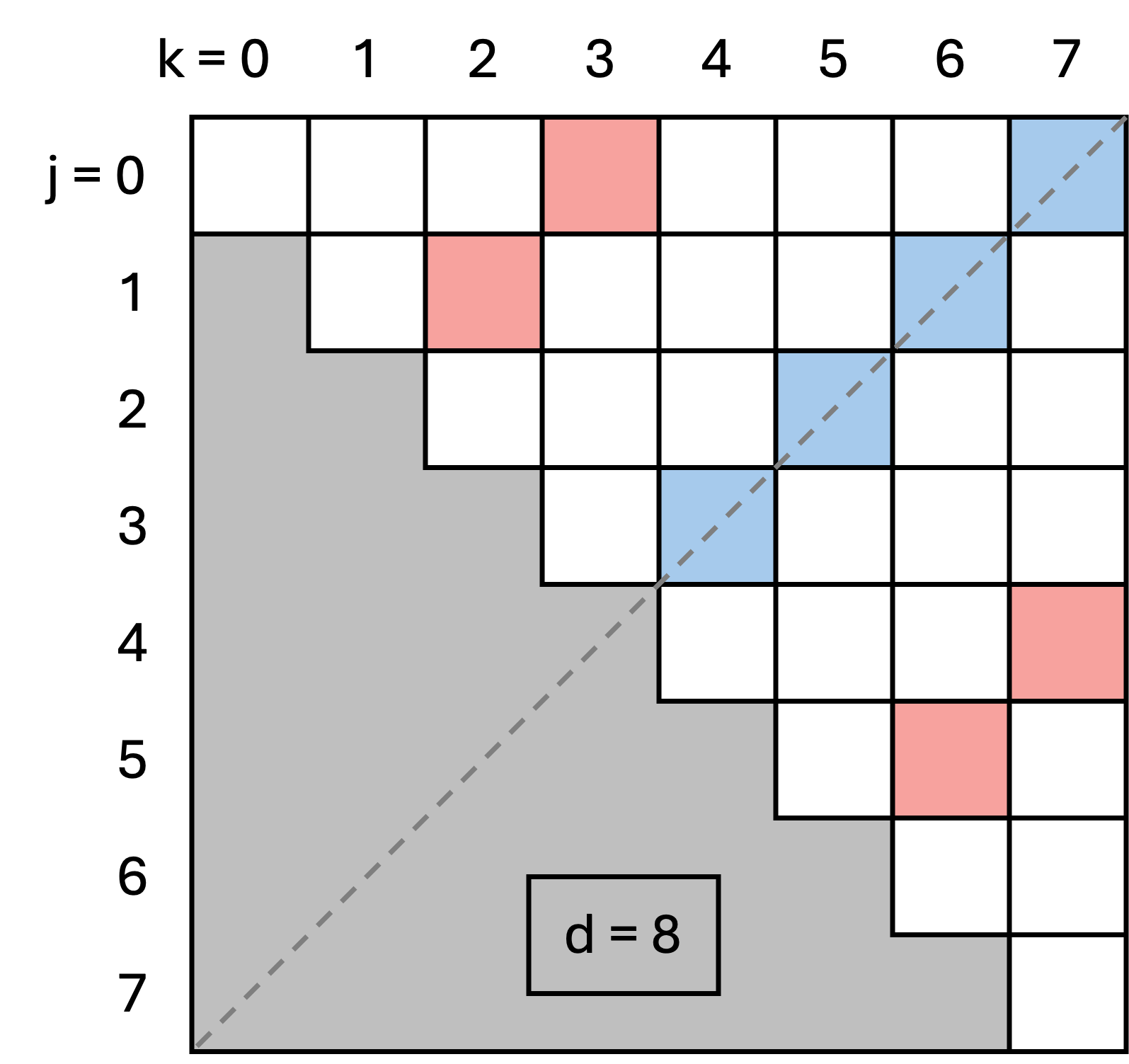}
    \caption{
    Illustration of the combinatorics involved in computing nullspace dimensions for $d=8$.
    Columns of $A_\dMRA$ and $A_\pMRA$ correspond to boxes in the upper triangular region, as they are indexed by pairs $\mset{j,k}$ with $0\leq j \leq k < d$.
    Boxes that are reflections across the diagonal dashed line correspond to columns that are scalar multiples of each other.
    Each pair gives an element in the null space.
    Boxes in blue along the dashed line indicate columns that have a unique support and produce no null vectors.
    Boxes in red indicate columns that are zero-columns in $A_\pMRA$ but not $A_{\dMRA}$; these produce additional null vectors which are tensor products of standard basis elements.
    }
    \label{fig:nullspace-diagram}
\end{figure}

Having verified our choice of basis spans the nullspace of $A_\pMRA$, we will now sketch the proof of Proposition~\ref{prop:M-pMRA-is-instance}, highlighting some subtle differences to Proposition~\ref{prop:M-dMRA-is-instance} though the same technique is used.
\begin{proof}[Proof sketch of Proposition~\ref{prop:M-pMRA-is-instance}]
    In Definition~\ref{def:pMRA-M-matrix}, the function $\gamma$ serves the same role as $\beta$ in Definition~\ref{def:pMRA-M-matrix}, being a column-index lookup function to place the vectors of $N_\dMRA(\hat x[E])$ into a matrix $N$.
    The construction of $M_\pMRA(\hat x)$ from this matrix then follows along the same lines as that in Proposition~\ref{prop:M-dMRA-is-instance}.

    However, $M_\pMRA(g\cdot \hat x)$ is no longer row-equivalent to $M_\pMRA(\hat x)$ since we canonically set $\hat x[d]=1$ in our formula.
    Indeed, when $g\in \{\fshiftel_\ell\}_{\ell=0}^{2d-1} \subseteq \dihgroup{2d}$, multiplying by $D_\ell$ (defined in the proof of Proposition~\ref{prop:M-dMRA-is-instance}) may change the sign of $\hat x[d]$.
    However, observe that this coefficient is used to weight one of the basis vectors in $N_\pMRA(\hat x[E])$ which is just a standard symmetric tensor product of two basis vectors.
    Multiplying $N$ on the right by a diagonal matrix with a sign in the correct place will fix the sign issue and we can conclude that $M_\pMRA(g\cdot \hat x)$ is still matrix-equivalent to $M_\pMRA(\hat x)$.
\end{proof}

Finally, we can verify that the recursive step succeeds.

\begin{proof}[Proof of Theorem~\ref{thm:pMRA-correct-extension}]
    This proof follows from identical logic to that in the proof of Theorem~\ref{thm:dMRA-correct-extension} with the replacements $A= A_\pMRA$, $b= b_\pMRA$, $M= M_\pMRA(\hat x)$ and the reference to Proposition~\ref{prop:dMRA-variety-constrained-linear-system} should be replaced with Proposition~\ref{prop:pMRA-variety-constrained-linear-system}.
\end{proof}

\subsection{Numerical Evidence for Conjectures}
\label{subsec:numerical-evidence-conjectures}
For $x\in\Rl^{2d}$ with Fourier coefficients $\hat x = \mathcal Fx$,
one can produce strong numerical evidence that $M_\dMRA(\hat x)$ and $M_\pMRA(\hat x)$ have full column-rank symbolically by picking $x$ at random from some absolutely continuous distribution in $\Rl^{2d}$ and verifying that the smallest singular value of $M_\dMRA(\hat x)$ or $M_\pMRA(\hat x)$ are each significantly far from zero. Note that demonstrating a \emph{single} choice of $x$ for which $M$ has full column-rank implies that $M$ has full column-rank symbolically, and therefore implies success of aided recovery for generic inputs of this dimension.
The condition number of a rectangular matrix $M$ is defined to be
\begin{equation}
    \kappa(M) = \frac{\sigma_{\max}(M)}{\sigma_{\min}(M)}.
\end{equation}
Since $\kappa(M) = \kappa(\alpha M)$ for $\alpha\neq 0$, the condition number is a scale-invariant measure of how close $M$ is to having a rank-deficit where if $M$ has linearly-dependent columns, we canonically define $\kappa(M) := \infty$.
In the \href{https://github.com/lemniscate8/MRA-variant-orbit-recovery}{following repository}\footnote{Repository is found at https://github.com/lemniscate8/MRA-variant-orbit-recovery for those reading in print.}, we furnish the experiments for the values $d=4,8$ and $16$ where $x$ is chosen to have standard real Gaussian entries.
Out of 500 random trials, we report the smallest condition number obtained for each of $M_\dMRA(\hat x)$ and $M_\pMRA(\hat x)$ in Table~\ref{tab:confidence-intervals}. The results support Conjectures~\ref{conj:full-rank-dMRA-M-matrix} and \ref{conj:full-rank-pMRA-M-matrix}, respectively.
\begin{table}[h]
    \centering
    \begin{tabular}{c|cc}
        $d$ & $\kappa(M_\dMRA)$ & $\kappa(M_\pMRA)$ \\ \hline
         4 & 6.07 & 74.96\\
         8 & 13.36 & 129.78 \\
        16 & 135.81 & 1464.38
    \end{tabular}
    \caption{The smallest condition numbers of both $M_\dMRA(\hat x)$ and $M_\pMRA(\hat x)$ out of 100 trials where each entry of $x\in\Rl^{2d}$ is chosen independently from a standard normal distribution.}
    \label{tab:confidence-intervals}
\end{table}
When $d=32$, we were not able to get our numerical method to converge when computing the smallest singular value for instances of $M_\dMRA(\hat x)$ and $M_\pMRA(\hat x)$.
However, we note that the implementation given in Section~\ref{sec:numerical-experiments} and our numerical experiments there give indirect evidence that $M_\dMRA$ and $M_\pMRA$ still have full column-rank symbolically in this case.

\section{Analysis of the Base Case}
\label{sec:analysis-base-case}
The recursive procedure described above can reduce the problem to an orbit recovery problem of constant size.
In these cases, the theoretical computational complexity is not an issue, but we we do need a proof that the orbit is uniquely identifiable and an explicit procedure to recover it.

\subsection{Frequency Marching in Dihedral MRA}
By the work of Edidin and Katz \cite{edidin_reflection-invariant_2026} the question of identifiability has been solved for the dMRA model: namely, the first three moments separate generic orbits.

Counting the number of rows and columns of $M_\dMRA$ shows that our recursive method cannot work starting from $d=1$ ($x\in \Rl^2$), which is unfortunate since this case is trivial to achieve orbit recovery in as the zeroth and first Fourier coefficients are real and can be estimated from the first and second moments, respectively. 
We then consider the case $d=2$ ($x\in \Rl^4$) as our base case.

The technique of frequency marching arises from the observation that the degree-2 cyclic or dihedral invariants give the magnitudes of Fourier coefficients.
One can then use the degree-3 invariants to set up a system of equations modulo $2\pi$ to solve for the phases of Fourier coefficients.
Edidin and Katz considered frequency marching as a strategy for recovery, but primarily considered it as a proof strategy, since it does not give an efficient algorithm.
However, since we only care about signals of small dimension for our base case, we demonstrate that frequency marching recovers the orbit of a signal of length four.
This example also helps elucidate why this method is inefficient for larger signals (although the concurrent work~\cite{proj-mra} provides a fix to this).

\begin{theorem}
    \label{thm:dMRA-base-case-unique-recovery}
    When $x \in\Rl^4$ with $\hat x[j]\neq 0$ for $j\in[4)$, the third moment $T_\dMRA^{(3)}(\hat x)$ may be used to recover the $\dihgroup{4}$-orbit of $\hat x$.
\end{theorem}

\begin{proof}[Proof of Theorem~\ref{thm:dMRA-base-case-unique-recovery}]
    First, observe that since $x\in\Rl^4$, $\hat x[0]$ and $\hat x[2]$ are real and $\hat x[1] = \overline{\hat x[3]}$.
    While $T_\dMRA^{(1)}(\hat x)[j] = \hat x[0]$ for every $j=0,1,2,3$, we can also recover this from the third moment via 
    \[
    \hat x[0] = \sqrt[3]{T^{(3)}_\dMRA(\hat x)[\mset{0,0,0}] / w_{\mset{0,0,0}}}.
    \]
    Similarly, for each $j=1,2,3$,
    \[
    T^{(3)}_\dMRA[\mset{0,j,-j}] = w_{\mset{0,j,-j}} \hat x[0]\hat x[j]\hat x[-j] =  w_{\mset{0,j,-j}} \hat x[0]|\hat x[j]|^2
    \]  so we can recover the magnitudes of Fourier coefficients from the third moment.
    Using the Fourier coefficient magnitudes we may define the phases of Fourier coefficients to be $\phi\in[0,2\pi)^4$ where $\hat x[j] = |\hat x[j]| e^{i\phi[j]}$ for $j=0,1,2,3$.
    However, $\phi[0]$ is known, and $\phi[3]\equiv-\phi[1]\pmod{2\pi}$ so we need only determine $\phi[1]$ and $\phi[2]$.
    Since $\hat x[2]$ is real, for some choice of sign $s_1\in\{-1,1\}$ we have that $\hat x[2] = s_1 |\hat x[2]|$ and $\phi[2] = \frac\pi 2(1-s_1)$.
    Next, by Proposition~\ref{prop:invariant-tensor-structure},
\begin{align*}
    T^{(3)}_\dMRA[\mset{1,1,2}] &= \frac{w_{\mset{1,1,2}}}2 \left(\hat x[1]^2\hat x[2] + \hat x[3]^2 \hat x[2]\right) \\
    &=  \frac{w_{\mset{1,1,2}}}2 |\hat x[1]|^2 \hat x[2]\left(e^{2i\phi[1]} +  e^{2i\phi[3]}\right) \\
    &=  \frac{w_{\mset{1,1,2}}}2 |\hat x[1]|^2 \hat x[2]\left(e^{2i\phi[1]} +  e^{-2i\phi[1])}\right) \\
    &=  w_{\mset{1,1,2}} |\hat x[1]|^2 \hat x[2] \cos(2\phi[1]).
\end{align*}
Using the expression for $\phi[2]$, we know $\phi[1]$ must satisfy
\begin{equation}
    \label{eq:frequency-march-equation-example}
    \cos(2\phi[1]) = s_1\frac{T^{(3)}_\dMRA[\mset{1,1,2}]}{w_{\mset{1,1,2}} |\hat x[1]|^2 |\hat x[2]|}.
\end{equation}
Treating $\phi[1]$ as an unknown value in $[0,2\pi)$, there are at most four solutions for this value $\phi[1]$.
Further, if $\phi[1]$ is a solution, so are $-\phi[1]$ and $\phi[1]+\pi$.
It follows that for some choice of sign $s_2\in\{-1,1\}$ and $k\in\{0,1\}$,
\[
    \phi[1] = \frac12s_2\cos^{-1}\left(s_1\frac{T^{(3)}_\dMRA[\mset{1,1,2}]}{w_{\mset{1,1,2}} |\hat x[1]|^2 |\hat x[2]|} \right)  + k\pi.
\]
Finally, for each choice of $s_1,s_2\in\{-1,1\}$ and $k\in\{0,1\}$ the assignments
\begin{align*}
    \phi[1] &\gets \frac12s_2\cos^{-1}\left(s_1\frac{T^{(3)}_\dMRA[\mset{1,1,2}]}{w_{\mset{1,1,2}} |\hat x[1]|^2 |\hat x[2]|} \right) + k\pi \\ 
    \phi[2] &\gets \frac{\pi}2(1-s_1)
\end{align*}
and $\phi[3]\gets-\phi[1]\pmod{2\pi}$ enumerate (by varying $s_1,s_2,k$) all 8 elements of a $\dihgroup{4}$-orbit for some signal. 
Since at least one choice of $s_1, s_2$ and $k$ must produce the true phases for $\hat{x}$, this procedure must recover all elements of the $\dihgroup{4}$-orbit of $\hat x$.
\end{proof}

Frequency marching succeeds because we only need to use one relationship between the phases $\phi[1]$ and $\phi[2]$ given by $\eqref{eq:frequency-march-equation-example}$.
In general, for $x\in \Rl^{2d}$, relationships between phases take the form
\[
    \cos(\phi[j_1] + \phi[j_2] + \phi[j_3]) = \frac{T_\dMRA^{(3)}[\mset{j_1, j_2, j_3}]}{w_{\mset{j_1, j_2, j_3}}|\hat x[j_1]||\hat x[j_2]||\hat x[j_3]|} =: \alpha_{j_1, j_2, j_3}
\]
with $j_1,j_2,j_3 \in[2d)$, $j_1 + j_2 + j_3 \equiv 0\pmod{2d}$ and $\phi[j] = -\phi[2d-j\pmod{2d}]$ for $j\in[2d)$.
As noted by Edidin and Katz, we can only determine a linear equation for the phases modulo $2\pi$ up to sign, i.e.\
\[
    \phi[j_1] + \phi[j_2] + \phi[j_3] \equiv \pm \cos^{-1}(\alpha_{j_1, j_2, j_3}) \pmod{2d}.
\]
In this case $d=2$, the number of signs was sufficiently small that each assignment of sign variables produced a consistent system of equations, each giving an element of the orbit of $\hat x$.
For larger $d$ one must resort to an exhaustive search over all choices of signs, most of which will produce inconsistent linear systems.
As a result, we use frequency marching in four dimensions as demonstrated in Theorem~\ref{thm:dMRA-base-case-unique-recovery} as a base case for our recursive method to recover a signal from the dMRA third moment.

\paragraph{Further discussion.} Comparing to the frequency marching implementation by Bendory, Boumal, Ma, Zhao and Singer~\cite[Algorithm 4]{bendory_bispectrum_2018} for regular MRA, our method differs in that we first estimate $\phi[2]$ and then $\phi[1]$.
Their method requires finding the exact value of $\phi[1]$ to begin estimating further frequencies.
The method of doing this exactly given at the beginning of~\cite[Section~IV]{bendory_bispectrum_2018} requires the cyclic invariants and there is not an immediate adaptation wherein we can determine $\phi[1]$ exactly using dihedral invariants.
Instead, we observe that when the dimension is $d=2^k$, a consequence of Lemma~\ref{prop:recursive-invariants} is that one may use the phase $\phi[2^m]$ and the entry $T^{(3)}_\dMRA[\mset{2^{m-1}, 2^{m-1}, d-2^m}]$ to recover the phase $\phi[2^{m-1}]$ by solving
\[
    2\phi[2^{m-1}] + \phi[d-2^m] \equiv 2\phi[2^{m-1}] - \phi[2^m]  \equiv \pm \cos^{-1}(\alpha_{2^{m-1}, 2^{m-1}, 2^m}) \pmod{2d}
\]
for any choice of sign and for $m=k-1, k-2,\dots, 1$ until $\phi[1]$ is found.
Our proof uses this observation for $d=2^2$ in a single step.

\subsection{Frequency Marching in Projected MRA}
In this section, we will demonstrate that frequency marching can still be used in the projected MRA model as an algorithm for the base case.

First, we will explain why the orbits of signals of small dimension cannot be uniquely identified.
Next, we show in Theorem~\ref{thm:pMRA-base-case-list-recovery} that for $x\in\Rl^8\cap \Span(1,-1,\dots,1,-1)^\perp$ one may narrow down the orbit of $\hat x$ to one of four possibilities using a frequency marching-like procedure.
Finally, using results of Edidin and Katz~\cite{edidin_reflection-invariant_2026} but applied to the pMRA model, we show in Theorem~\ref{thm:pMRA-base-case-unique-recovery} that for $x\in\Rl^{16}\cap \Span(1,-1,\dots,1,-1)^\perp$ the third moment separates generic dihedral orbits.

\paragraph{Unidentifiability for small dimensions.} For $d=2$ ($x\in\Rl^4\cap \Span(1,-1,1,-1)^\perp$), examination of the entries of $T^{(3)}_\pMRA$ reveals that the phase of $\phi[1]$ (and $\phi[3]$) cannot be determined.
This occurs because the only two non-zero entries are $T_\pMRA^{(3)}[\mset{0,0,0}]$ and $T_\pMRA^{(3)}[\mset{0,1,1}]$.
Each multiset contains a 0 index and as a result these elements are invariants of the form $\hat x[0]\hat x[j]\hat x [-j]$ and contain no information about the phase of $\hat x[1]$.

\paragraph{Modified frequency marching.}
Next, for $d=4$ ($x\in\Rl^8\cap \Span(1,-1,\dots,1,-1)^\perp$), numerical tests we performed indicated that for a randomly generated real signal there were three additional orbits which shared the first three moments with the true signal. We show that this is indeed the case.
\begin{theorem}
    \label{thm:pMRA-base-case-list-recovery}
    For $x\in \Rl^8 \cap \Span(1,-1,1,\ldots,-1)^\perp$ with $\hat x[j]\neq 0$ for all $j\neq 4$, the moment $T_\pMRA^{(3)}(\hat x)$ determines a set of at most four $\dihgroup{8}$-orbits, one of which is the orbit of $\hat x$.
\end{theorem}
\begin{proof}
    By similar observations to those in the proof of Theorem~\ref{thm:dMRA-base-case-unique-recovery} using Proposition~\ref{prop:invariant-tensor-structure}, we can see that
    \[
    \begin{split}
        T^{(3)}_\dMRA[\mset{0,0,0}] &= 8 w_{\mset{0,0,0}} \hat x[0]^3 \quad \text{and} \\
        T^{(3)}_\dMRA[\mset{0,j,j}] &= 4 w_{\mset{0,j,j}} e^{\pi i (0+j+j)/8}\hat x[0]\hat x[j]\hat x[-j],
    \end{split}
    \]
    so $\hat x[0]$ and the Fourier coefficient magnitudes are given by
    \[
    \begin{split}
        \hat x[0] &= \sqrt[3]{T^{(3)}_\dMRA[\mset{0,0,0}]/(8w_{\mset{0,0,0}})} \quad \text{and} \\
        |\hat x[j]|^2 &= \frac{|T^{(3)}_\dMRA[\mset{0,j,j}]|}{4 w_{\mset{0,j,j}}\hat x[0]}, \quad j=1,2,3.
    \end{split}
    \]
    Next, there are only three non-zero entries of the third moments which do not contain a zero index:
    \begin{align*}
        T_\pMRA^{(3)}[\mset{1,1,2}] &= e^{i\pi/2} w_{\mset{1,1,2}}\left(\hat x[1]^2\hat x[6] + \hat x[7]^2\hat x[2]\right), \\
        T_\pMRA^{(3)}[\mset{1,2,3}] &= e^{3i\pi/4}w_{\mset{1,2,3}}\left(\hat x[1]\hat x[2]\hat x[5] + \hat x[7]\hat x[6]\hat x[3]\right),\\
        T_\pMRA^{(3)}[\mset{2,3,3}] &= -w_{\mset{2,2,3}} \left(x[2]\hat x[3]^2 + \hat x[6]\hat x[5]^2\right).
    \end{align*}
    Using $\hat x[j] = e^{i\phi[j]} |\hat x[j]|$ with $\phi[j]\in[0,2\pi)$ and $\phi[8-j] =-\phi[j]$ for $j=1,2,3$, these three entries allow us to determine the following three relationships between phases:
    \begin{align*}
        a_1 &:= \frac{T_\pMRA^{(3)}[\mset{1,1,2}]}{2 e^{i\pi/2}w_{\mset{1,1,2}} |\hat x[1]|^2 |\hat x[2]|} = \cos(2\phi[1]-\phi[2]), \\
        a_2 &:= \frac{T_\pMRA^{(3)}[\mset{1,2,3}]}{2 e^{3i\pi/4} w_{\mset{1,1,2}} |\hat x[1]| |\hat x[2]||\hat x[3]|} =\cos(\phi[1]+\phi[2]-\phi[3]), \\
        a_3 &:= -\frac{T_\pMRA^{(3)}[\mset{2,2,3}]}{2w_{\mset{2,2,3}} |\hat x[2]|^2 |\hat x[3]|} = \cos(\phi[2]+2\phi[3]).
    \end{align*}
    Setting $\alpha_j := \arccos(a_j)$ for $j=1,2,3$, this determines a system of three equations up to a sign and multiple of $2\pi$:
    \begin{align*}
        2\phi[1]-\phi[2] &= \pm\alpha_1 + 2\pi k_1,  \\
        \phi[1]+\phi[2]-\phi[3] &= \pm\alpha_2 + 2\pi k_2, \\
        \phi[2]+2\phi[3] &= \pm\alpha_3  + 2\pi k_3.
    \end{align*}
    A global change of sign $s_g \in \{-1,1\}$ for these equations corresponds to the action of $\freflel$ on $\hat x$ so we need two additional sign variables, $s_1, s_2\in \{-1,1\} $ to produce all possible systems.
    Hence, all solutions are given via the assignments
    \begin{align*}
        \phi[1] &\gets s_g\frac18\left(3(s_1\alpha_1+2\pi k_1) + 2(s_2\alpha_2+2\pi k_2) + (\alpha_3+2\pi k_3)\right ), \\
        \phi[2] &\gets s_g\frac14\left(-(s_1\alpha_1+2\pi k_1) + 2(s_2\alpha_2+2\pi k_2) + (\alpha_3+2\pi k_3)\right), \\
        \phi[3] &\gets s_g\frac18\left((s_1\alpha_1+2\pi k_1) -2(s_2\alpha_2+2\pi k_2) + 3(\alpha_3+2\pi k_3)\right ),
    \end{align*}
    and $\phi[8-j]\gets-\phi[j] \pmod{2\pi}$ for $j=1,2,3$.
    Observe that if we fix the sign variables $s_1,s_2$ and set $s_g=1$, $k_1=k_2=k_3=0$ to produce a solution $y(s_1, s_2) \in \Cx^{8}$, then each choice of $s_g, k_1, k_2$ and $k_3$ produces a solution $z=\freflel^{(1-s_g)/2}\fshiftel_{(3k_1 + 2k_2 + k_3)} \cdot y(s_1, s_2)$. 
    It follows that each choice $(s_1,s_2)$ produces an orbit, while choices of $(s_g, k_1, k_2, k_3)$ enumerate the 16 elements of the orbit.
    Since these are all the possible assignments to produce signals $y$ such that $T^{(3)}_\pMRA(y) =  T^{(3)}_\pMRA(\hat x)$, at least one choice of $s_1$ and $s_2$ must produce the orbit of $\hat x$.
\end{proof}

\paragraph{Identifiability for larger signals.}
Next, we will show that when $d=8$ ($x\in \Rl^{16} \cap \Span(1,-1,1,\ldots,-1)^\perp$), a generic signal's orbit is determined by the exact projected MRA third moment.
To show this we will need a few constructions and results from Edidin and Katz~\cite{edidin_reflection-invariant_2026}.

Define an $O(2)$-group action on a real vector $x\in\Rl^{2d}$ in the Fourier basis  generated by the reflection element $\freflel$ defined in \eqref{eq:fourier-reflect-operation} and by a now \emph{continuous} cyclic shift $\fshiftel_\theta$ for $\theta\in[0,2\pi)$ acting on a signal $\hat x\in\Cx^{2d}$ in the Fourier basis entry-wise by
\begin{equation}
    \label{eq:continuous-Fourier-shift}
    (\fshiftel_\theta\cdot \hat x)[j]= \begin{cases}
        e^{-i j\theta}\hat x[j] & j=0,1,\dots, d-1, \\
        e^{i (2d-j)\theta}\hat x[j] & j=d,d+1,\dots, 2d-1. \\
    \end{cases}
\end{equation}
Observe that when $\theta$ a multiple of $\frac{2\pi}{2d}$, the action is that of $\dihgroup{2d}$ as defined in \eqref{eq:fourier-shift-operation}.
Define the \emph{restricted action} of $O(2)$ to be the action on $\hat x$ where we have excluded the zeroth Fourier coefficient $\hat x[0]$ and Nyquist component $\hat x[d]$, i.e., the $O(2)$-action restricted to a subspace spanned by standard basis vectors $\vec{e}_1, \vec{e}_2,\dots, \vec{e}_{d-1}, \vec{e}_{d+1}, \dots, \vec{e}_{2d-1}$.
\begin{remark}
    Our definition of the action of $O(2)$ in \eqref{eq:continuous-Fourier-shift} differs from the one given by Edidin and Katz in \cite[Section 3]{edidin_reflection-invariant_2026} since they index the Fourier basis elements using negative indices.
    In addition, a difference in sign in the phase produces a rotation by $\theta$ that is in the opposite direction as Edidin and Katz's definition. These differences in convention will not affect the results.
\end{remark}

Drawing from Edidin and Katz's calculations in \cite[Section~3.1.2]{edidin_reflection-invariant_2026}, in the context of our setup the degree-2 $O(2)$-invariant polynomials are
\[
\hat x[j]\hat x[-j] = |\hat x[j]|^2, \quad 0\leq j \leq d
\]
(with indices taken modulo $2d$), and the degree-3 $O(2)$-invariant polynomials are 
\begin{align*}
    &\hat x[j_1]\hat x[j_2]\hat x[-j_1-j_2] + \hat x[-j_1]\hat x[-j_2]\hat x[j_1+j_2] \\
    &= \hat x[j_1]\hat x[j_2]\hat x[-j_1-j_2] + \overline{\hat x[j_1]\hat x[j_2]\hat x[-j_1-j_2]}, \quad 0\leq j_1\leq j_2 \leq d,\quad  0\leq j_1+ j_2\leq d
\end{align*}
(with indices taken modulo $2d$).
The invariant polynomials for the restricted $O(2)$-action exclude the polynomials containing the terms $\hat x[0]$ and $\hat x[d]$.
\begin{theorem}[{\cite[Theorem 3.2]{edidin_reflection-invariant_2026}}]
    \label{thm:generic-O(2)-separation}
    When $x\in\Rl^{2d}$ for $d\neq 4$, the invariants of degree two and three for the restricted action of $O(2)$ separate generic orbits.
\end{theorem}

To prove that the invariants from the pMRA model separate generic orbits when $d=8$ (meaning the signal has length $16$) we will first show in Lemma~\ref{lemma:pMRA-has-almost-all-invariants} that we actually have access to the dMRA invariants, except for those involving the term $\hat x[d]$.
If we had all the invariants, we could simply use Edidin and Katz's results as presented.

\begin{lemma}
    \label{lemma:pMRA-has-almost-all-invariants}
    Let $x\in \Rl^{16}\cap\Span(1,-1,\dots, 1,-1)^\perp$ with Fourier transform $\hat x$, and let $\mset{j_1,j_2,j_3}$ be a multiset of $[16)\setminus\{8\}$ such that $j_1 + j_2 + j_3 \equiv 0\pmod{16}$. Then there is a multiset $\mset{j'_1, j'_2, j'_3}$ of $[8)$ such that
    \begin{equation}
        \label{eq:invariant-relationships}
        T^{(3)}_\pMRA(\hat x)[\mset{j'_1, j'_2, j'_3}] = K(\{j'_1, j'_2, j'_3\}) (-1)^{(j_1 + j_2 + j_3)/16} e^{\pi i (j'_1 + j'_2 + j'_3)/16} \, T^{(3)}_\dMRA(\hat x)[\mset{j_1, j_2, j_3}],
    \end{equation}
    where $K(\mset{j'_1, j'_2, j'_3})$ is 8 when $\mset{j'_1, j'_2, j'_3}=\mset{0,0,0}$, 4 when $\mset{j'_1, j'_2, j'_3}$ has the form $\mset{0,j,j}$, and 2 otherwise.
\end{lemma}
\begin{proof}
    We split into cases based on the number of zero indices in the multi-set $\mset{j'_1,j'_2,j'_3}$.
    Due to the modular condition, non-zero entries in $T^{(3)}_\pMRA$ appear when $\mset{j'_1,j'_2,j'_3} = \mset{0,0,0}$, when $\mset{j'_1,j'_2,j'_3}=\mset{0, j, j}$ for some $j\in[16]\setminus\{0,8\}$, or when every entry is non-zero.
    The first two cases follow by inspection.
    By the first expression from Proposition~\ref{prop:invariant-tensor-structure}, if $\mset{j_1,j_2,j_3}$ is a multiset of $[16)$ such that $j_1 + j_2 + j_3 \equiv 0\pmod{16}$ then
    \[
    T^{(3)}_\dMRA(\hat x)[\mset{j_1,j_2,j_3}] = \frac{w_{\mset{j_1,j_2,j_3}}}{2}(\hat x[j_1]\hat x[j_2]\hat x[j_3] + \hat x[-j_1]\hat x[-j_2]\hat x[-j_3]).
    \]
    First, for $\mset{j_1,j_2,j_3} = \mset{0,0,0}$, using the second expression from Proposition~\ref{prop:invariant-tensor-structure}, we may simplify to get
    \[
        T^{(3)}_\pMRA(\hat x)[\mset{0,0,0}] = 8w_{\mset{0,0,0}} \hat x[0]^3 = 8\cdot T^{(3)}_\dMRA(\hat x)[\mset{0,0,0}]
    \]
    and so \eqref{eq:invariant-relationships} holds for this multiset index.
    Second, for $j\in\{1,2,\dots,7\}$,
    \[
        T^{(3)}_\pMRA(\hat x)[\mset{0,j,j}] = 4 e^{\pi i (0+j+j)/16}w_{\mset{0,j,j}} \hat x[0]\hat x[j] \hat x[-j] = 4 e^{\pi i (0+j+j)/16}\cdot T^{(3)}_\dMRA(\hat x)[\mset{0,j,-j}]
    \]
    with indices taken modulo 16, so \eqref{eq:invariant-relationships} holds for these multiset indices.
    Third, with the bulk of cases remaining, we will show that for any multiset $\mset{j_1,j_2,j_3}$ of $[16)\setminus\{0,8\}$ there is a multiset $\mset{j'_1,j'_2,j'_3}$ of $[8)\setminus\{0\}$ so that \eqref{eq:invariant-relationships} holds.
    To this end, define a map on indices $f:[16)\setminus\{0,8\} \to [8)\setminus\{0\}$ by
    \[
    f(j) = 
    \begin{cases}
        j & \text{if } j<8, \\
        16-j & \text{otherwise}.
    \end{cases}
    \]
    We extend $f$ to multisets $\mset{j_1,j_2,j_3}$ of $[16)\setminus\{0,8\}$ elementwise by $f(M) = \mset{f(m) : m \in M}$.
For indices $j_1,j_2,j_3\in[8)\setminus\{0\}$, we may then rewrite the second expression from Proposition~\ref{prop:invariant-tensor-structure} using a summation over the pre-image of $f$,
\begin{align*}
    T^{(3)}_\pMRA(\hat x)[\mset{j'_1,j'_2,j'_3}] = \frac{w_{\mset{j'_1,j'_2,j'_3}}}2  &\sum_{\mset{j_1,j_2,j_3}\in f^{-1}(\mset{j'_1,j'_2,j'_3})} \Biggl(\phi(\mset{j_1,j_2,j_3},\mset{j'_1,j'_2,j'_3}) \Biggl.\\
    &  \left.\cdot \left(\hat x[j_1]\hat x[j_2]\hat x[j_3] + \hat x[-j_1]\hat x[-j_2]\hat x[-j_3]\right)\bb1\left\{\sum_{k=1}^r j_k \equiv 0 \pmod{16}\right\}\right)
\end{align*}
where
\[
\phi(\mset{j_1,j_2,j_3},\mset{j'_1,j'_2,j'_3}) = (-1)^{(j_1 + j_2 + j_3)/16} e^{\pi i (j'_1 + j'_2 + j'_3)/16}.
\]
Defining the set
\[
Z=\{\mset{j_1,j_2,j_3} : j_1,j_2,j_3\in [16)\setminus\{0,8\},\ j_1 + j_2 + j_3 \equiv 0\pmod{16}\},
\]
we may absorb the indicator function into the sum condition to get
\[
    T^{(3)}_\pMRA(\hat x)[\mset{j'_1,j'_2,j'_3}] = \frac{w_{\mset{j'_1,j'_2,j'_3}}}2\sum_{\mset{j_1,j_2,j_3}\in f^{-1}(\mset{j'_1,j'_2,j'_3})\cap Z} \phi(\,\cdots)\left(\hat x[j_1]\hat x[j_2]\hat x[j_3] + \hat x[-j_1]\hat x[-j_2]\hat x[-j_3]\right).
\]
However, note that (a) $f^{-1}(\mset{j'_1,j'_2,j'_3})\cap Z = (f|_Z)^{-1}(\mset{j'_1,j'_2,j'_3})$, and (b) the summand is invariant under the substitution $\mset{j_1,j_2,j_3}\mapsto\mset{16-j_1,16-j_2,16-j_3}$.
We will show that, in fact, 
\[
f|_Z^{-1}(\mset{j'_1,j'_2,j'_3}) = \{\mset{j_1,j_2,j_3}, \mset{16-j_1,16-j_2,16-j_3}\}
\]
so that summation simplifies to
\[
    T^{(3)}_\pMRA(\hat x)[\mset{j'_1,j'_2,j'_3}] = 2 (-1)^{(j_1 + j_2 + j_3)/16} e^{\pi i (j'_1 + j'_2 + j'_3)/16} \, T^{(3)}_\pMRA(\hat x)[\mset{j_1,j_2,j_3}].
\]
First, if $f|_Z(\mset{j_1,j_2,j_3}) = \mset{j'_1,j'_2,j'_3}$ then it is simple to see that $f|_Z^{-1}(\mset{j'_1,j'_2,j'_3}) \supseteq \{\mset{j_1,j_2,j_3}, \mset{16-j_1,16-j_2,16-j_3}\}$.
We will show by contradiction these are the only two multisets in the pre-image.
Suppose there is some third multiset $K=\mset{k_1,k_2,k_3}$ so that $K\in Z$ and $f(K)=\mset{j'_1,j'_2,j'_3}$.
Then by the definition of $f$, there will be some permutation $\pi \in S_3$ such that either $k_{\pi(i)} = j_i$ or $k_{\pi(i)} = 16-j_i$ for each $i=1,2,3$. 
Without loss of generality we may assume that $\pi=\mathrm{id}$.
It must be that either $k_1 = 16-j_1$, $k_2 = j_2$ and $k_{\pi(3)} = j_3$, or $k_1 = j_1$, $k_{\pi(2)} = 16-j_2$ and $k_3 = 16-j_3$ for some $\pi\in S_3$.
We may assume the first since if $K\in f|_Z^{-1}(\mset{j'_1,j'_2,j'_3})$ then $\mset{16-k_1, 16-k_2, 16-k_3} \in f|_Z^{-1}(\mset{j'_1,j'_2,j'_3})$.
But since $K\in Z$, $k_1 + k_2 + k_3 = 16-j_1 + j_2 + j_3\equiv 0 \pmod{16}$ and $\{j_1,j_2,j_3\} \in Z$ gives $j_1 + j_2 + j_3 \equiv 0 \pmod{16}$, we find that $(j_1 + j_2 + j_3) - (16-j_1 + j_2 + j_3)  \equiv 0 \pmod{16}$ which implies $2j_1 \equiv0\pmod{16}$ but then $j_1 \equiv0\pmod{8}$ and $j_1=0$ or $j_1=8$ so $\mset{j_1,j_2,j_3} \not\in Z$ which contradicts our initial assumptions.
\end{proof}
Lemma~\ref{lemma:pMRA-has-almost-all-invariants} shows that we can recover (most of) the dihedral MRA 3rd moment from the projected MRA 3rd moment.
Using this and Theorem~\ref{thm:generic-O(2)-separation} we will verify generic identifiability in the pMRA model in the base case $d=8$.
\begin{theorem}
    \label{thm:pMRA-base-case-unique-recovery}
    For generic $x\in \Rl^{16}\cap \Span(1,-1,\dots,1,-1)^\perp$, the moment $T_\pMRA^{(3)}(\hat x)$ uniquely determines the $\dihgroup{16}$-orbit of $x$.
\end{theorem}
\begin{proof}
Let $x\in\Rl^{16}\cap\Span(1,-1,\dots,1,-1)^\perp$ be generic.
Suppose that there is a $y\in \Rl^{16}\cap\Span(1,\dots,-1)^\perp$ such that $T^{(3)}_\pMRA(\hat x) = T^{(3)}_\pMRA(\hat y)$.
We will show that $y$ is in the $\dihgroup{16}$-orbit of $x$.
First, by Lemma~\ref{lemma:pMRA-has-almost-all-invariants}, we have that
\[
    T^{(3)}_\pMRA(\hat x) = T^{(3)}_\pMRA(\hat y) \implies  
    T^{(3)}_\dMRA(\hat x)[\mset{j_1,j_2,j_3}] = T^{(3)}_\dMRA(\hat y)[\mset{j_1,j_2,j_3}] \text{ for all } j_1,j_2,j_3\in [16)\setminus\{8\}.
\]
Observe that since
\[
    T^{(3)}_\dMRA(\hat x)[\mset{0,0,0}]=w_{\mset{0,0,0}} \hat x[0]^3
\]
we can infer that $\hat x[0] = \hat y[0]$, i.e.\ $x$ and $y$ have the same mean value.
Next, for $j\in[16)\setminus\{8\}$ since $x$ is real, $\hat x[-j] = \overline{\hat x[j]}$ and since
\[
    T^{(3)}_\dMRA(\hat x)[\mset{0,j,-j}] = w_{\mset{0,j,-j}} \hat x[0] \hat x[j]\hat x[-j] = w_{\mset{0,j,-j}} \hat x[0] |\hat x[j]|^2,
\]
$|\hat x[j]|=|\hat y[j]|$ for $j\in[8)$, i.e.\ $x$ and $y$ have the same power spectrum.
It remains to show that the phases of $\hat x[j]$ and $\hat y[j]$ must be related by a dihedral group action.
We will first show that $\hat x$ and $\hat y$ must be in the same $O(2)$-orbit.
Observe that for $j_1, j_2, j_3 \in [16)\setminus\{0,8\}$ with $j_1 + j_2 - j_3 =0$,
\[
    T^{(3)}_\dMRA(\hat x)[\mset{j_1,j_2,j_3}] = \frac{w_{\mset{j_1,j_2,j_3}}}2\left(\hat x[j_1]\hat x[j_2]\hat x[j_3] + \hat x[-j_1]\hat x[-j_2]\hat x[-j_3]\right),
\]
so the $O(2)$-invariant polynomials for $\hat x$ and $\hat y$ have the same values.
By Theorem~\ref{thm:generic-O(2)-separation}, these invariants separate generic $O(2)$-orbits, so we know that $\hat x$ and $\hat{y}$ are in the same $O(2)$-orbit, i.e.\ $\hat y = \fshiftel_\theta \hat x$ or $\hat y = \freflel\fshiftel_\theta \hat x$ for some $\theta \in [0,2\pi)$.
Without loss of generality assume the first occurs.
In our final step, we will resolve $\theta = \frac{2\pi k}{16}$ for some $k\in [16)$ using the same argument as in \cite[Section 3.3]{edidin_reflection-invariant_2026}, which implies $\hat x$ and $\hat y$ are in the same $\dihgroup{16}$-orbit.
Pick $j_1+j_2+j_3 = 16$ with $0<j_1,j_2,j_3<8$ (e.g.\ $j_1=3$, $j_2=6$ and $j_3=7$). Then
\begin{align*}
    T^{(3)}_\dMRA(\hat x)[\mset{j_1,j_2,j_3}] &= T^{(3)}_\dMRA(\hat y)[\mset{j_1,j_2,j_3}]  \\
    \implies \hat x[j_1]\hat x[j_2]\hat x[j_3] + \hat x[-j_1]\hat x[-j_2]\hat x[-j_3] &= \hat y[j_1]\hat y[j_2]\hat y[j_3] + \hat y[-j_1]\hat y[-j_2]\hat y[-j_3]  \\
    \implies\hat x[j_1]\hat x[j_2]\hat x[j_3] + \hat x[-j_1]\hat x[-j_2]\hat x[-j_3] &= e^{-i\theta(j_1+j_2+j_3)} \hat x[j_1]\hat x[j_2]\hat x[j_3] + e^{i\theta(j_1+j_2+j_3)}\hat x[-j_1]\hat x[-j_2]\hat x[-j_3]  \\
    \implies \hat x[j_1]\hat x[j_2]\hat x[j_3] + \hat x[-j_1]\hat x[-j_2]\hat x[-j_3] &= e^{-16i\theta} \hat x[j_1]\hat x[j_2]\hat x[j_3] + e^{16i\theta}\hat x[-j_1]\hat x[-j_2]\hat x[-j_3].
\end{align*}
Define $\phi\in[0,2\pi)$ to be such that $|\hat x[j_1]\hat x[j_2]\hat x[j_3]|e^{-i\phi} = \hat x[j_1]\hat x[j_2]\hat x[j_3]$.
Then
\begin{align*}
    e^{-i\phi} +e^{i\phi} &= e^{-i\phi}e^{-16i\theta} + e^{-i\phi}e^{-16i} \\
    \implies \quad \cos(\phi) &= \cos(\phi + 16\theta).
\end{align*}
Since $\phi$ is generic, this equality only holds when $16\theta =2\pi k$ for some $k\in [16)$, so $\theta = \frac{2\pi k}{16}$.
Thus, $\hat y$ is in the $\dihgroup{16}$-orbit of $\hat x$.
Hence, the third order invariants from projected MRA separate generic orbits.
\end{proof}

\paragraph{Algorithmic implications.}
While unique recovery of an orbit is not possible for signals of length 8, we use frequency marching in 8 dimensions ($d=4$) as demonstrated in Theorem~\ref{thm:pMRA-base-case-list-recovery} as a base case for our recursive method to recover a set of four orbit representatives in $\Cx^8$ from the dMRA third moment.
If each of these is used as the input to Algorithm~\ref{alg:extend-representative-pMRA} to recover some larger signal in $\Cx^{16}$,
Theorem~\ref{thm:dMRA-base-case-unique-recovery} then ensures that any solutions to the variety-constrained linear system belong to the same orbit.
Alternatively if any of the other three orbit representatives fail to extend to an orbit representative in $\Cx^{16}$, they may simply be discarded.
See the next section and proof of Theorem~\ref{thm:pMRA-main} for a rigorous explanation.

\section{Putting it Together: Proofs of Main Results}
\label{sec:main-results}

First, we define the matrix $M=M(d)$ referenced in Theorem~\ref{thm:dMRA-main}.
For $d=2^k$ and $x\in\Rl^{2d}$ we give this matrix as the direct sum of matrices
\begin{equation}
\label{eq:master-M-dMRA}
M(\hat x) := M_\dMRA(\hat x[K_3]) \oplus M_\dMRA(\hat x[K_4])\oplus \dots \oplus M_\dMRA(\hat x[K_k])\oplus M_\dMRA(\hat x),
\end{equation}
where $K_{m} := \{ 2^{k+1-m} j : j\in [2^m)\}$ so that each $\hat x[K_m]$ is a subsequence of $2^m$ Fourier coefficients of $x$.
We omit the dependence on $\hat x$ (as in the statement of Theorem~\ref{thm:dMRA-main}) when considering $M$ to be a symbolic matrix in the $2^k+1$ real and $2^k-1$ imaginary components of $\hat x$ in the same sense as Remark~\ref{rem:symbolic-real-and-imaginary-components}.

\begin{proof}[Proof of Theorem~\ref{thm:dMRA-main}]
First, observe that the effect of extracting even Fourier coefficients repeatedly has the same effect as indexing by directly by the sets $K_m$ in that $\hat x[K_k] = \hat x[E_{2\cdot 2^k}]$ and $(\hat x[K_m])[E_{2^m}] = \hat x[K_{m-1}]$ where $E_d$ is defined in \eqref{eq:even-numbers}.
By Proposition~\ref{prop:recursive-invariants}, one may extract the subtensor of $T_\dMRA(\hat x)$ which will have the form $T^{(3)}_\dMRA(\hat x[E_{2\cdot2^k}]) = T^{(3)}_\dMRA(\hat x[K_k])$.
Applying this repeatedly one may then extract sequence of sub-tensors of the form
\[
   T^{(3)}_\dMRA(\hat x[K_2]), T^{(3)}_\dMRA(\hat x[K_3]), \dots, T^{(3)}_\dMRA(\hat x[K_k]),T^{(3)}_\dMRA(\hat x).
\]
For the sub-tensor $T^{(3)}_\dMRA(\hat x[K_2])\in(\Cx^4)^{\oasterisk 3}$, by Theorem~\ref{thm:dMRA-base-case-unique-recovery} there is an efficient method for finding a representative $y_2$ in the $\dihgroup{4}$-orbit of $\hat x[K_2]$.
Since $x$ is generic and under the assumption that the  matrix $M$ has full column-rank symbolically, the instance $M(\hat x)$ has full column-rank and as a result each of
\[
    M_\dMRA(\hat x[K_3]), M_\dMRA(\hat x[K_4]), \dots, M_\dMRA(\hat x[K_k]),\text{ and } M_\dMRA(\hat x)
\]
has full column-rank.
Further, we note that for each $g\in\dihgroup{2\cdot 2^k}$, the instances $M(g\cdot \hat x)$ will have full column-rank  as a consequence of Proposition~\ref{prop:M-dMRA-is-instance}. Theorem~\ref{thm:dMRA-correct-extension} then guarantees that Algorithm~\ref{alg:extend-representative-dMRA} run on inputs $y_2$ and $T^{(3)}_\dMRA(\hat x[K_3])$ produces $y_3$ in the $\dihgroup{8}$-orbit of $\hat x[K_3]$.
At each stage, this process may be repeated to recover $y_m$ the $\dihgroup{2^m}$-orbit of $\hat x[K_m]$ until the final step recovers $y$ in the $\dihgroup{2d}$-orbit of $\hat x$.

Finally, since the method for initializing $y_2$ as described in the proof of Theorem~\ref{thm:dMRA-base-case-unique-recovery} takes constant time and Algorithm~\ref{alg:extend-representative-dMRA} runs in polynomial time, the recursive procedure as a whole runs in polynomial time.
\end{proof}

Similar to the dihedral MRA case, we define the matrix $M'=M'(d)$ referenced in Theorem~\ref{thm:pMRA-main}.
For $d=2^k$ and $x\in\Rl^{2d}\cap\Span(1,-1,\dots,1,-1)$, we give this matrix as the direct sum of matrices
\begin{equation}
    \label{eq:master-M-pMRA}
    M'(\hat x) := M_\pMRA(\hat x[K_4]) \oplus M_\pMRA(\hat x[K_5])\oplus\dots \oplus M_\pMRA(\hat x[K_k])\oplus M_\pMRA(\hat x)
\end{equation}
where $K_m$ is defined as before.
We will also omit the dependence on $\hat x$ (as in the statement of Theorem~\ref{thm:pMRA-main}) when considering $M'$ to be a symbolic matrix in the $2^k$ real and $2^k-1$ imaginary components of $\hat x$ in the same sense as in Remark~\ref{rem:symbolic-real-and-imaginary-components-no-nyq}.

\begin{proof}[Proof of Theorem~\ref{thm:pMRA-main}]
    In a similar manner to that in the proof of Theorem~\ref{thm:dMRA-main}, by Proposition~\ref{prop:recursive-invariants}, we may extract a sequence of sub-tensors of the form
    \[
        T^{(3)}_\pMRA(\hat x[K_3]), T^{(3)}_\pMRA(\hat x[K_4]), \dots, T^{(3)}_\pMRA(\hat x[K_k]),T^{(3)}_\pMRA(\hat x).
    \]
    For the sub-tensor $T^{(3)}_\dMRA(\hat x[K_3])\in(\Cx^8)^{\oasterisk 3}$, by Theorem~\ref{thm:pMRA-base-case-list-recovery} there is an efficient method for finding a set of four vectors $y_{3,1}, y_{3,2}, y_{3,3}$ and $y_{3,4}\in\Cx^8$ at least one of which (say $y_{3,1}$ without loss of generality) is in the $\dihgroup{4}$-orbit of $\hat x[K_2]$.
    Since $x$ is generic and under the assumption that the  matrix $M'$ has full column-rank symbolically, the instance $M(\hat x)$ has full column-rank and as a result each of
    \[
        M_\pMRA(\hat x[K_4]), M_\pMRA(\hat x[K_5]), \dots, M_\pMRA(\hat x[K_k]),\text{ and } M_\pMRA(\hat x)
    \]
    has full column-rank. 
    Theorem~\ref{thm:pMRA-correct-extension} then guarantees that Algorithm~\ref{alg:extend-representative-pMRA} run on inputs $y_{3,1}$ and $T^{(3)}_\dMRA(\hat x[K_4])$ to produces $y_4\in\Cx^{16}$ in the $\dihgroup{16}$-orbit of $\hat x[K_4]$.
    For the remaining vectors, we have no guarantees that Algorithm~\ref{alg:extend-representative-pMRA} will succeed: if it fails on each of $y_{3,j}$ for any $j\in\{2,3,4\}$ before succeeding on $y_{3,1}$, there is no issue.
    However, if Algorithm~\ref{alg:extend-representative-pMRA} also succeeds for inputs $y_{3,j}$ and $T^{(3)}_\pMRA(\hat x[K_4])$ for some $j\in \{2,3,4\}$ then it produces some $y_4'$ such that $b_\pMRA(y_4') = b_\pMRA(\hat x[K_4])$.
    But then $T_\pMRA(y_4') = T_\pMRA(\hat x[K_4])$.
    Since $x$ is generic, so is the subset of Fourier coefficients $\hat x[K_4]$, and by Theorem~\ref{thm:pMRA-base-case-unique-recovery} we can conclude that $y_4'$ is also in $\dihgroup{16}$-orbit of $\hat x[K_4]$.
    Hence, if Algorithm~\ref{alg:extend-representative-pMRA} does not output ``Fail'' for any inputs $y_{3,j}$, $j\in[4]$ and $T^{(3)}_\pMRA(\hat x[K_4])$, the output is guaranteed to lie in the in the $\dihgroup{16}$-orbit of $\hat x[K_4]$.
    From here the process is then the same as the dihedral case, where Algorithm~\ref{alg:extend-representative-pMRA} may be repeatedly applied to recover $y_m$ the $\dihgroup{2^m}$-orbit of $\hat x[K_m]$ until the final step recovers $y$ in the $\dihgroup{2d}$-orbit of $\hat x$.
    
    Finally, since the method for finding the candidates $y_{3,1}, y_{3,2}, y_{3,3}$ and $y_{3,4}$ given in Theorem~\ref{thm:pMRA-base-case-list-recovery} takes constant time and Algorithm~\ref{alg:extend-representative-dMRA} runs in polynomial time, the recursive procedure as a whole runs in polynomial time.
\end{proof}

\section{Efficient Implementation}
\label{sec:efficient-implementation}
Our algorithms obscure some computational details in the \textsc{FindVarietyConstrainedSolutions} subroutine (our Algorithm~\ref{alg:algorithm1-prime}) by using Algorithm 1 from \cite{johnston_computing_2023}.
However, we feel the process deserves more coverage since the statement of Algorithm 1 in \cite{johnston_computing_2023} is given at a high level of abstraction.

To this end, first we briefly introduce some standard notation for tensor decompositions that is helpful in a concise description.
\begin{definition}[Factor matrices]
    Given matrices $A\in\K^{n_1\times R}$, $B\in\K^{n_2\times R}$, and $C\in\K^{n_3\times R}$ and a rank-$R$ tensor defined by
    \[
        T := \sum_{r=1}^R A[:,r]\otimes B[:,r]\otimes C[:,r],
    \]
    we say that $A$, $B$ and $C$ are the \emph{factor matrices} of $T$ and may denote a tensor constructed in this manner concisely by $T:=[\![A,B,C]\!]$.
\end{definition}
If $\Pi\in\Rl^{R\times R}$ is a permutation matrix and $D_1,D_2, D_3 \in \K^{R\times R}$ are diagonal matrices where $D_1 D_2 D_3 = I$ then $[\![A\Pi,B\Pi,C\Pi]\!] = [\![A,B,C]\!]$ and $[\![AD_1,BD_2,CD_3]\!] = [\![A,B,C]\!]$
so the factor matrices are determined only up to permutation and re-scalings of columns of each factor in a way that preserve the Frobenius norm of each rank-1 summand.

We synthesize the application of Bouss\'e et al.\ \cite{bousse_linear_2018} to give a method for finding the solution(s) to $Au=b$ subject to $u\in\mathcal X_1^\vee$ where $A\in\K^{m\times n}$, $u\in \K^n$, using Algorithm 1 from \cite{johnston_computing_2023}. 

\begin{framed}
    \begin{center}
        \textbf{Algorithm~\ref{alg:algorithm1-prime} (detailed procedure).}
    \end{center}
\begin{enumerate}
    \item \label{step:find-basis} Find any particular solution $x_p$ to $Au=b$ and $\{n_j\}_{j=1}^L$ a basis for the nullspace of $A$. Concatenate as columns the matrix
    \[
    N = \begin{bmatrix}
        x_p & n_1 & n_2 & \cdots & n_L
    \end{bmatrix} \in \K^{n\times (L+1)}.
    \]
    \item \label{step:find-kernel} Compute a basis $\{c_j\}_{j=1}^s\subseteq (\K^{(L+1)})^{\oasterisk 2}$ for $\ker (\Phi N^{\oasterisk 2})$ where $\Phi$ has rows a basis for $I(\mathcal X_1^\vee)_2$.

    If $\Phi N^{\oasterisk 2}$ has full rank, report ``No solutions.''
    \item Reshape each vector to a symmetric matrix $c_j \mapsto C_j\in \K^{(L+1)\times (L+1)}$.
    Use $\{C_j\}_{j=1}^s$ as slabs of a partially symmetric tensor $\mathcal T\in \K^{(L+1)\times (L+1)\times s}$.
    \item Run simultaneous diagonalization\textsuperscript{\ref{footnote:sim-diag}} on $\mathcal T$ so that $\mathcal T = [\![W,W,V]\!]$ or report ``Fail'' if simultaneous diagonalization fails.
    \item \label{step:normalize-weights} Find the normalized weight matrix $\tilde W = W \diag(W[1,:])^{-1}$ with all 1's in the first row.
    \item Solution(s) are the columns of $X = U\tilde W$.
    \item (Optional) Matricize each column of $X$ and compute its symmetric rank-1 decomposition.
\end{enumerate}
\end{framed}

\footnotetext{\label{footnote:sim-diag} If $s=1$, the tensor degenerates to a single matrix $C_1=[\![W,W,V]\!] = [\![w,w,\lambda]\!]  = \lambda ww^\top$ for $w\in\K^n$ and $\lambda\in\Rl^+$ so simultaneous diagonalization becomes a symmetric rank-1 factorization in this special case.}

In comparison with the original algorithm  \cite[Algorithm 1]{johnston_computing_2023}, in addition to our specific application which necessitates the normalization step (Step~\ref{step:normalize-weights}), we note that it is possible to run simultaneous diagonalization directly on the reshaping of the \emph{coefficient} vectors $\{c_j\}_{j=1}^s$ and save space and time running simultaneous diagonalization on $(L+1)\times(L+1)$ matrices rather than $n\times n$ matrices so that the complexity of this step scales with the dimension of the subspace, not the ambient space.
Only in the final step do we change back to the standard basis after which one may compute decompositions of the planted solutions.

Unfortunately, in general $\Phi N^{\oasterisk 2}$ may be a very large matrix with a small kernel (one-dimensional in the case of unique recovery).
In our applications we also can produce a very sparse basis for our nullspace and $\Phi$ is also sparse by construction.
To take better advantage of these structural properties, we will give an alternative method which converts the problem into a top right singular vector computation.

First, define the set of all possible splittings of a multiset $\mset{a,b,c,d}$ of $[d)$ into two multisets of two elements as
\[
    \operatorname{split}(\mset{a,b,c,d}) = \left\{\mset{\mset{j,k}, \mset{\ell,m}} : \mset{j,k,\ell,m}=\mset{a,b,c,f}\right\}.
\]

\begin{definition}
    \label{def:orthogonal-basis-of-symmetric-minors}
    Let $\mathcal Q$ be the set of vectors
    \[
    q_{abcf} = \sum_{\mset{\mset{j,k}, \mset{\ell,m}}\in \operatorname{split}({\mset{a,b,c,f}})} (\vec{e}_j\oasterisk \vec{e}_k) \oasterisk (\vec{e}_\ell \oasterisk \vec{e}_m).
    \]
\end{definition}
\begin{proposition}
    \label{prop:basis-for-orthogonal-subspace}
    $\mathcal Q$ is an orthogonal basis for $I(\mathcal X_1^\vee)_2^\perp$ which has dimension $\mchoose{d}{4}$.
\end{proposition}
See Section~\ref{subsec:polynomial-subspaces} for the proof.
Using this, we may replace Steps \ref{step:find-basis} and \ref{step:find-kernel} with Steps \ref{step:find-orthonormal-basis} and \ref{step:find-right-singular-vectors}:
\begin{framed}
    \begin{enumerate}[label=\arabic*${}^*$. ,ref=\arabic*${}^*$]
    \item \label{step:find-orthonormal-basis}  Find $\tilde x_p$, the least-squares solution to $Au=b$ (given by $x_p = A^+b$ where $A^+$ is the psuedoinverse of $A$) and find $\{\tilde n_j\}_{j=1}^L$ an \emph{orthonormal} basis for the nullspace of $A$.
    Concatenate as columns the matrix
    \[
    \tilde N = \begin{bmatrix}
        \frac{\tilde x_p}{\|\tilde x_p\|_2} & \tilde n_1 & \tilde n_2 & \cdots & \tilde n_L
    \end{bmatrix} \in \K^{d\times (L+1)}.
    \]
    \item \label{step:find-right-singular-vectors} Compute a basis $\{c_j\}_{j=1}^s \subseteq (\K^{(L+1)})^{\oasterisk 2}$ of the \emph{right singular subspace} of $\Phi_\perp \tilde N^{\oasterisk 2}$ associated with the top singular value of 1 where $\Phi_\perp$ has rows an orthonormal basis for $I(\mathcal X_1^\vee)_2^\perp$.

    If each singular value of $\Phi_\perp \tilde N^{\oasterisk 2}$ is less than 1, report ``No solutions.''
\end{enumerate}
\end{framed}
Most importantly, we claim that these two methods are equivalent.
\begin{proposition}
    The basis produced by\label{prop:equivalent-intersection-computations}
    steps~\ref{step:find-basis} and \ref{step:find-kernel} and the basis produced by steps~\ref{step:find-orthonormal-basis} and \ref{step:find-right-singular-vectors} span the same subspace.
\end{proposition}

The orthogonal basis $\mathcal Q$ can be easily normalized.
The advantage of Steps~\ref{step:find-orthonormal-basis} and \ref{step:find-right-singular-vectors} is that we can find top singular vectors --- a method that can be done efficiently using power iteration --- instead of finding the kernel of a massive matrix.
By Lemmas~\ref{lemma:characterize-nullspace-dMRA} and \ref{lemma:characterize-nullspace-pMRA} bases for the nullspaces needed for the variety-constrained linear systems may be made very sparse (a constant number of entries per row) and orthogonal.
This simplifies finding an orthonormal basis in step~\ref{step:find-orthonormal-basis} and the sparsity significantly improves matrix-vector the efficiency of multiplication against $\Phi_\perp \tilde N^{\oasterisk 2}$.
\begin{theorem}
    \label{thm:efficient-computation}
    The matrix-vector products $\Phi_\perp \tilde N^{\oasterisk 2}v$ and  $w^\top\Phi_\perp \tilde N^{\oasterisk 2}$ may be performed in $O(d^4)$ floating-point operations.
\end{theorem}
The number of operations is the same order as the number of rows or columns of $\Phi_\perp \tilde N^{\oasterisk 2}$.
Without a sparse construction for $\tilde N$, a naive matrix-vector multiplication would require $O(d^8)$ operations.

\begin{proof}[Proof of Theorem~\ref{thm:efficient-computation}.]
    For both the dMRA and pMRA model, $\tilde N$ has one dense column $\tilde x_p/\|\tilde x_p\|_2$.
    However, by Lemmas~\ref{lemma:characterize-nullspace-dMRA} and \ref{lemma:characterize-nullspace-pMRA} the remaining columns have non-overlapping support and contain at most two non-zero entries.
    As a result, $\tilde N$ has at most two non-zero entries per row and $d$ rows.
    It then follows from Lemma~\ref{lemma:symmetric-product-equivalences} that $\tilde N^{\oasterisk 2}$ will have at most 4 non-zero entries per row.
    Next, the rows of $\Phi_\perp$ are also sparse, containing at most 3 non-zero entries.
    The matrix multiplications $\Phi_\perp \tilde N^{\oasterisk 2}$ has rows which are linear combinations of the rows of $N^{\oasterisk 2}$ from which we can infer that each row contains no more than 12 non-zero entries.
    This gives no more than $23\mchoose{d}{4}$ operations by summing scalar products and additions to compute the dot-product of $v$ with every row of $\Phi_\perp \tilde N^{\oasterisk 2}$ for the right multiplication.

    For the left multiplication, the first column of $\tilde N$ is dense with $\mchoose{d}{2}$ entries.
    As a result, one column of $\tilde N^{\oasterisk 2}$ contains $O(d^4)$ entries, $O(d^2)$ columns contain $O(d^2)$ entries and $O(d^4)$ columns contain a constant number of entries.
    Each row of $\Phi_\perp$ has a constant number of non-zero entries so performing the multiplication $\Phi_\perp \tilde N^{\oasterisk 2}$ required $O(d^4)$ arithmetic operations.
    Similarly, finding $w^\top\Phi_\perp \tilde N^{\oasterisk 2}$ will require $O(d^4)$ operations.
\end{proof}

Next, to demonstrate the equivalence of these two methods we first define the principal angles and principal vectors between two subspaces $\mathcal A, \mathcal B\subseteq \K^n$ where $p\ge \dim(\mathcal A) \ge \dim(\mathcal B) \ge q$.
\begin{definition}
    \label{def:principles-angles-vectors}
    The principal angles $\theta_k$ and principal vectors $u_k\in\mathcal A,v_k\in\mathcal B$ between two subspaces are recursively defined for $k=1,2,\dots, q$ by
    \[
    \cos(\theta_k) = \max_{u\in\mathcal A} \max_{v\in\mathcal B} u^*v =: u_k^* v_k, \qquad \|u\|=1,\|v\| = 1
    \]
    under the orthogonality constraints
    \[
        u_j^*u = 0,\quad v_j^*v=0 \qquad j=1,2,\dots, k-1.
    \]
\end{definition}
Our simplification arises from an efficient method for finding principles angles by Bjorck and Golub.
\begin{theorem}[{\cite[Theorem 1]{bjorck_numerical_1973}}]
    \label{thm:principle-subspaces}
    Let $Q_\mathcal A$ and $Q_\mathcal B$ have columns that are unitary bases for subspaces $\mathcal A, \mathcal B\subseteq \K^n$.
    If the $p\times q$ matrix
    \[
        M=Q_\mathcal A^* Q_\mathcal B
    \]
    has a singular value decomposition
    \[
        M = X\Sigma Y^*, \qquad \Sigma=\diag(\sigma_1, \sigma_2,\dots,\sigma_q)
    \]
    with $X^*X = Y^*Y=YY^*=I_q$ and $\sigma_1\ge \sigma_2\ge\dots\ge\sigma_q$ then
    \[
    \cos(\theta_k) = \sigma_k, \quad u_k = Q_\mathcal A X[:,k], \quad  \text{and}\quad  v_k = Q_\mathcal B Y[:,k], \quad k=1,2,\dots, q.
    \]
\end{theorem}
Note that the singular values of 1 correspond to principal angles of zero; in other words, we can extract these principal vectors to find a subspace intersection.
Finally, the following lemma ensures that an orthonormal basis may be lifted to an orthonormal basis for a symmetric space using our symmetric power.
\begin{lemma}
    \label{lemma:symmetric-lift-independence}
    If the columns of $A\in\K^{m\times n}$ where $m\leq n$ are orthonormal then so are the columns of the $k$-th symmetric lift $A^{\oasterisk k}$.
    Similarly, if the columns of $A$ are linearly independent then so are the columns of the $k$-th symmetric lift $A^{\oasterisk k}$.
\end{lemma}
See Section~\ref{subsec:properties-symmetric-products} for the proof.
This is not true for $A^{\otimes k}$ since it contains redundant columns.

\begin{proof}[Proof of Proposition~\ref{prop:equivalent-intersection-computations}]
    We will show that both $\ker(\Phi N^{\oasterisk 2})$ and the right singular subspace of $\Phi_\perp \tilde N^{\oasterisk 2}$ associated with the top singular value of 1, may be used to construct a basis for $\operatorname{colspan}(N^{\oasterisk 2}) \cap I(\mathcal X_1^\vee)_2^\perp$.
    
    The first case is simplest.
    Suppose that $c\in \ker(\Phi N^{\oasterisk 2})$.
    Clearly $N^{\oasterisk 2} c \in \colspan(N^{\oasterisk 2})$.
    But $\Phi (N^{\oasterisk 2}c) = \vec{0}$ so $N^{\oasterisk 2} c\in \ker(\Phi)$.
    By construction, $\rowspan(\Phi) = I(\mathcal X_1^\vee)_2$ so $\ker(\Phi) =I(\mathcal X_1^\vee)_2^\perp$ and then $N^{\oasterisk 2} c \in I(\mathcal X_1^\vee)_2^\perp \cap \colspan(N^{\oasterisk 2})$.
    
    Next, for the second case observe by Lemma~\ref{lemma:symmetric-lift-independence} since $\tilde N$ is an orthonormal basis, so is $\tilde N^{\oasterisk 2}$.
    By Theorem~\ref{thm:principle-subspaces}, the SVD \[
        \Phi_\perp \tilde N^{\oasterisk 2}=X\Sigma Y^*
    \]
    gives the principal vectors between these subspaces via $u_k=\Phi_\perp X[:,k]$ and $v_k=\tilde N^{\oasterisk 2} Y[:,k]$.
    Suppose that $\sigma_1=\sigma_2=\dots=\sigma_s=1$.
    Then the first $s$ principles angles are zero and as a consequence of Definition~\ref{def:principles-angles-vectors} we have $u_k=v_k$ for $k=1,2,\dots,s$. Since $u_k\in I(\mathcal X_1^\vee)_2^\perp$ and $v_k\in \colspan(\tilde N^{\oasterisk 2})=\colspan( N^{\oasterisk 2})$ it follows that $v_k \in \colspan(N^{\oasterisk 2}) \cap I(\mathcal X_1^\vee)_2^\perp$ for $s=1,2,\dots, k$ and these vectors are a basis for the intersection.
\end{proof}

\subsection{Numerical Experiments with Full Algorithm}
\label{sec:numerical-experiments}
We implement Algorithms~\ref{alg:extend-representative-dMRA} and \ref{alg:extend-representative-pMRA} using our more efficient method for solving variety-constrained linear systems in Python.
For our procedure, we generated real signals $x\in \Rl^{64}$ by selecting entries uniformly at random from a standard normal distribution.
The exact third moments given by \eqref{eq:dMRA-invariant-tensor} and \eqref{eq:pMRA-invariant-tensor} for the dMRA and pMRA model are computed.
Our program for reconstruction reads the invariants and applies the Algorithms~\ref{alg:extend-representative-dMRA} and~\ref{alg:extend-representative-dMRA} recursively.
For the base case of the recursion, when the recursion reaches a signal length of 4 for dihedral MRA we use the frequency marching procedure from Theorem~\ref{thm:dMRA-base-case-unique-recovery}, and when the recursion reaches a signal length of 8 for projected MRA we use our modified frequency marching procedure from Theorem~\ref{thm:pMRA-base-case-list-recovery} and the next step in the recursion determines which of the four orbit representatives extend to a length-16 orbit representative.

To measure any error in recovery, we measure the dihedral-invariant distance
\begin{equation}
    d(\hat x, y) = \min_{g\in\dihgroup{2d}}\|\hat x - g\cdot y\|_2
\end{equation}
between the true set of Fourier coefficients $\hat x$ and the vector generated by our procedure $y$.
In projected MRA, we omit the Nyquist component (the 32nd coefficient) of each signal when computing this distance which is equivalent to setting each to zero.
In random tests, our algorithm recovers an estimate for the true signal with a reconstruction error close enough to zero to be consistent with loss of precision in floating point computations with a median error of $2.37\times 10^{-10}$ for dihedral MRA and $1.21\times 10^{-11}$ for projected MRA out of 1000 trials.

Since the algorithm is recursive and floating point error calculations are inexact we can also get a sense for how errors propagates through the recursion.
We can split up the vector of Fourier coefficients $\hat x\in \Cx^2$ based on which coefficients are recovered at each level of recursion.
We define the $m$-th cohort of Fourier coefficients to be the odd coefficients at level $m$ of the recursion, i.e.\ if $y \in \Cx^{2^k}$,
\begin{align*}
    C_0(y) &:= (y[0]), \\
    C_m(y) &:= (y[2^{k-m}(2j-1)] : j\in[2^{m-1}]), \quad 1<m\leq k.
\end{align*}
We define the errors for cohorts as $e_0 = |\hat x[0] - y[0]|$ and $e_m = d(C_m(\hat x), C_m(y)) / 2^{m-1}$ for $1<m\leq k$ where we have normalized by the number of elements in the cohort.
Table~\ref{tab:cohort-errors} gives a median error for each cohort over 1000 samples.
While these median errors are small and near machine precision, note that they grow for each cohort of odd Fourier coefficients recovered.
We also remark that cohort one is always the Nyquist component so we exclude it from analysis in the pMRA case.
\begin{table}[ht]
    \centering
    \renewcommand{\arraystretch}{1.5}
    \begin{tabular}{c|c|c}
         Cohort Error & dMRA & pMRA \\
         \hline
        $e_0$ & $2.11\times 10^{-15}$ & $1.11\times 10^{-15}$ \\
        $e_1$ & $1.56\times 10^{-14}$ & N/A \\
        $e_2$ & $2.39\times 10^{-14}$ & $2.24\times 10^{-14}$ \\
        $e_3$ & $1.15\times 10^{-12}$ & $1.88\times 10^{-14}$ \\
        $e_4$ & $1.76\times 10^{-12}$ & $8.98\times 10^{-14}$ \\
        $e_5$ & $3.27\times 10^{-12}$ & $1.53\times 10^{-13}$ \\
        $e_6$ & $6.07\times 10^{-12}$ & $3.45\times 10^{-13}$
    \end{tabular}
    \caption{Median reconstruction errors for each cohort of odd Fourier coefficients over 1000 random samples. }
    \label{tab:cohort-errors}
\end{table}

These results can be verified by the Python scripts in the \href{https://github.com/lemniscate8/MRA-variant-orbit-recovery}{following repository.}\footnote{Repository is found at https://github.com/lemniscate8/MRA-variant-orbit-recovery for those reading in print.}.
We also furnish unit tests verifying various mathematical calculations done in this paper.
See the repository for a detailed description of the contents.

\section{Future Directions}
\label{sec:future-directions}

\paragraph{Removing the power-of-two assumption.}
Our methodology exploits the specific structure of both dihedral and projected MRA, relying on a particular recursive structure which requires that the dimension of the signal be a power of two.
It would be interesting to adapt our approach to all signal lengths. The concurrent work~\cite{proj-mra} achieves this using a different style of algorithm (for both dihedral MRA and a variation of projected MRA).

\paragraph{Proving the conjectures.}
While proof of generic recovery in some dimension $d$ can be verified by checking the rank of certain symbolic matrices, 
we were not able to prove in general that these matrices always have full column-rank. 
Techniques for such a proof would likely lead to improved bounds for finding rank-1 symmetric matrices planted in a subspace~\cite{johnston_computing_2023} (see previous work on the non-symmetric case by the authors and Dastidar~\cite{dastidar_improving_2025}).

\paragraph{Continuous groups and heterogeneity.}
It would be interesting to explore whether our methods can be applied in \emph{continuous} multi-reference alignment, that is, orbit recovery under the action of an infinite group such as $\mathrm{SO}(2)$ or $\mathrm{O}(2)$, potentially with a projection step. It would additionally be interesting to explore whether our methods can handle a heterogeneous combination of different signals (as described in the introduction).

\paragraph{Other projection operators.}
With regards to projection operations, a close analysis of Proposition~\ref{prop:invariant-tensor-structure} reveals that most monomials present as entries in $T^{(3)}_\dMRA$ also appear in $T^{(3)}_\pMRA$. 
As a result, the effect of projection is mostly negligible for lower order moments in the sense that we have almost the same information, even though the projection cuts the signal dimension in half.
Determining more general conditions on a projection operation so that the sample complexity remains the same as for the ``un-projected'' version is an interesting direction for future work.

\paragraph{Robustness to noise.}
Finally, we have given an algorithm in the exact setting where we have abstracted away the estimation of moments from samples.
It remains open to give an end-to-end analysis of our method applied to noisy samples, and to bound the sample complexity.
This would require, in particular, a robust analysis of the intersection algorithm \cite[Algorithm 1]{johnston_computing_2023}; this currently only exists in a non-planted setting~\cite{bhaskara_new_2024}.
To illustrate that robustness can be a non-trivial issue, we point to the setting of cryo-ET: while an ``exact'' algorithm runs in polynomial time~\cite{bandeira_estimation_2023}, a more detailed analysis of the noise tolerance requires \emph{quasi-polynomial} time and sample complexity~\cite{liu_algorithms_2022}. Determining whether our method will incur similar costs is an interesting question for further work.

\appendix

\section{Appendix}
\label{sec:appendix}
\subsection{Properties of Symmetrized Products}
\label{subsec:properties-symmetric-products}
\begin{proof}[Proof of Lemma~\ref{lemma:canonical-orthonormal-basis-symmetric-space}]
    It is well known that vectors of the form $\vec{e}_M$ given in \eqref{eq:symmetric-vector} where $M$ is an $k$-element multiset of $[n]$ are proportional to $\Sym_k(\bigotimes_{j\in M} \vec{e}_j)$ which span $S^k(\K^n)$.
    Since there are $\mchoose{n}{k}$-many vectors and $\dim S^k(\K^n)=\mchoose{n}{k}$, showing these vectors are orthonormal imply they are an orthonormal basis for $S^k(\K^n)$.
    Consider that
    \[
    \langle \vec{e}_{\mset{j_1,\dots,j_k}},\vec{e}_{\mset{j'_1,\dots,j'_k}}\rangle = \frac{w_{\mset{j_1,\dots,j_k}}w_{\mset{j'_1,\dots,j'_k}}}{(k!)^2}\sum_{\sigma\in S_k}\sum_{\tau\in S_k} \prod_{i=1}^k \langle \vec e_{j_{\sigma(i)}}, \vec e_{j'_{\sigma(i)}} \rangle.
    \]
    One can see that there must be a permutation sending the elements of $\mset{j_1,\dots,j_k}$ to $\mset{j'_1,\dots,j'_k}$ for this inner product to be non-zero, so vectors of this form are orthogonal. Next,
    \[
        \|\vec{e}_{\mset{j_1,\dots,j_k}}\|^2 = \frac{w^2_{\mset{j_1,\dots,j_k}}}{(k!)^2}\sum_{\sigma\in S_k}\sum_{\tau\in S_k} \prod_{i=1}^k \langle \vec e_{j_{\sigma(i)}}, \vec e_{j_{\tau(i)}} \rangle = \frac{w^2_{\mset{j_1,\dots,j_k}}}{k!}\sum_{\sigma\in S_k} \prod_{i=1}^k \langle \vec e_{j_{\sigma(i)}}, \vec e_{j_i} \rangle
    \]
    The sum then counts the number of permutations that do not change $(j_1,\dots,j_k)$ and $\frac{w^2_{\mset{j_1,\dots,j_k}}}{k!}$ is the reciprocal of this quantity (see \eqref{eq:multiset-weight}) so $\|\vec{e}_{\mset{j_1,\dots,j_k}}\|^2$ = 1.
\end{proof}
\begin{proof}[Proof of Lemma~\ref{lemma:symmetric-product-inner-product}.]
We use the standard inner product on $\K^{\mchoose{n}{k}}$ and the Definition~\ref{def:symmetrized-Kronecker-product} for the symmetrized Kronecker product to expand:
\[
    \langle u^1 \oasterisk \cdots \oasterisk u^k, v^1 \oasterisk \cdots \oasterisk v^k \rangle = \sum_{1\leq j_1\ \leq \dots \leq j_k \leq n} \frac{w^2_{\mset{j_1,\dots, j_k}}}{(k!)^2}\left(\sum_{\sigma\in S_k} u^1[j_{\sigma(1)}] \cdots u^k[j_{\sigma(k)}]\right)\left(\sum_{\sigma\in S_k} v^1[j_{\sigma(1)}]  \cdots v^k[j_{\sigma(k)}]\right).
\]
Next, notice that $w^2_{\mset{j_1,\dots, j_k}}$ counts the number of unique permutations of the multiset $\mset{j_1,\dots, j_k}$.
We could instead write the sum indexed by all $k$-tuples in $[n]^{k}$ giving
\[
    \sum_{j_1,j_2\dots, j_k \in [n]} \left(\frac1{k!}\sum_{\sigma\in S_k} u^1[j_{\sigma(1)}] \cdots u^k[j_{\sigma(k)}]\right)\left(\frac1{k!}\sum_{\sigma\in S_k} v^1[j_{\sigma(1)}]  \cdots v^k[j_{\sigma(k)}]\right).
\]
Finally, we may rewrite the two sums as entries of Kronecker products
\[
    \sum_{j_1,j_2\dots, j_k \in [n]} \left(\frac1{k!}\sum_{\sigma\in S_k} (u^1\otimes \cdots \otimes u^k)[(j_{\sigma(1)}, \dots, j_{\sigma(k)})]\right)\left(\frac1{k!}\sum_{\sigma\in S_k} (v^1\otimes \cdots \otimes v^k)[(j_{\sigma(1)}, \dots, j_{\sigma(k)})]\right),
\]
which then is the inner product of symmetric projections,
\[
    \langle \Sym_k(u^{1}\otimes u^2 \otimes \cdots\otimes u^k) ,  \Sym_k(v^{1}\otimes v^2 \otimes \cdots\otimes v^k)\rangle.
\]
\end{proof}
\begin{proof}[Proof of Corollary~\ref{cor:symmetric-product-inner-product-expansion}.]
By Lemma~\ref{lemma:symmetric-product-inner-product}, for vectors $u^1,\dots, u^k,v^1,\dots,v^k \in \K^n$,
\[
\langle u^1 \oasterisk \cdots \oasterisk u^k, v^1 \oasterisk \cdots \oasterisk v^k \rangle = \langle \Sym_k(u^{1}\otimes u^2 \otimes \cdots\otimes u^k) ,  \Sym_k(v^{1}\otimes v^2 \otimes \cdots\otimes v^k)\rangle.
\]
Expanding the right-hand side by linearity of symmetric projection yields
\[
\Sym_k(u^{1}\otimes \cdots\otimes u^k) ,  \Sym_k(v^{1} \otimes \cdots\otimes v^k)\rangle = \frac1{(k!)^2}\sum_{\sigma\in S_k}\sum_{\tau\in S_k} \langle u^{\sigma(1)} \otimes \cdots\otimes u^{\sigma(k)}, v^{\tau(1)} \otimes \cdots\otimes v^{\tau(k)} \rangle.
\]
The inner product of a Kronecker product is the product of individual inner products:
\[
\langle u^{\sigma(1)} \otimes \cdots\otimes u^{\sigma(k)}, v^{\tau(1)} \otimes \cdots\otimes v^{\tau(k)} \rangle = \prod_{i=1}^k \langle u^{\sigma(i)}, v^{\tau(i)} \rangle,
\]
so
\[
\langle u^1 \oasterisk \cdots \oasterisk u^k, v^1 \oasterisk \cdots \oasterisk v^k \rangle = \frac1{(k!)^2}\sum_{\sigma\in S_k}\sum_{\tau\in S_k} \prod_{i=1}^k \langle u^{\sigma(i)}, v^{\tau(i)} \rangle.
\]
Finally, the product is invariant under permutation so we may reorder then re-index to eliminate a sum:
\begin{align*}
    \sum_{\sigma\in S_k}\sum_{\tau\in S_k} \prod_{i=1}^k \langle u^{\sigma(i)}, v^{\tau(i)} \rangle
    = \sum_{\sigma\in S_k}\sum_{\tau\in S_k} \prod_{i=1}^k \langle u^{i}, v^{\tau(\sigma^{-1}(i))} \rangle = \sum_{\sigma\in S_k}\sum_{\tau\in S_k} \prod_{i=1}^k \langle u^{i}, v^{\tau(i)} \rangle = k! \sum_{\tau\in S_k} \prod_{i=1}^k \langle u^{i}, v^{\tau(i)} \rangle,
\end{align*}
giving
\[
\langle u^1 \oasterisk \cdots \oasterisk u^k, v^1 \oasterisk \cdots \oasterisk v^k \rangle = \frac1{k!}\sum_{\sigma\in S_k} \prod_{i=1}^k \langle u^{i}, v^{\sigma(i)} \rangle.
\]

\end{proof}

\begin{proof}[Proof of Lemma~\ref{lemma:symmetric-product-equivalences}.]
    We will expand \eqref{eq:symmetric-matrix-product-row} and \eqref{eq:symmetric-matrix-product-column} to show that produces the same entry-wise expression as Definition~\ref{def:symmetric-matrix-product}.
    
    Indexing into the summand of~\eqref{eq:symmetric-matrix-product-row} --- reproduced here:
    \[
        (A^1 \oasterisk A^2 \oasterisk \dots \oasterisk A^k)[\mset{j_1,j_2,\dots,j_k}, :]= \frac{w_{\mset{j_1,\dots,j_k}}}{k!}\sum_{\sigma\in S_k}  A^1[j_{\sigma(1)}, :] \oasterisk A^2[j_{\sigma(2)}, :] \oasterisk \cdots \oasterisk A^k[j_{\sigma(k)}, :],
    \]
    --- by a multiset $\mset{j'_1,j'_2,\dots,j'_k}$ of $[n]$ and expanding by Definition~\eqref{eq:symmetrized-Kronecker-product} gives
    \[
        (A^1[j_{\sigma(1)}, :] \oasterisk A^2[j_{\sigma(2)}, :] \oasterisk \cdots \oasterisk A^k[j_{\sigma(k)}, :])[\mset{j'_1,j'_2,\dots,j'_k}] = \frac{w_{\mset{j'_1,\dots,j'_k}}}{k!} \sum_{\tau\in S_k} \prod_{i=1}^K A^i[j_{\sigma(i)}, j_{\tau(i)}].
    \]
    Substitution into \eqref{eq:symmetric-matrix-product-row} and rearranging sums and terms yields
    \[
        (A^1 \oasterisk A^2 \oasterisk \dots \oasterisk A^k)[\mset{j_1,j_2,\dots,j_k}, \mset{j'_1,j'_2,\dots,j'_k}] = 
        \frac{w_{\mset{j_1,\dots, j_k}}w_{\mset{j'_1,\dots, j'_k}}}{(k!)^2} \sum_{\sigma\in S_k} \sum_{\tau\in S_k} \prod_{i=1}^k A^i[j_{\sigma(i)}, j'_{\tau(i)}]
    \]
    which is exactly Definition~\ref{def:symmetric-matrix-product}.

    Next, indexing into the summand of the \eqref{eq:symmetric-matrix-product-column} --- reproduced here:
    \begin{equation}
        (A^1 \oasterisk A^2 \oasterisk \dots \oasterisk A^k)[:,\mset{j'_1,j'_2,\dots,j'_k}]= \frac{w_{\mset{j'_1,\dots, j'_k}}}{k!}\sum_{\tau\in S_k}  A^1[:,j'_{\tau(1)}] \oasterisk A^2[:,j'_{\tau(2)}] \oasterisk \cdots \oasterisk A^k[:,j'_{\tau(k)}]
    \end{equation}
    --- by a multiset $\{j_1,j_2,\dots,j_k\}$ of $[m]$ and expanding by Definition~\eqref{eq:symmetrized-Kronecker-product} is essentially the same process and gives
    \[
    (A^1[:, \sigma(j'_1)] \oasterisk A^2[:, \sigma(j'_2)] \oasterisk \cdots \oasterisk A^k[:, \sigma(j'_k)])[\{j_1,j_2,\dots,j_k\}] = \frac{w_{\mset{j_1,\dots,j_k}}}{k!} \sum_{\tau\in S_k} \prod_{i=1}^K A^i[j_{\sigma(i)}, j'_{\tau(i)}].
    \]
    Substitution into \eqref{eq:symmetric-matrix-product-column}, rearranging sums and terms also yields
    \[
        (A^1 \oasterisk A^2 \oasterisk \dots \oasterisk A^k)[\mset{j_1,j_2,\dots,j_k}, \mset{j'_1,j'_2,\dots,j'_k}] = 
        \frac{w_{\mset{j_1,\dots, j_k}}w_{\mset{j'_1,\dots, j'_k}}}{(k!)^2} \sum_{\sigma\in S_k} \sum_{\tau\in S_k} \prod_{i=1}^k A^i[j_{\sigma(i)}, j'_{\tau(i)}]. \qedhere
    \]
\end{proof}

\begin{proof}[Proof of Lemma~\ref{lemma:symmetric-matrix-products}.]
    We will show the expressions are equivalent. Let $\mset{j_1, \dots, j_k}$ be a multiset of $[m]$ and let $\mset{j'_1, \dots, j'_k}$ be a multiset of $[n]$.
    We may write an entry of matrix multiplication $(A^1\oasterisk \cdots \oasterisk A^k)(B^1\oasterisk \cdots \oasterisk B^k)$ as an inner product of the conjugate-transpose of a row of $(A^1\oasterisk \cdots \oasterisk A^k)$ and a column of $(B^1\oasterisk \cdots \oasterisk B^k)$ with expressions given by Lemma~\ref{lemma:symmetric-product-equivalences}:
    \begin{align*}
        &(A^1\oasterisk \cdots \oasterisk A^k)(B^1\oasterisk \cdots \oasterisk B^k)[\mset{j_1, \dots, j_k}, \mset{j'_1, \dots, j'_k}] \qquad \qquad (*)\\
        &= \frac{w_{\mset{j_1,\dots,j_k}}w_{\mset{j'_1,\dots, j'_k}}}{(k!)^2}\left(\sum_{\sigma\in S_k}  A^1[j_{\sigma(1)}, :] \oasterisk \cdots \oasterisk A^k[j_{\sigma(k)}, :]\right)\left(\sum_{\tau\in S_k}  B^1[:,j'_{\tau(1)}] \oasterisk \cdots \oasterisk B^k[:,j'_{\tau(k)}]\right) \\
        &= \frac{w_{\mset{j_1,\dots,j_k}}w_{\mset{j'_1,\dots, j'_k}}}{(k!)^2}\sum_{\sigma\in S_k} \sum_{\tau\in S_k}  \left(A^1[j_{\sigma(1)}, :] \oasterisk \cdots \oasterisk A^k[j_{\sigma(k)}, :]\right) \left(B^1[:,j'_{\tau(1)}] \oasterisk \cdots \oasterisk B^k[:,j'_{\tau(k)}]\right) \\
        &= \frac{w_{\mset{j_1,\dots,j_k}}w_{\mset{j'_1,\dots, j'_k}}}{(k!)^2}\sum_{\sigma\in S_k} \sum_{\tau\in S_k}  \langle (A^1[j_{\sigma(1)}, :] \oasterisk \cdots \oasterisk A^k[j_{\sigma(k)}, :])^*, B^1[:,j'_{\tau(1)}] \oasterisk \cdots \oasterisk B^k[:,j'_{\tau(k)}]\rangle.
\end{align*}
Applying Corollary~\ref{cor:symmetric-product-inner-product-expansion} gives
\begin{align*}
    (*) &= \frac{w_{\mset{j_1,\dots,j_k}}w_{\mset{j'_1,\dots, j'_k}}}{(k!)^2}\sum_{\sigma\in S_k} \sum_{\tau\in S_k}  \frac1{k!}\sum_{\rho\in S_k} \langle (A^i[j_{\sigma(i)}, :])^* , B^{\rho(i)}[:, j'_{\tau(i)}]\rangle \\
    &= \frac1{k!}\sum_{\rho\in S_k}  \frac{w_{\mset{j_1,\dots,j_k}}w_{\mset{j'_1,\dots, j'_k}}}{(k!)^2}\sum_{\sigma\in S_k} \sum_{\tau\in S_k}  (A^iB^{\rho(i)})[j_{\sigma(i)}, j'_{\tau(i)}] \\
    &= \frac1{k!}\sum_{\rho\in S_k}  \left((A^1 B^{\rho(1)})\oasterisk (A^2 B^{\rho(2)}) \oasterisk \cdots \oasterisk (A^k B^{\rho(k)})\right)[\mset{j_1, \dots, j_k}, \mset{j'_1, \dots, j'_k}]
\end{align*}
where the last line follows from Definition~\ref{def:symmetric-matrix-product} of the symmetric matrix product.
\end{proof}

\begin{proof}[Proof of Lemma~\ref{lemma:symmetric-lift-independence}.]
    We will first show that the claim holds when $A$ has  orthonormal columns then use this to show the claim holds for linearly independent columns of $A$.

    \begin{enumerate}
        \item \emph{Orthonormal.} Assume that $A=Q \in \K^{m\times n}$ has orthonormal columns, where $m\ge n$.
    We will show that $(Q^{\oasterisk k})^* Q^{\oasterisk k} = I$ by expanding its entries.
    By matrix multiplication and the column-wise definition of the symmetric lift given by Lemma~\ref{lemma:symmetric-product-equivalences},
    \[
        ((Q^{\oasterisk k})^* Q^{\oasterisk k})[\mset{j_1, \dots, j_k}, \mset{j'_1, \dots, j'_k}] = \langle w_{\mset{j_1, \dots, j_k}} Q^{\oasterisk k}[:, \mset{j_1, \dots, j_k}], w_{\mset{j'_1, \dots, j'_k}} Q^{\oasterisk k}[:, \mset{j'_1, \dots, j'_k}] \rangle.
    \]
    Next, using Corollary~\ref{cor:symmetric-product-inner-product-expansion} this inner product becomes
    \[
        \frac{w_{\mset{j_1, \dots, j_k}}w_{\mset{j'_1, \dots, j'_k}}}{k!}\sum_{\sigma\in S_k} \prod_{i=1}^k \langle Q[:,j_i], Q[:,j'_{\sigma(i)}] \rangle = \frac{w_{\mset{j_1, \dots, j_k}}w_{\mset{j'_1, \dots, j'_k}}}{k!}\sum_{\sigma\in S_k} \prod_{i=1}^k  \begin{cases}
            1 & j_i = j'_{\sigma(i)}, \\
            0 & \text{otherwise},
        \end{cases}
    \]
    where the piecewise expression follows from $Q$ having orthonormal columns.
    Now this expression can only be non-zero if $\mset{j_1, \dots, j_k} = \mset{j'_1, \dots, j'_k}$, otherwise for any permutation at least one $i\in[k]$ will have $j_i = j'_{\sigma(i)}$.
    In the case that $\mset{j_1, \dots, j_k} = \mset{j'_1, \dots, j'_k}$,
    \[
        \frac{w^2_{\mset{j_1, \dots, j_k}}}{k!} \sum_{\sigma \in S_k} \begin{cases}
            1 & (j_1,j_2,\dots,j_k) = (j_{\sigma(1)},j_{\sigma(2)},\dots,j_{\sigma(k)}), \\
            0 & \text{otherwise}.
        \end{cases}
    \]
    Now, the sum counts how many permutations of the multiset will remain unchanged if we fix an ordering.
    But this is exactly $\prod_{m \in \multi(\{j_1, \dots, j_k\})} m!$ which cancels with components of the weight $w^2_{\mset{j_1, \dots, j_k}}$.
    As a result when $\mset{j_1, \dots, j_k} = \mset{j'_1, \dots, j'_k}$ this simplifies to 1 and
    \[
        ((Q^{\oasterisk k})^* Q^{\oasterisk k})[\mset{j_1, \dots, j_k}, \mset{j'_1, \dots, j'_k}] = \begin{cases}
            1 & \mset{j_1, \dots, j_k} = \mset{j'_1, \dots, j'_k}, \\
            0 & \text{otherwise}.
        \end{cases}
    \]
    This is the identity matrix of size $\mchoose{n}{k}$ so $Q^{\oasterisk k}$ has orthonormal columns.
    \item \emph{Linearly independent}. Now, suppose that $A$ has linearly independent columns. Then $A$ has a thin-QR decomposition $A=QR$ where $Q$ is has orthonormal columns $R$ is a full-rank upper triangular matrix, i.e.\ $R[j,j'] = 0$ when $j>j'$.
    Without loss of generality we may assume $R$ has non-zero real entries along its diagonal, $R[j,j] > 0$ for $j\in[n]$, since it has full rank.
    By Corollary~\ref{cor:lift-of-product} we have that $A^{\oasterisk k} = Q^{\oasterisk k} R^{\oasterisk k}$.
    Step 1 shows that $Q^{\oasterisk k}$ is orthonormal so we need to show $R^{\oasterisk k}$ is full-rank and upper triangular.
    We will first show that $R^{\oasterisk k}$ is upper triangular.
    First, recall a matrix is upper triangular if for $j\in[m]$, $j'\in[n]$, $R[j,j'] = 0$ when $j>j'$.
    Expanding entries in $R^{\oasterisk k}$ by Definition~\ref{def:symmetric-matrix-product} gives
    \[
    (R^{\oasterisk k})[\mset{j_1,j_2,\dots,j_k}, \mset{j'_1,j'_2,\dots,j'_k}] = 
        \frac{w_{\mset{j_1,\dots, j_k}}w_{\mset{j'_1,\dots, j'_k}}}{(k!)^2} \sum_{\sigma\in S_k} \sum_{\tau\in S_k} \prod_{i=1}^k R[j_{\sigma(i)}, j'_{\tau(i)}].
    \]
    We may simplify by re-indexing permutation to get
    \[
        (R^{\oasterisk k})[\mset{j_1,j_2,\dots,j_k}, \mset{j'_1,j'_2,\dots,j'_k}] = 
        \frac{w_{\mset{j_1,\dots, j_k}}w_{\mset{j'_1,\dots, j'_k}}}{k!} \sum_{\sigma\in S_k} \prod_{i=1}^k R[j_i, j'_{\sigma(i)}].
    \]
    Now if $M=\mset{j_1,j_2,\dots,j_k}$ and $M' = \mset{j'_1,j'_2,\dots,j'_k}$ we will show that $(R^{\oasterisk k})[M, M'] = 0$ whenever $M\succ M'$ by contraposition.
    Suppose that $(R^{\oasterisk k})[M, M'] \neq 0$ for some $k$-element multisets $M$ and $M'$ of $[d]$.
    Then it must be that there is some $\sigma : [k]\to [k]$ such that for every $i\in[k]$, $j_i \ge j'_{\sigma(i)}$.
    Now set $S = \{ i\in[k]: j_i > j'_{\sigma(i)}\}$.
    If $S=\emptyset$ then $M=M'$.
    Supposing that $S$ is non-empty, set $X=\mset{j_i\in M :i\in S}$ and $Y = \mset{j'_{\sigma(i)}\in M :i\in S}$.
    We then have that $X\subseteq M$, $X\neq \emptyset$ and $M' = (M\setminus X)\cup Y$ by construction.
    Further, by our assumption on $\sigma$, for any $x\in X$ there is a $y\in Y$ so that $x>y$.
    By Definition~\ref{def:Dershowitz-Manna-order}, we have $M \prec M'$ and $M \nsucc N'$.
    It follows then that when $M \succ M'$, $(R^{\oasterisk k})[M, M'] = 0$ and $R^{\oasterisk k}$ is upper triangular.
    Next, if we examine the diagonal of $R^{\oasterisk k}$,
    \[
        (R^{\oasterisk k})[\mset{j_1,\dots,j_k},\mset{j_1,\dots,j_k}] = \frac{w^2_{\mset{j_1,\dots,j_k}}}{k!} \sum_{\sigma\in S_k} \prod_{i=1}^k R[j_i, j_{\sigma(i)}] \ge \frac{w^2_{\mset{j_1,\dots,j_k}}}{k!} \prod_{i=1}^k R[j_i, j_i] > 0
    \]
    since $R$ is full-rank and a non-zero diagonal elements.
    Hence, so $R^{\oasterisk k}$ is full-rank too.
    Then $Q^{\oasterisk k} R^{\oasterisk k}$ is a valid QR-decomposition of $A^{\oasterisk k}$ and it follows that if $A$ has linearly independent columns then so does $A^{\oasterisk k}$. \qedhere
    \end{enumerate}
\end{proof}

\subsection{Group Action and Projection in the Fourier Basis}
\label{subsec:fourier-basis}
\begin{proof}{Proof of Proposition~\ref{lemma:almost-diagonalize-groups}.}
    First, $\mathcal F\in\Cx^{2d\times 2d}$ has entries $\mathcal F[j,k] = e^{-2\pi i jk/(2d)}$ and $\mathcal F^{-1}[j,k] = \frac1{2d}e^{2\pi i jk/(2d)}$.
    
    The elements of the cyclic group on $2d$ elements are generated by $\shiftel_1$ since $\shiftel_\ell = (\shiftel_1)^\ell$ so it is sufficient to show that $\widehat{R}_1 = \mathcal F \shiftel_1 \mathcal F^{-1}$:
    \begin{align*}
         (\mathcal F \shiftel_1 \mathcal F^{-1})[j,k] &= \mathcal F[j,:] (R_1F^{-1})[:,k] \\
        &= \sum_{\ell=0}^{2d-1} e^{-2\pi i j\ell/(2d)} \cdot \frac1{2d} e^{2\pi i (\ell-1) k/(2d)} \\
        &= e^{-2\pi i k/2d}\frac1{2d}\sum_{\ell=0}^{2d-1} e^{-2\pi i (j-k)\ell/(2d)} \\
        &= \begin{cases}
            e^{-2\pi i k/2d} & j\equiv k \pmod{2d} \\
            0 &\text{otherwise}.
        \end{cases}
    \end{align*}
    So $\mathcal F \shiftel_1 \mathcal F^{-1}  = \fshiftel_1$ as defined in \eqref{eq:fourier-shift-operation} and then
    $\mathcal F \shiftel_\ell \mathcal F^{-1} = \mathcal F (\shiftel_1)^\ell \mathcal F^{-1} = (\mathcal F \shiftel_1 \mathcal F^{-1})^\ell = (\fshiftel_1)^\ell = \fshiftel_\ell$.

    Next, notice that $\freflel = JR_{-1}$ since for $x\in\K^{2d}$ with entries $x=(x_0,x_1,\dots, x_{2d-1})$,
    \[
        \reflel\shiftel_{-1}x = JR_{-1}\begin{bmatrix}
            x_0 \\ x_1 \\\vdots \\ x_{2d-2}\\ x_{2d-1}
        \end{bmatrix} = J\begin{bmatrix}
            x_1 \\ x_2 \\\vdots \\ x_{2d-1} \\ x_0
        \end{bmatrix} = \begin{bmatrix}
            x_0 \\ x_{2d-1} \\\vdots \\ x_2 \\ x_1
        \end{bmatrix} = \freflel x.
    \]
    So
    \begin{align*}
        (\mathcal F JR_{-1} \mathcal F^{-1})[j,k] &= (\mathcal F \freflel \mathcal F^{-1})[j,k] \\
        &= (\mathcal F)[j,:] (\freflel \mathcal F^{-1})[:,k] \\
        &= \sum_{\ell=0}^{2d-1} e^{-2\pi i j\ell/(2d)} \cdot \frac1{2d} e^{2\pi i (-\ell) k/(2d)} \\
        &= \frac1{2d}\sum_{\ell=0}^{2d-1} e^{-2\pi i (j+k)\ell/(2d)} \\
        &= \begin{cases}
            1 & j\equiv -k \pmod{2d} \\
            0 &\text{otherwise}.
        \end{cases}
    \end{align*}
    This is $\freflel$ as given in \eqref{eq:fourier-reflect-operation} so $\mathcal F JR_{-1} \mathcal F^{-1} = \freflel$.
    
    Finally, we aim to show that there exists 
    a full rank square matrix $C$ such that
    \[
    D\projOp \mathcal C^{-1} = \fprojOp.
    \]
    Note that the matrix $\projOp$ has the same rowspan as the matrix $\bar{\projOp} := \projOp^\top \projOp \in \Rl^{2d\times 2d}$.
    Inspecting the structure of $\bar{\projOp}$ gives the simplification $\bar{\projOp} = I + \reflel$.
    Since this matrix is square we may find its conjugation by $\mathcal F$,
    \[
        \mathcal F \bar{\projOp}\mathcal{F}^{-1} = I + \mathcal F \reflel\mathcal{F}^{-1} = I + \mathcal F \freflel \shiftel_1\mathcal{F}^{-1} = I +  \freflel \mathcal F \shiftel_1\mathcal{F}^{-1} = I + \freflel \fshiftel_1.
    \]
    Elementwise, this gives $(\mathcal F \bar{\projOp}\mathcal{F}^{-1}v)[j] = v[j] + e^{-2\pi i (2d-j)/(2d)}v[2d-j]$ for $j=0,1,2,\dots, 2d-1$. With a small algebraic simplification, the rows of this matrix are
    \[
        \mathcal F \bar{\projOp}\mathcal{F}^{-1}[j,:] = \vec{e}_j + e^{2\pi i j/(2d)}\vec{e}_{-j},
    \]
    with indices of the standard basis vectors taken modulo $2d$.
    Notably, $\mathcal F \bar{\projOp}\mathcal{F}^{-1}[d,:] = 0$ which is why the Nyquist component is unrecoverable, and $\mathcal F \bar{\projOp}\mathcal{F}^{-1}[2d-j,:] = e^{-2\pi i j/(2d)} \mathcal F \bar{\projOp}\mathcal{F}^{-1}[2d-j,:]$ so we define $\fprojOp$ to be the first $d$ rows of $\mathcal F \bar{\projOp}\mathcal{F}^{-1}$ (those indexed 0 to $d-1$).
    To make explicit the change of basis we can write
    \[
    D := \mathcal F[[d), :] \Pi^\top \in \Cx^{d\times d}
    \]
    where it then follows that $D\projOp\mathcal{F}^{-1} = \fprojOp$. (Note $\mathcal F[[d), :]$ is the first $d$ rows of the matrix $\mathcal F$.)
    The original projection matrix $\projOp$ had full row-rank and so does $\fprojOp$ (for both, the rows are an orthogonal set of vectors) so $C$ is a change of basis.
\end{proof}

\begin{proof}{Proof of Proposition~\ref{prop:two-clean-projections-sufficient}.}
Let $\ell_1, \ell_2\in[2d)$ be known and suppose that $\ell_1-\ell_2$ does not divide $2d$.
We will consider the pMRA model in the Fourier basis.
Consider the block matrix given by
\[
    P_{\ell_1,\ell_2} = \begin{bmatrix}
        \fprojOp \fshiftel_{\ell_1} \\ \fprojOp \fshiftel_{\ell_2}
    \end{bmatrix} \in \Cx^{2d \times 2d}.
\]
For a signal's Fourier transform $\hat{x}$, $P_{\ell_1,\ell_2} \hat{x}$ produces the concatenation of the projections from directions $\ell_1$ and $\ell_2$.
Now the kernel of $ P_{\ell_1,\ell_2}$ contains $\Span\{\vec e_d\}$ as previously noted but if $P_{\ell_1,\ell_2}$ has rank $2d-1$ we may recover the rest of $\hat{x}$.
We will demonstrate this by showing that Gaussian elimination produces $2d-1$ pivots.
Consider that the rows of $P_{\ell_1,\ell_2}$ have the form
\[
{
\arraycolsep=1pt
\begin{array}{rrr}
     P_{\ell_1,\ell_2}[&j,& :]=e^{-2\pi i j\ell_1/(2d)}\vec{e}_j + e^{2\pi i j(\ell_1+1)/(2d)}\vec{e}_{-j} \\
     P_{\ell_1,\ell_2}[&j +d,& :]=e^{-2\pi i j\ell_2/(2d)}\vec{e}_j + e^{2\pi i j(\ell_2+1)/(2d)}\vec{e}_{-j}
\end{array}
}, \qquad \text{for }j\in[d),
\]
where indices on are taken modulo $2d$ as needed.
For $j=0$, $P_{\ell_1,\ell_2}[0,:] = P_{\ell_1,\ell_2}[d,:]$ so we get our missing pivot since the row operation $P_{\ell_1,\ell_2}[d,:] \gets P_{\ell_1,\ell_2}[d,:] - P_{\ell_1,\ell_2}[0,:]$ produces a zero row.
Next, if we apply the row operations
\[
    P_{\ell_1,\ell_2}[j+d,:] \gets P_{\ell_1,\ell_2}[j+d,:] - e^{-2\pi i (\ell_2-\ell_1)/(2d)}P_{\ell_1,\ell_2}[j,:]
\]
for each $j=1,2,\dots,d-1$ then each row has at most one non-zero entry of the form
\begin{align*}
    P_{\ell_1,\ell_2}[j+d,:] &= (e^{2\pi i j(\ell_2+1)/(2d)} - e^{-2\pi i (\ell_2-\ell_1)/(2d)}e^{2\pi i j(\ell_1+1)/(2d)})\vec{e}_{-j} \\
    &= (e^{2\pi i j(\ell_2+1)/(2d)} - e^{-2\pi i (\ell_2-2\ell_1-1)/(2d)})\vec{e}_{-j}.
\end{align*}
These two phases will only cancel if $-j(\ell_2+1) \equiv j(\ell_2 - 2\ell_1 - 1) \pmod{2d}$ for some $j=1,2,\dots,d-1$.
Simplifying this produces $j(\ell_1-\ell_2)\equiv 0 \pmod{2d}$.
Since $\ell_1-\ell_2$ does not divide $2d$ the modular condition is not satisfied for any row and so each has a non-zero entry in the $(d-j)$-th column.
Thus, the row-echelon form for $P_{\ell_1,\ell_2}$ contains pivots in every column except for the $d$-th one and we may recover all coefficients except the $d$-th.
\end{proof}

\subsection{Invariant Tensor Structure}
\label{subsec:invariant-tensor-structure}
\begin{proof}[Proof of Proposition~\ref{prop:invariant-tensor-structure}]
    First, for $x\in \Rl^d$ we will use the structure of the invariant tensor from the original MRA model since we may express the other tensors in terms of this one:
    \[
        T^{(r)}_\mathrm{MRA}(\hat x) = \frac1d \sum_{\ell=0}^{d-1} (\widehat R \hat{x})^{\oasterisk r}.
    \]
    We will use Definition~\ref{def:symmetrized-Kronecker-product} to expand entries in the tensor:
    \begin{align*}
        T^{(r)}_\mathrm{MRA}(\hat x)[\mset{j_1,\dots, j_r}] &= \frac1d\sum_{d=\ell}^{d-1} w_{\mset{j_1,\dots, j_r}} \prod_{k=1}^r (\fshiftel_\ell \hat{x})[j_k] \\
        &= w_{\mset{j_1,\dots, j_r}} \frac1d\sum_{d=\ell}^{d-1} \prod_{k=1}^r e^{-2\pi i \ell j_k/d}x[j_k]\\
        &=  w_{\mset{j_1,\dots, j_r}}\frac1d\sum_{d=\ell}^{d-1} \exp\left(-\frac{2\pi i \ell}d \sum_{k=1}^rj_k\right)\prod_{k=1}^r x[j_k] \\
        &=  w_{\mset{j_1,\dots, j_r}}\frac1d\sum_{d=\ell}^{d-1} \left(\exp\left(-\frac{2\pi i}d \sum_{k=1}^rj_k\right)\right)^\ell\prod_{k=1}^r x[j_k] \\
        &=   w_{\mset{j_1,\dots, j_r}}\left(\prod_{k=1}^r x[j_k]\right) \bb1 \left\{\sum_{k=1}^rj_k\equiv 0\pmod{d}\right\}.
    \end{align*}
    The last line follows since the sum is a geometric series with the common ratio a complex root of unity.
    As a result it only sums to zero when the base is one which only occurs under the modular condition.
    Next, we use the fact that
    \[
        T_\dMRA^{(r)}(\hat x) = \frac12 \left(T^{(r)}_\mathrm{MRA}(\hat x) + T^{(r)}_\mathrm{MRA}(\widehat J\hat x)\right),
    \]
    so
    \[
        T_\dMRA^{(r)}(\hat x)[\mset{j_1,\dots, j_r}] = \frac{w_{\mset{j_1,\dots, j_r}}}2\left(\prod_{k=1}^r x[j_k] + \prod_{k=1}^r x[-j_k]\right)\bb1 \left\{\sum_{k=1}^rj_k\equiv 0\pmod{d}\right\}.
    \]
    Finally, for the pMRA model with $x\in\Rl^{2d}$,
    \begin{align*}
        T_\pMRA^{(r)}(\hat x)[\mset{j_1,\dots, j_r}] &= w_{\mset{j_1,\dots, j_r}}\frac1{2d}\sum_{d=\ell}^{2d-1} \prod_{k=1}^r (\fprojOp \fshiftel_\ell \hat{x})[j_k] \\
        &= w_{\mset{j_1,\dots, j_r}}\frac1{2d}\sum_{d=\ell}^{2d-1} \prod_{k=1}^r (e^{-2\pi i (j_k \ell)/(2d)}x[j_k] + e^{2\pi i (j_k )/(2d)} e^{2\pi i (j_k \ell)/(2d)} x[-j_r]) \\
        &= w_{\mset{j_1,\dots, j_r}}\frac1{2d}\sum_{d=\ell}^{2d-1} \prod_{k=1}^r \sum_{s\in\{-1,1\}}e^{\pi i (1-s)(j_k )/(2d)} e^{-2\pi i (sj_k \ell)/(2d)}x[s j_k]. \qquad (*)
    \end{align*}
    We may exchange the product and sum to get a sum over vectors of signs, giving
    \begin{align*}
        (*) &= w_{\mset{j_1,\dots, j_r}}\frac1{2d}\sum_{d=\ell}^{2d-1}  \sum_{s\in\{-1,1\}^r} \prod_{k=1}^re^{\pi i (1-s_k)(j_k )/(2d)} e^{-2\pi i (s_k j_k \ell)/(2d)}x[s_k j_k] \\
        &= w_{\mset{j_1,\dots, j_r}}\frac1{2d}\sum_{d=\ell}^{2d-1}  \sum_{s\in\{-1,1\}^r}  \exp\left(\frac{\pi i}{2d}\sum_{k=1}^r (1-s_k)j_k \right) \exp\left(\frac{-2\pi i \ell}{2d}\sum_{k=1}^r s_k j_k\right) \prod_{k=1}^r  x[s_k j_k] \\
        &= w_{\mset{j_1,\dots, j_r}}\sum_{s\in\{-1,1\}^r}  \exp\left(\frac{\pi i}{2d}\sum_{k=1}^r (1-s_k)j_k \right) \frac1{2d}\sum_{d=\ell}^{2d-1} \left(\exp\left(\frac{-2\pi i}{2d}\sum_{k=1}^r s_k j_k \right)\right)^\ell \prod_{k=1}^r  x[s_k j_k] \\
        &= w_{\mset{j_1,\dots, j_r}}\sum_{s\in\{-1,1\}^r} \exp\left(\frac{\pi i}{2d}\sum_{k=1}^r j_k \right) (-1)^{\frac1{2d}\sum_{k=1}^r s_k j_k} \left(\prod_{k=1}^r  x[s_k j_k]\right)\bb1 \left\{\sum_{k=1}^rs_kj_k\equiv 0\pmod{d}\right\}.
    \end{align*}
    In the sum, all terms are invariant to negation of the sign vector $s \mapsto -s$ except for the product $\prod_{k=1}^r  x[s_k j_k]$ so if we reindex by a negation of the sign vector we simply get
    \[
    w_{\mset{j_1,\dots, j_r}}\sum_{s\in\{-1,1\}^r}  \exp\left(\frac{\pi i}{2d}\sum_{k=1}^r j_k \right) (-1)^{\frac1{2d}\sum_{k=1}^r s_k j_k} \left(\prod_{k=1}^r  x[-s_k j_k]\right) \bb1 \left\{\sum_{k=1}^r -s_kj_k\equiv 0\pmod{d}\right\}.
    \]
    Averaging these two equivalent equations gives
    \[
    \frac{w_{\mset{j_1,\dots, j_r}}}2\sum_{s\in\{-1,1\}^r} \exp\left(\frac{\pi i}{2d}\sum_{k=1}^r j_k \right) (-1)^{\frac1{2d}\sum_{k=1}^r s_k j_k} \left(\prod_{k=1}^r  x[s_k j_k] + \prod_{k=1}^r  x[-s_k j_k]\right)\bb1 \left\{\sum_{k=1}^rs_kj_k\equiv 0\pmod{d}\right\}.
    \]  
\end{proof}

\begin{proof}[Proof of Proposition~\ref{prop:power-spectrum-pMRA}]
    Let $x\in\Rl^{2d}$. By Proposition~\ref{prop:invariant-tensor-structure}, for $j_1,j_2\in[d)$,
\begin{equation}
\begin{split}
    T_\pMRA^{(2)}(\hat x)[\mset{j_1,j_2}] =  &\frac12(\hat x[j_1]\hat x[j_2] + \hat x[-j_1]\hat x[-j_2])\bb1\{j_1 + j_2 \equiv 0 \pmod{2d}\} \\
    & +\frac12 e^{\pi i (j_1+j_2)/(2d)} (-1)^{(j_1-j_2)/(2d)}(\hat x[j_1]\hat x[-j_2] + \hat x[-j_1]\hat x[j_2])\bb1\{j_1 - j_2 \equiv 0 \pmod{2d}\},
\end{split}
\end{equation}
with indexing into $\hat{x}$ taken modulo $2d$.
The conditions, $j_1 + j_2 \equiv 0 \pmod{2d}\}$ and $j_1 - j_2 \equiv 0 \pmod{2d}\}$ only hold simultaneously when $j_1=j_2=0$.
In this case, the entry is a multiple $\hat x[0]^2$ and hence one can only determine the value of $|\hat x[0]|$.
Beyond this, there are no other conditions under which the first indicator variable is non-zero.
The second indicator is satisfied along the diagonal, and since $x$ is a real signal $\hat x[-j] = \overline{\hat x[j]}$, so for $0 < j < d$,
\[
    T_\pMRA^{(2)}(\hat x)[\mset{j,j}] = e^{\pi i j/d} \hat x[j]\hat x[-j] = e^{\pi i j/d} |\hat x[j]|^2.
\]
Taking the magnitude of these entries recovers $|x[j]|^2$ for $0 < j < d$ and this is the power spectrum of $x$ excluding the coefficient $\hat x[d]$.
\end{proof}
\subsection{Polynomial Subspaces}
\label{subsec:polynomial-subspaces}
\begin{proof}[Proof of Proposition~\ref{prop:basis-for-cutout}]
    We proceed with three main steps.
    First, we show that $\Span(\mathcal P) \subseteq I_2(\mathcal X_1^\vee)$.
    Second, we show that $|\mathcal P| = \dim I_2(\mathcal X_1^\vee)$.
    Third, we show that $\mathcal P$ is a linearly independent set.
    Together these three ensure that $\Span(\mathcal P)= I_2(\mathcal X_1^\vee)$.
    
    First, let $p_{jk\ell m} \in \mathcal P$ and let $x\in \Cx^d$.
    If $z=(x^{\oasterisk 2})^{\oasterisk 2}$ then we will show $\langle p, z\rangle = 0$. 
    Observe that entries of $z$ are indexed by two-element multisets of two-element multisets so that
    \[
    z[\mset{\mset{j,k},\mset{\ell,m}}] = w_{\mset{\mset{j,k},\mset{\ell,m}}} x^{\oasterisk 2}[\mset{j,k}]x^{\oasterisk 2}[\mset{\ell,m}] = w_{\mset{\mset{j,k},\mset{\ell,m}}} w_{\mset{j,k}}  w_{\mset{\ell,m}} x[j]x[k]x[\ell]x[m].
    \]
    Then $\langle p_{jk\ell m}, z\rangle =  x[j]x[k]x[\ell]x[m] -  x[j]x[k]x[m]x[k] = 0$.

    Second, fix indices $0\leq j_1\leq k_1 < d$ and define the
    \[
        \mathrm{TR}(j_1,k_1) = \{(j_2,k_2) : 0\leq j_2 \leq k_2 < d, j_1 < j_2, k_1 < k_2\}
    \]
    to be the set of possible indices for $j_2$ and $k_2$ given our conditions on the four indices.
    This region has a trapezoid-shaped region with corners $(j_1+1, k_1+1)$, $(k_1+1,k_1+1)$, $(d-1,d-1)$ and $(j_1+1, k_1+1)$ which degenerates to  triangle when $j_1 = k_1$.
    The mutually exclusive conditions $j_2 \leq k_1$ or $j_2 > k_1$ split this into a rectangle of width $(d-1-k_1)$ and height $(k_1-j_1)$ and triangle with base $(d-1-k_1)$. Therefore 
    $|\mathrm{TR}(j_1,k_1)| = (d-1-k_1)(k_1-j_1) + \frac{(d-k_1)(d-1-k_1)}{2}$.
    The size of $\mathcal P$ is then
    \[
    |\mathcal P| = \sum_{k=0}^{d-1} \sum_{j=0}^k \left((d-1-k)(k-j) + \frac{(d-k)(d-1-k)}{2}\right) = \frac1{12}(d+1)d^2(d-1) = \dim I_2(\mathcal X_1^\vee).
    \]
    Third, by definition each vector $p\in \mathcal P$ has a support of size two.
    Vectors with non-overlapping support are orthogonal; it is not hard to see that when multisets $\mset{j_1,k_1,j_2,k_2} \neq  \mset{j'_1,k'_1,j'_2,k'_2}$ then $\langle p_{j_1k_1j_2k_2}, p_{j'_1k'_1j'_2k'_2}\rangle = 0$ since these vectors have disjoint support.
    We may then partition $\mathcal P$ into sets
    \[
        \mathcal P_{j_1k_1j_2k_2} = \{ p_{j'_1,k'_1,j'_2,k'_2} : \mset{j'_1,k'_1,j'_2,k'_2} = \mset{j_1,k_1,j_2,k_2}\}
    \]
    so that the spans of these $\mathcal P_{j_1k_1j_2k_2}$ give an orthogonal decomposition for $\Span(\mathcal P)$.
    It remains to show that each $\mathcal P_{\{j_1,k_1,j_2,k_2\}}$ contains a linearly independent set of vectors. 
    We simply note that if $j_1 \neq k_1$ and $j_2\neq k_2$ then $0\leq j_1 \leq j_2 < d$, $0\leq k_1 \leq k_2 < d$, $j_1<k_1$ and $j_2<k_2$ so
    \[
        \mathcal P_{j_1k_1j_2k_2} = \{p_{j_1k_1j_2k_2}, p_{j_1j_2k_1k_2}\},
    \]
    which are linearly independent vectors since they have different supports,
    otherwise $P_{j_1k_1j_2k_2} = \{p_{j_1k_1j_2k_2} \}$.
    Hence, each partition is spanned by linearly independent vectors and as a whole $\mathcal P$ is a basis.
\end{proof}

\begin{proof}[Proof of Proposition~\ref{prop:basis-for-orthogonal-subspace}]
    As previously noted, $\Span(v^{\oasterisk 2} : v\in\mathcal X_1^\vee)\cong (\Cx^d)^{\oasterisk 4}$ so $\mathcal Q$ is the correct size to be an orthogonal basis for $I(\mathcal X_1^\vee)_2^\perp$.
    
    First, we will confirm that $\mathcal Q$ is a basis by showing it is an orthogonal set of vectors.
    Notice that since they are standard basis vectors of $((\Cx^n)^{\oasterisk 2})^{\oasterisk 2}$,  $\langle (\vec{e}_a\oasterisk \vec{e}_b) \oasterisk (\vec{e}_c \oasterisk \vec{e}_d), (\vec{e}_{a'}\oasterisk \vec{e}_{b'}) \oasterisk (\vec{e}_{c'} \oasterisk \vec{e}_{d'})\rangle = 0$ whenever $\mset{\mset{a,b},\mset{c,d}} \neq \mset{\mset{a',b'},\mset{c',d'}}$.
    When multisets $\mset{a,b,c,d} \neq \mset{a',b',c',d'}$, each produces a different set of two element partitions so $\langle q_{\mset{a,b,c,d}}, q_{\mset{a',b',c',d'}}\rangle = 0$.
    
    Next, we will show that for all $p\in\mathcal P$ from Definition~\ref{def:basis-of-symmetric-minors} and $q\in\mathcal Q$, $p\perp q$.
    Consider that
    if $\mset{a,b,c,d} \neq \mset{a',b',c',d'}$ then $\langle p_{abcd}, q_{a'b'c'd'}\rangle = 0$.
    We need only verify that $\langle  p_{abcd}, q_{abcd}\rangle = 0$ and $\langle  p_{abcd}, q_{acbd}\rangle = 0$.
    For the first, expanding terms then canceling weights allows us to write the inner product as a sum of inner products of standard basis vectors in $((\Cx^n)^{\oasterisk 2})^{\oasterisk 2}$,
    \begin{align*}
        \langle  p_{abcd}, q_{a'b'c'd'}\rangle &= \langle \vec{e}_{\mset{\mset{a,b},\mset{c,d}}}, \vec{e}_{\mset{\mset{a,b},\mset{c,d}}} \rangle \\
        & \quad - \langle \vec{e}_{\mset{\mset{a,d},\mset{c,b}}}, \vec{e}_{\mset{\mset{a,b},\mset{c,d}}} \rangle \\ 
        &\quad + \langle \vec{e}_{\mset{\mset{a,b},\mset{c,d}}}, \vec{e}_{\mset{\mset{a,c},\mset{b,d}}} \rangle \\
        & \quad - \langle \vec{e}_{\mset{\mset{a,d},\mset{c,b}}}, \vec{e}_{\mset{\mset{a,c},\mset{b,d}}} \rangle \\ 
        &\quad + \langle \vec{e}_{\mset{\mset{a,b},\mset{c,d}}}, \vec{e}_{\mset{\mset{a,d},\mset{b,c}}} \rangle \\
        & \quad - \langle \vec{e}_{\mset{\mset{a,d},\mset{c,b}}}, \vec{e}_{\mset{\mset{a,d},\mset{b,c}}} \rangle \\ 
        &= 1 - 1 = 0.
    \end{align*}
    The computation for $\langle  p_{abcd}, q_{acbd}\rangle = 0$ is nearly identical.
    So the span of $\mathcal Q$ is $I(\mathcal X_1^\vee)_2^\perp$.
\end{proof}

\newpage

\section*{Acknowledgments}
\label{sec:acknowledgments}
\addcontentsline{toc}{section}{\nameref{sec:acknowledgments}}

We thank Matthias K\"oppe and Benjamin Lovitz for helpful discussions.

\phantomsection
\addcontentsline{toc}{section}{References}
\bibliography{references}
\bibliographystyle{alpha}
\end{document}